\documentclass[11pt]{article}
\usepackage[margin = 1in]{geometry}
\usepackage[utf8]{inputenc}

\usepackage{hyperref}
\hypersetup{colorlinks=true,citecolor=blue,linkcolor=blue,urlcolor=blue}
\usepackage{amsmath}
\usepackage{amsthm}
\usepackage{amssymb}   
\usepackage{enumerate}
\usepackage{thm-restate}
\usepackage{mathtools}
\usepackage{amsfonts}
\usepackage{dsfont}
\usepackage{amssymb} 
\usepackage{tikz}
\usepackage{bbold}

\theoremstyle{plain}
\newtheorem{theorem}{Theorem}
\newtheorem*{theorem*}{Theorem}
\newtheorem{corollary}[theorem]{Corollary}
\newtheorem{lemma}[theorem]{Lemma}
\newtheorem{proposition}[theorem]{Proposition}
\newtheorem{observation}[theorem]{Observation}
\newtheorem{definition}[theorem]{Definition}
\newtheorem{remark}[theorem]{Remark}

\usepackage[UKenglish]{babel}
\usepackage[UKenglish]{isodate}
\usepackage{url}
\usepackage{soul}
\usepackage[backgroundcolor=green!40]{todonotes}

\let\epsilon=\varepsilon

\newcommand{\G}{\mathcal{G}}

\newcommand{\B}{\mathcal{B}}
\newcommand{\A}{\mathcal{A}}
\newcommand{\emm}{\mathrm{e}}
\newcommand{\Ec}{\mathcal{E}}

\newcommand{\Nc}{\mathcal{N}}

\newcommand{\Eb}{\mathbf{E}}

\newcommand{\E}{\mathds{E}}
\newcommand{\V}{\mathcal{V}}
\renewcommand{\P}{\mathds{P}}

\newcommand{\TV}{\mathrm{TV}}

\newcommand{\Vin}{\V_\gamma}
\newcommand{\out}{\mathrm{out}}
\newcommand{\1}{\mathbb{1}}
\newcommand{\Kb}{\mathbb{K}}

\newcommand{\ord}{\mathrm{ord}}
\newcommand{\dis}{\mathrm{dis}}

\newcommand{\Bin}{\text{Bin}}

\newcommand{\Poisson}{\text{Poisson}}
\DeclareMathOperator{\CW}{CW}

\title{Low-temperature Sampling  on Sparse Random Graphs\thanks{For the purpose of Open Access, the authors have applied a CC BY public copyright licence to any Author Accepted Manuscript version arising from this submission. All data is provided in full in the results section of this paper.}} 
\author{Andreas Galanis \and Leslie Ann Goldberg \and Paulina Smolarova}
\date{23rd July 2026}
\begin{document}

\maketitle
\begin{abstract}
We consider sampling in the so-called low-temperature regime, which is typically characterised by non-local behaviour and strong global correlations. Canonical examples include sampling independent sets on bipartite graphs and sampling from the ferromagnetic $q$-state Potts model. Low-temperature sampling is computationally intractable for general graphs, but recent advances based on the polymer method have made significant progress for graph families that exhibit certain expansion properties that reinforce the correlations, including for example expanders, lattices and dense graphs.

One of the most natural graph classes that has so far escaped this algorithmic framework is the class of sparse Erd\H{o}s-R\'enyi random graphs whose expansion only manifests for sufficiently large subsets of vertices; small sets of vertices on the other hand have vanishing expansion which makes them behave independently from the bulk of the graph and therefore weakens the correlations. At a more technical level, the expansion of small sets is crucial for establishing the Kotecky-Priess condition which underpins the applicability of the  framework.

Our main contribution is to develop the polymer method in the low-temperature regime for sparse random graphs. As our running example, we use the Potts and random-cluster models on $G(n,d/n)$ for $d=\Theta(1)$, where we show a polynomial-time sampling algorithm for all sufficiently large $q$ and $d$, at all temperatures. Our approach applies more generally for models that are monotone. Key to our result is a simple polymer definition that blends easily with the connectivity properties of the graph and allows us to show that polymers have size at most $O(\log n)$.

\end{abstract}

\section{Introduction}

 We consider the problem of sampling from high-dimensional distributions in the so-called low-temperature regime, which is typically characterised by non-local behaviour. 
A classical example that has been studied in this context is the problem of sampling independent sets on bipartite graphs where it is perhaps intuitive to expect that a typical independent set is correlated with one of the two sides of the bipartition. Another standard example is the ferromagnetic Potts model  on (not necessarily proper) $q$-colourings weighted by the number of monochromatic edges (see below for definitions) where, at low temperatures, a typical colouring is expected to be  correlated with  one of the $q$ monochromatic colourings.

Starting from \cite{helmuth2019algorithmic, doi:10.1137/19M1286669}, a series of works have demonstrated that 
 this correlation can be utilised  to  obtain fast sampling algorithms for certain graph classes that exhibit strong expanding properties; perhaps the most prominent graph example is the class of  random regular graphs.
 Lattices   (such as $\mathbb{Z}^d$) can also be treated using contour models. By contrast, 
it is far from clear how to apply this intuition to general graphs. In fact,  sampling at low-temperatures is conjectured to be hard on general graphs, captured by so-called \#BIS-hardness in the case of the Potts and independent-set examples. 

Perhaps the most natural class of graphs that  has so far remained elusive from this algorithmic framework is the class of sparse Erd\H{o}s-R\'enyi random graphs like $\G(n,d/n)$ for $d=\Theta(1)$. This class has actually posed challenges even for high-temperature sampling \cite{Kuikui,efthymiou, DBLP:conf/approx/BlancaZ23, bezakova, efthymiou_colorings, yin_zhang, MosselSly} where the presence of high-degree variables typically complicates the applicability of standard sampling  techniques. For low-temperature sampling however the obstacles come primarily from another side, namely from the presence of many small ``sparse'' parts which have weak, if any, expansion.  For example, in $\G(n,d/n)$, this phenomenon not only causes the graph to be disconnected, but even inside the giant component there are \emph{induced} paths of size roughly $\Theta(\log n$) that therefore have almost no expansion; similarly, one can find various other sparse induced components which are scattered inside the giant.

Our main contribution is new polymer model for the random cluster model for sparse random graphs in the low-temperature regime, applicable even for graphs with non-uniform and unbounded degrees. As our running example, we use $\G(n,d/n)$  for $d=\Theta(1)$ large enough, where we show a polynomial-time approximate sampling algorithm for all sufficiently large $q$ at all temperatures. The key to our result is a simple polymer definition which blends easily with the connectivity properties of the graph and allows us to show that polymers have size at most $O(\log n)$. 
We also note that our approach more broadly applies to other monotone models.

To formally state our result, we recall a few standard definitions. The random cluster (RC) model with real parameters \(q,\beta > 0\) is a weighted edge model arising from statistical physics, assigning probabilities to the edge subsets of a given graph \(G = (V, E)\). For each edge subset \(F\subset E\), the associated weight of the \textit{configuration} \(F\) is given as
\[w_G(F)=w_{G;q,\beta}(F) = q^{c(F)} (\emm^\beta-1)^{|F|},\] where \(c(F)\) is the number of components in the graph \((V,F)\). We refer to edges in $F$ as the \emph{in-edges} of the configuration (and edges in $E\backslash F$ as the \emph{out-edges}).  Denote by \(\Omega_G\) the set of all configurations, i.e., all edge subsets of \(G\). The Gibbs distribution  \(\pi_G=\pi_{G;q,\beta}\)  is defined by 
$\pi_G(F) = {w_G(F)}/{Z_G}$ for $F\in \Omega_G,$ where $Z_G = \sum_{F'\in\Omega_G} w_G(F')$ is the partition function. The random cluster model can be viewed as an equivalent edge-representation of the  \textit{Potts model} (when \(q\) is an integer), which is supported on vertex assignments $\sigma:V\rightarrow \{1,\hdots,q\}$ weighted by $e^{\beta m(\sigma)}$ where $m(\sigma)$ is the number of monochromatic edges under $\sigma$.

For the random graph $\G(n,d/n)$ which is our focus here, the qualitative structure is conjectured to align with that on the random regular graph that has been extensively studied. In the so-called high temperature regime (low $\beta$)  a typical configuration is expected to be correlated with the all-out configuration $(F=\emptyset)$ and to consist of many small components -- this set of configurations is known as the ``disordered phase''; whereas in the so-called low temperature regime (high $\beta$) a typical configuration is expected to be correlated with the all-in configuration $(F=E)$ and to consist of a giant component (and other small components) -- this set of configurations is known as the ``ordered phase''. The interesting feature of the random-cluster model is that, on several classes of graphs, there is a critical threshold $\beta_c$ where the two phases coexist and each have probability bounded below by a constant. This coexistence typically causes severe complications to sampling in a window around $\beta_c$, see for example \cite{coja2023,mean1,gore1997swendsen} for various metastability phenomena and \cite{GJ,GSVY,CAI2016690} for computational hardness results that build on this.

Understanding this picture on $\G(n,d/n)$ more precisely is quite a bit more complicated than the random regular graph, even with currently available analysis tools; the key difference is that the local neighbourhood of a vertex is given by a Poisson tree (rather than a regular tree). For example,  the value of the log-partition partition function $\lim_{n\rightarrow \infty}\frac{1}{n}\log Z_G$ for low-temperatures is likely a rather involved expectation over Poisson trees based on understanding fixed-point equations on the underlying set of trees.\footnote{This is based on the fact that the model is ``replica symmetric'' \cite{replica}; to the best of our knowledge the formula for the log-partition function has not been established yet.} Likewise, understanding the performance of Markov chain algorithms on $\G(n,d/n)$ becomes extremely involved; in fact, even getting a fairly precise understanding of what happens  on a Poisson tree for low temperatures is open for both the Potts and random cluster models.

Our main result is to obtain an algorithm for $\G(n,d/n)$ by 
introducing suitable developments to  the polymer framework and showing how to use these together with recent Markov Chain Monte Carlo (MCMC) techniques. As a corollary of our techniques, we obtain various tools for the RC distribution on $\G(n,d/n)$ that shed light on the properties of the distribution.
\begin{theorem}\label{thm:mainthm}
    Let \(d\) be a large enough real. Then, for all sufficiently large reals \(q\), there is a (randomised) algorithm $\mathcal{A}$ such that the following holds whp over \(G\sim\G(n,d/n)\), for any inverse temperature $\beta>0$.
    On input $G$ and $\epsilon\in(0,1)$, with probability at least~$3/4$, the algorithm $\mathcal{A}$ outputs in $\mathrm{poly}(n,1/\epsilon)$ time a sample $F\in \Omega_G$ whose distribution is within TV-distance $\epsilon$ from the RC distribution $\pi_{G;q,\beta}$, and an estimate $\hat Z$ for the RC partition function $Z_G=Z_{G;q,\beta}$ that satisfies $\hat Z=(1\pm \epsilon)Z_G$. 
\end{theorem}
The success probability of~$3/4$ in Theorem~\ref{thm:mainthm} can be powered in the standard way.
As we will describe later in more detail (Section~\ref{sec:RCdynamics}), our algorithm is based on running the \textit{Glauber dynamics} for the random-cluster model from a suitable initial configuration, though it is a bit more involved than that since, in an interval of temperatures $[\beta_0,\beta_1]$, it needs to approximate the appropriate mixture of the ordered and disordered phases. 

The running time of the algorithm in Theorem~\ref{thm:mainthm} is $\text{poly}(n, 1/\epsilon)$ where the implicit constant in the exponent of scales as $O(\log q)$; this is largely because of a crude polynomial mixing-time bound on Poisson trees with wired boundary (i.e., Poisson trees where the leaves are conditioned to belong to the same component, say  by contracting them into a single vertex). It is an open question  to obtain sharper bounds with the wired boundary condition, even for regular trees when $q$ is non-integer (cf. \cite[Theorem 5]{SWtrees}).

The lower bound on $q$ in Theorem~\ref{thm:mainthm} is roughly  $\exp(\Omega(d \log d))$, which is essentially the ``standard'' bound where polymer-based techniques work (at least on bounded-degree graphs). Having $q$ this large facilitates showing closeness of a typical sample to the extreme configurations. In the case of bounded-degree graphs,  there has been some progress in lowering the value of $q$, mainly on random regular graphs 
\cite{blanca2021random,galanis2024plantingmcmcsamplingpottsarxiv}. For large~$\beta$ these results come from  first/second moment methods, though these results do not extend to non-integer~$q$ or to our~$G(n,d/n)$ setting.

We can improve significantly upon the running time using a different (deterministic) algorithm based on estimating the probability that an edge is an in-edge (based on a certain correlation decay property known as weak spatial mixing (WSM) within the phase, see Section~\ref{sec:proofThm2}) and integrating the expected number of in-edges. We expect that this different approach to utilising WSM will be useful in other settings where the underlying graph has tree-like neighbourhoods.
\begin{theorem}\label{thm:mainthm2}
    Let $R>0$ be arbitrarily large. For all sufficiently large reals $d$,   for all sufficiently large reals \(q\), there is an algorithm~$\mathcal{B}$ such that the following holds w.h.p. over \(G\sim\G(n,d/n)\), for any temperature $\beta>0$.
    On input $G$, the algorithm $\mathcal{B}$ outputs in $n^{1+\frac{1}{R}}$ time an estimate $\hat Z$ for the RC partition function $Z_G=Z_{G;q,\beta}$  that satisfies $\hat Z=(1\pm \tfrac{1}{n^{R}})Z_G$. 
\end{theorem}

It may help to explain the parameters in Theorem~\ref{thm:mainthm2}. The parameter~$q$
controls the accuracy of the algorithm via the decay rate in Lemma~\ref{lem:ordered-WSM}, and $d$ controls the running time via the length $r$ in Lemma~\ref{lem:ordered-WSM} (capturing the size of the polymer).  So, for any fixed ``target'' (captured by $R>0$), by first taking $d$~large (w.r.t.~$R$), and then taking $q$~large (w.r.t.~$d$), we can make both the running time and the accuracy sufficiently small in terms of the target (as per the theorem statement). See Remark~\ref{rem:d} for further comments about the restriction that~$d$ be sufficiently large, and how this arises (via Proposition~\ref{prop:no-large-ordered-polymer}) in the proof of Lemma~\ref{lem:ordered-WSM}.

\subsection{Discussion and Further Related Work}

Most of the known low-temperature algorithms are based on the so-called polymer method which was developed in \cite{helmuth2019algorithmic, doi:10.1137/19M1286669}. Intuitively, a polymer represents a local part of a configuration that deviates from a certain  extremal/ground state. The success of the method typically relies on the Koteck\'y-Preiss condition~\cite{KP} that roughly controls the number of polymers and the growth of their weights.  This condition guarantees that the log partition function can be approximated by truncating a relevant convergent series expansion (the cluster expansion) and then applying the interpolation techniques of~\cite{barvinokbook,PR}; see also \cite{fastpolymers} for MCMC variants. The 
method has been 
very successful  in the low-temperature regime~\cite{helmuth2019algorithmic, doi:10.1137/19M1286669,  
liao2019counting, Pottsexp, RCM-Helmuth2020, Unique, 10.1145/3470865,containers1, spectral, 10756117}. 

The main difficulty in applying the polymer method on ${\cal G}(n,d/n)$ is that small polymers, such as induced paths of length $\Theta(\log n)$ or other sparse induced subgraphs present in the graph, can have unusually large weight violating the growth rate required by Koteck\'y-Preiss type conditions. This difficulty causes undesirable restrictions; for example, \cite{GALANIS2022104894} worked on a  sparse random graph model with degrees $\geq 3$, a condition that ensures the absence of such non-expanding parts and precludes therefore ${\cal G}(n,d/n)$. As we will see in the following section, the main contribution of the present paper is to develop the polymer method for sparse random graphs such as $\mathcal{G}(n,d/n)$ by introducing a suitable polymer definition.

The algorithm provided in our proof of Theorem~\ref{thm:mainthm} 
is based on the   Glauber dynamics Markov chain for the random-cluster model. For the case of random regular graphs, the mixing properties of Glauber dynamics are  well-studied. Blanca and Gheissari~\cite{Blanca2}  showed that the random cluster Glauber dynamics mixes in time \(O(n\log n)\) for all \(q\geq 2\) and \(\beta < \beta_u\) where $\beta_u$ is 
the uniqueness threshold on the regular tree, and they showed that the same bound applies on $\G(n,d/n)$ as well \cite{rcm-on-unbounded-degree-graphs}; see also \cite{BlaGhe} for mixing-time results on other graph classes that apply for large $\beta$. However, for \(\beta \) near the  ordered/disordered threshold~$\beta_c$ (satisfying $\beta_c>\beta_u$),   the chain undergoes an exponential slowdown due to metastability phenomena and phase coexistence \cite{RCM-Helmuth2020,coja2023}. Our results suggest a similar exponential slowdown around criticality.

However, despite the worst-case mixing result, it is still possible to obtain a fast sampling algorithm based on Glauber dynamics using an appropriate initialisation from a ground state, see  \cite{SinclairsGheissari2022,RClattice,Galanis_Goldberg_Smolarova_2024}. Polymer techniques are helpful in this setting since they can be used to show a notion called weak spatial mixing within a phase introduced in \cite{SinclairsGheissari2022}, see Section~\ref{sec:proofThm2} for definitions. In the next section, we discuss more thoroughly how to adapt the polymer method for $\G(n,d/n)$ which is the main bottleneck of our results.

\subsection{Overview: old and new polymers}

As we have noted, polymers
capture small, independent deviations  from a certain extremal configuration called the ``ground state''. For example, for the random cluster model the natural ground states are  the all-in and all-out configurations, and for the ferromagnetic Potts model the ground states are the monochromatic colourings. 

To apply the polymer framework, the \textit{weight} of a polymer needs to be defined in a way that
enables 
us to 
relate the weight of a configuration to the product of weights of its polymers, and thus to relate the probability of a polymer appearing in a configuration to its weight. For models with only local interactions, such as the Potts model, the characterization of these small deviations, i.e., the definition of polymers, is clear from the model. For instance, for the Potts model on  low temperatures (high $\beta$), a typical configuration has the majority of vertices having the same colour. Thus the natural choice for polymers of a $q$-colouring $\sigma$ are maximal connected sets $\gamma$ which are not coloured with the majority colour. The weight  $w_\gamma$ measures how much  weight is lost by having bichromatic edges within the polymer $\gamma$.

However, when considering how a polymer should be defined for a configuration $F$ of the random cluster model, the answer is not as straightforward. For low-temperatures (high $\beta$), the natural ground state is the all-in configuration, thus a natural choice would be to consider (connected) sets of all-out edges. However the random cluster model has long-distance interactions, and in particular, the weight (and thus the probability) of a configuration $F$ depends not only on the number of out-edges (or in-edges), but also on the number of components. For a component of \((V,F)\), 
it may not be the case that the set of out-edges separating it from the rest of the graph is connected. If multiple polymers  are allowed to enclose a single component, it is not clear how to capture that component's contribution to the weight of the configuration using (multiplicative) polymer weights.

One insightful solution to this problem was given for the case of the random regular graph by Helmuth, Jenssen and Perkins \cite{RCM-Helmuth2020}, where they iteratively added more edges to the polymers, and then used the strong expansion properties of random regular graphs to ensure the desired multiplicative properties.  For completeness, we will define their polymers for a configuration $F$: 
Let \(\B_0 = E\setminus F\). For $k\geq 0$, define inductively \(\B_{k+1}\) to be the set of edges in \(\B_k\) along with all edges incident to vertices that have at least a \(\tfrac 59\)-fraction of their 
incident edges 
in \(\B_k\). Then the set of \textit{polymer edges} is \(\B_{\inf}(F)=\bigcup_{k\geq 0} \B_k\) and polymers are connected components of edges in \(\B_{\inf}\). Using the strong expansion of random regular graphs, they show that for any ordered configuration  \((V,F)\), there is one giant component containing \(>n/2\) vertices, and every other component has all its incident edges 
contained in a polymer. While the constant~\(\tfrac{5}{9}\) in the definition could be lowered to any \(\tfrac12 + \tau\), to avoid \(\B_{\inf}(F)= E\) it has to be greater than a half. 

The fact that an ordered configuration has a giant component also applies for $\G(n,d/n)$.  However, 
having the constant $5/9$ (or anything larger than $1/2$) causes problems in graphs where degrees can be small.  For instance,  a \(3\)-cycle of a \(3\)-regular graph has expansion $\tfrac{1}{2}$. Suppose that the edges of the $3$-cycle are in-edges, but the edges separating it from the rest of graph are out. 
Then these separating edges are not necessarily contained in a single polymer. This makes it difficult to capture the component's contribution, as explained earlier.
These problems persist even  if 
the  expansion of sufficiently large subgraphs 
is above~$1/2$ by any margin.

In the case of \(\G(n,d/n)\) the same problem is present in a more severe form:  while  large-enough connected sets can be shown to have expansion \(>\tfrac 12\) (for \(d\) large enough), there are small subgraphs with as many as $\Theta(\log n)$ vertices with extremely low expansion. These small subgraphs can  not be captured by these expansion-based polymers and it is therefore not clear how to extend the polymer definition of \cite{RCM-Helmuth2020} for the case of \(\G(n,d/n)\).

To start working towards a polymer definition, note that expanding properties of \(\G(n,d/n)\) that hold for sets of size $\Omega(\log n)$ can be used to show  that \((V,F)\) contains a giant component and all the other components are small (provided that $|F|$ is sufficiently large relative to $|E|$). So, a natural-looking solution for the polymer-definition problem is  to simply take the  polymers of a configuration $F$ to be components formed by the union of the set of out-edges $E\backslash F$ and the set of edges in $E$ incident to vertices in 
small components of $(V,F)$ (irrespectively of whether they belong to $F$). While this definition can be endowed with a polymer-weight definition that does satisfy the desired multiplicative properties, these polymers do not capture properly the deviations from the all-in  configuration. For instance, suppose there was a connected subgraph \(S\) of \((V,F)\) with exactly one in-edge $e$ in the cut \((S,\overline S)\). Then it would make sense that vertices of \(S\) belonged to a polymer since the removal of a single edge $e$ carves out $S$ as a component, affecting significantly  the marginal probability that an edge in the cut $(S,\overline S)$ is an in-edge.
Now, if the cut \((S,\overline S)\) is large, having $S$ as a  
component in a configuration (after removal of $e$) is a large deviation from the ordered state and is extremely unlikely to happen, and the marginal probabilities of edges in the cut should not be so sensitive to the status of a single edge. Hence, such ``almost-components'' must be enclosed in a polymer.

For our analysis specifically, we care about marginals of edges incident to a particular vertex $v$, thus we do not want the marginals of edges not in a polymer to be sensitive to the state of the edges incident to $v$.  Thus, as we show in Section~\ref{sec:ordered}, 
the natural choice of ordered polymers   is to take polymer edges to be the union of the set of out-edges of $F$ and the set of edges of $G$ incident to vertices in subgraphs that would be disconnected by removing this vertex $v$. 
In Section~\ref{sec:ordered} we show that with good probability the polymers have size \(O(\log n)\) and we use this to prove \textit{weak spatial mixing within the phase}. This is then used for bounding the mixing time from the Glauber dynamics (Theorem~\ref{thm:mainthm}) and converting to the counting algorithm (Theorem~\ref{thm:mainthm2}).

\section{Proof Outline of Theorems~\ref{thm:mainthm} and~\ref{thm:mainthm2}}\label{sec:proof-walkthrough} 

\subsection{The largest component of the graph}
 
It is well-known that for a graph $G$   with multiple components, the Gibbs distribution $\pi_G$ is a product distribution over the individual components, and that the partition function \(Z_G\) is the product over the partition function over individual components. 

For \(d > 1\), \(\G(n,d/n)\) w.h.p. contains one component of linear size while all remaining components have size \(O(\log n)\) and contain at most one cycle, see for example \cite[Section 5]{random-graphs-janson-book}.    We can therefore brute-force \(O(\log n)\)-sized components by enumerating all possible edge configurations in $\text{poly}(n)$ time; in fact, even  the mixing time from \textit{worst-case} initial configuration on such a component is going to be at most $\text{poly}(n)$ (see, e.g., \cite[Lemma 6.7]{blanca2021random}). 

 So, we only need to focus on the largest component of \(\G(n,d/n)\), which we denote by \(C = (V_C,E_C)\). Let \(n_C = |V_C|\) be the number of vertices. It is well-known (see e.g. \cite[Theorem 5.4]{random-graphs-janson-book}) that \(n_C\) is determined by the conjugate $\mu\in (0,1)$ of \(d\) which satisfies $\mu \emm^{-\mu}=d \emm^{-d}$. Using this, it is not hard to verify that for large $d$ we have w.h.p. \(n_C\geq(1-\emm^{-d/3})n\) and $|E_C|= \frac{dn}{2}\pm 2n \emm^{-d/3}$, see Appendix~\ref{sec:graph-proofs} for details.

For clarity of notation, we use \(\Omega_C\) to denote the set of configurations, \(w_C\) for the weights of configurations, \(\pi_C\) for the Gibbs distribution and \(Z_C\) for the partition function for the random cluster model on \(C\). For \(F\in\Omega_C\), we use \(c(F)\) to refer to the number of components in \((V_C,F)\).
 
\subsection{Phases of random cluster model on \(\G(n,d/n)\)}

Next, we formally introduce the notion of the \textit{ordered} and \textit{disordered} phases on \(C=C(\G(n,d/n))\), as well as some necessary notation.

\begin{definition}[Phases]\label{def:phases}
Let \(\eta := {1}/{1000}\).
The \emph{ordered phase} is \(\Omega^\ord_C := \{F\in\Omega_C \mid |F|\geq (1-\eta)|E_C|\}\). The  ordered distribution $\pi^\ord_C$  is defined for $F\in \Omega_C^\ord$ by $\pi^\ord_C(F) := {w_C(F)}/{Z^\ord_C}$, where \(Z^\ord_C := \sum_{F\in\Omega^\ord_C}w_C(F)\).
The \emph{disordered phase} is \(\Omega^\dis_C := \{F\in\Omega_C \mid |F| \leq \eta|E_C|\}\). The disordered distribution $\pi^\dis_C$ is defined for $F\in \Omega_C^\dis$ by $\pi^\dis_C(F) = {w_C(F)}/{Z^\dis_C}$, where \(Z^\dis_C := \sum_{F\in\Omega^\dis_C}w_C(F)\).
\end{definition}

We will use the following two   thresholds for~$\beta$, \(\beta_0\) and \(\beta_1\), defined from the following: 
\begin{equation}\label{eq:beta0beta1}
    \emm^{\beta_0} - 1 = q^{(2-1/10)/d} \mbox{ and $\emm^{\beta_1}-1 = q^{(2+1/10)/d}$}.
\end{equation}
Note that the constant \(1/10\) in the definitions above is somewhat arbitrary and could be decreased to any absolute constant \(\tau>0\). The next theorem tells us that a typical configuration on~\(\G(n,d/n)\) is either ordered or disordered and ``far away'' from phases' boundaries. It is proven in Appendix~\ref{sec:phase-transition}.

\begin{theorem}\label{thm:phase-transition}
Let $d$ be sufficiently large. Then, for all $q$ sufficiently large, the following holds w.h.p. over \(G\sim \G(n,d/n)\), where $F$ denotes a random configuration from the given distribution.
    \begin{enumerate}[(1)]
        \item For all \(\beta > 0\), \(\pi_C\Big(\tfrac{9\eta}{10}|E_C| \leq |F| \leq (1-\tfrac{9\eta}{10})|E_C|\Big) \leq \emm^{-n} \).
        \item For all \(\beta\geq \beta_0\), \(\pi_C^\ord\Big(|F| \leq (1-\tfrac{9\eta}{10})|E_C|\Big) \leq \emm^{-n} \).
        \item For all \(\beta\leq \beta_1\), \(\pi_C^\dis\Big(\tfrac{9\eta}{10}|E_C| \leq |F|\Big) \leq \emm^{-n}\).
        
        \item For \(\beta \geq \beta_1\), \(||\pi_C - \pi^\ord_C||_\TV \leq 2\emm^{-n}\).
        \item For \(\beta \leq \beta_0\), \(||\pi_C - \pi^\dis_C||_\TV \leq 2\emm^{-n}\).
    \end{enumerate}
\end{theorem}

\subsection{WSM and proof of Theorem~\ref{thm:mainthm2} via WSM}\label{sec:proofThm2}

The main challenge in order to obtain our results is showing WSM within the phase, which we define more formally here for the linear-sized component. Namely, for a vertex \(v\in V_C\) and an integer \(r\), let \(B_r(v)\) be the ball of radius \(r\) around \(v\), i.e. the induced subgraph consisting of all vertices distance \(\leq r\) from \(v\).
Let \(\pi_{B_r^-(v)}\) denote the Gibbs distribution on \(B\) with the \textit{free} boundary condition (i.e., conditional on all edges outside of \(B_r(v)\) being out) and \(\pi_{B_r^+(v)}\) denote the Gibbs distribution on \(B_r(v)\) with the \textit{wired} boundary condition (i.e.,  conditional on all  vertices at distance exactly \(r\) from \(v\) as being wired into a single component). For an edge   \(e\in E_C\), we use \(1_e\) denote the event that \(e\) is an in-edge.
\begin{definition}\label{def:wsm}
    Let $r>0$ be an integer and $K\in (0,1)$ be a real. 
    
    We say that \(C\) has \emph{WSM within the disordered phase} at distance \(r\) with rate $K$  if for every edge \(e\in E_C\) and  vertex \(v\in V_C\) incident to \(e\),  it holds that $\big|\pi_C^\dis(1_e) - \pi_{B^{-}_r(v)}(1_e)\big|\leq K^r$.
    
    Similarly, \(C\) has \emph{WSM within the ordered phase} at distance \(r\) with rate $K$ if for every edge \(e\in E_C\) and vertex \(v\in V_C\) incident to \(e\), it holds that $\big|\pi_C^\ord(1_e) - \pi_{B^+_r(v)}(1_e)\big|\leq K^r.$
\end{definition}

In Section~\ref{sec:ordered}, we show that for all \(d\) and \(q\) large enough, w.h.p. \(C\) has WSM within the ordered phase for all \(\beta\geq \beta_0\). The proof of WSM within the disordered phase for all \(\beta\leq\beta_1\) is deferred to the Appendix~\ref{sec:disordered}. 

\begin{lemma}\label{lem:ordered-WSM}
    There is a constant $A>0$ such that the following holds for all real $d$ sufficiently large. For all $q$ sufficiently large,  w.h.p. over \(G\sim\G(n,d/n)\), for all $\beta\geq \beta_0$, the largest component $C$ of $G$ has WSM within the ordered phase at distance $r=\lceil\frac{A}{d} \log n\rceil$ with rate $K=q^{-1/30}$.
\end{lemma}

Using this, we can prove Theorem~\ref{thm:mainthm2}. To outline briefly the main idea, consider arbitrarily large $R>0$ and $\beta^*\geq \beta_0$; then,  our goal is  to approximate $Z^{\ord}_{C;q,\beta^*}$ in time $n^{ O(1/R)}$ within relative error $n^{-R}$. We view $Z^{\ord}_{C;q,\beta}$ as a function of $\beta$ and set $g_C(\beta):=\frac{\partial 
 \log Z^{\ord}_{C;q,\beta}}{\partial \beta}$. Note in particular that
\begin{equation}\label{eq:gC}
\begin{aligned}
g_C(\beta)&=\frac{\frac{\partial 
Z^{\ord}_{C;q,\beta}}{\partial \beta}}{Z^{\ord}_{C;q,\beta}}=\frac{\emm^{\beta}}{\emm^{\beta}-1}\frac{\sum_{F\in \Omega^\ord_C} |F|(\emm^{\beta}-1)^{|F|}q^{c(F)}}{Z^{\ord}_{C;q,\beta}}
\\&=\frac{\emm^{\beta}}{\emm^{\beta}-1}\Eb_{F\sim \pi^{\ord}_C}\big[|F|\big]=\frac{\emm^{\beta}}{\emm^{\beta}-1}\sum_{e\in E_C} \pi^\ord_{C;q,\beta}(1_e).
\end{aligned}
\end{equation}
Moreover, by setting $\beta_\infty=n^{1/R}$ and integrating, we have that 
\begin{equation}\label{eq:interpolate123}\frac{Z^{\ord}_{C;q,\beta_\infty}}{Z^{\ord}_{C;q,\beta^*}}=\exp\bigg(\int^{\beta_\infty}_{\beta^*}g_C(\beta)\,\mathrm{d} \beta\bigg).
\end{equation}
We will show later that for $\beta=\beta_\infty$ it holds that $Z^{\ord}_{C;q,\beta_\infty}=(1\pm \emm^{-n^{\Omega(1/R)}}) q(\emm^\beta-1)^{|E_C|}$.
So, to obtain the desired approximation to $Z^{\ord}_{C;q,\beta^*}$, it suffices to approximate the integral $I:=\int^{\beta_\infty}_{\beta^*}g_C(\beta)\,\mathrm{d} \beta$. The standard way to approximate the integral would be to consider a sequence of $\beta$'s between $\beta^*$ and $\beta_\infty$ and, for each of them, approximate $g_C(\beta)$ using the WSM guarantee. Namely, by Lemma~\ref{lem:ordered-WSM}, for an edge $e\in E_C$ and a vertex $v$ incident to it, for the ball $B_e:=B_r(v)$ with $r=\lceil \tfrac{A}{d}\log n\rceil$, it holds that 
\begin{equation}\label{eq:WSMuse}
\big|\pi^\ord_{C;q,\beta}(1_e)-\pi_{B_e^+;q, \beta}(1_e)\big|\leq q^{-r/30}\leq \tfrac{1}{n^{R+3}},
\end{equation}
where the last inequality follows by taking $q$ large enough (with respect to $d$). Crucially,  $B_e$ is a 1-treelike\footnote{For a real $k$, a graph is called $k$-treelike if it becomes a tree after the removal of at most $k$ edges.} graph whose size $|V_G(B_e)|$ can be bounded by $d^r\log n=n^{O(1/R)}$; there  is a fairly standard recursive way to compute $\pi_{B_e^+;q,\beta}(1_e)$ exactly (cf. Lemma~\ref{lem:rationalfun} below).  However, in order to get the desired $ \tfrac{1}{n^R}$ relative error guarantee for the integral, we would need roughly $n^R$ many $\hat\beta$'s, which would lead to a huge running time.  

To obtain a faster algorithm, the key observation is that $\pi_{B_e^+;q,\beta}(1_e)$ is an explicit function of $\beta$ which can be represented using an explicit rational function that can be efficiently computed (using recursions) in time polynomial in $|V_G(B_e)|=n^{O(1/R)}$. Hence we can essentially perform the integration $\int^{\beta_\infty}_{\beta^*}\pi_{B_e^+;q,\beta}(1_e)$ symbolically. A technicality that arises is that we cannot get the exact antiderivative (since in general it will be in terms of algebraic numbers) but rather a numerical approximation to the antiderivative that is within additive $1/2^{n^{O(1/R)}}$ from its true value for all $\beta\in [\beta^*,\beta_\infty]$; all of that however requires only $n^{O(1/R)}$ bits of precision and hence can be carried out in $n^{O(1/R)}$ time.  

We next state more formally a few technical lemmas whose proofs are given in Appendix~\ref{sec:remproof2} and then show how to use them and conclude the proof of Theorem~\ref{thm:mainthm2}. We start with the computation of the marginals of edges in 1-treelike graphs.
\begin{restatable}{lemma}{rationalfun}\label{lem:rationalfun}
There is an algorithm that, on input an $n$-vertex  1-treelike graph $T$  rooted at $v$, an edge $e$ incident to $v$ and an integer $r\geq 1$, computes in time $2^{4r}n^{O(1)}$ rational functions $g^+(q,x), g^-(q,x)$ so that $\pi_{B^+_r(v);q,\beta}(1_e)=g^+(q,\emm^\beta-1)$ and $\pi_{B^-_r(v);q,\beta}(1_e)=g^-(q,\emm^{\beta}-1)$ for all $q,\beta>0$.
\end{restatable}
The next lemma captures the integration part of the argument. For an integer $t$, its size is the number of bits needed to represent it, and for a rational $w=\frac{w_1}{w_2}$ with integers $w_1,w_2$, its size is given by adding the sizes of $w_1$ and $w_2$. Moreover, for a polynomial $P(x)=\sum^n_{i=0}c_i x^{i}$ of degree $n$ and rational  coefficients, its size is given by   $n+\sum^n_{i=1}\mathrm{size}(c_i)$.
\begin{restatable}{lemma}{integrate}\label{lem:integrate}
There is an algorithm that, on input positive rationals $a,b, \epsilon$ each with size $\leq n$, and integer polynomials $P(x), Q(x)$ with non-negative coefficients and size $\leq n$, computes in time $\text{poly}(n)$ a number $\hat I$ so that $\hat I=I\pm \epsilon$ where $I=\displaystyle\int^b_{a}\frac{P(x)}{Q(x)}\mathrm{d} x$.
\end{restatable}
Finally, we will need the following estimate for the partition function for very large $\beta$.
\begin{restatable}{lemma}{betainf}\label{lem:betainf}
Let $d$ be large enough. For all $q$ sufficiently large,  for any arbitrarily small constant $\epsilon>0$, w.h.p. over $G\sim \mathcal{G}(n,d/n)$ it holds that $Z^{\ord}_{C;q,\beta}=\big(1\pm \emm^{-n^{\epsilon}}) q (\emm^{\beta}-1)^{|E_C|}$  for any $\beta\geq n^{2\epsilon}$.
\end{restatable}
\begin{proof}[Proof of Theorem~\ref{thm:mainthm2}]
We follow the argument outlined earlier in order to obtain a $\frac{1}{n^R}$-estimate $\hat Z^{\ord}$ for $Z^{\ord}_{C;q,\beta}$ when $\beta\geq \beta_0$. An analogous argument gives an approximation  $\hat Z^{\dis}$ for $Z^{\dis}_{C;q,\beta}$ for $\beta\leq \beta_1$ using WSM within the disordered phase (see Proposition~\ref{prop:wsm-disorder}). Then $\hat Z^\dis$ is the desired approximation for $\beta< \beta_0$,  $\hat Z^\dis+\hat Z^\ord$   for $\beta\in [\beta_0,\beta_1]$ and  $\hat Z^\ord$ for $\beta>\beta_1$.

So fix a target $\beta^*\geq \beta_0$ and focus on approximating $Z^{\ord}_{C;q,\beta^*}$. Set $\beta_\infty= n^{1/R}$. If $\beta^*\geq \beta_\infty$, from Lemma~\ref{lem:betainf} we can use $q(\emm^{\beta^*}-1)^{|E_C|}$ as our estimate $\hat Z^\ord$, so we can assume that $\beta^*<\beta_\infty$.  From~\eqref{eq:interpolate123} we have that $\frac{Z^{\ord}_{C;q,\beta_\infty}}{Z^{\ord}_{C;q,\beta^*}}=\exp\bigg(\int^{\beta_\infty}_{\beta^*}g_C(\beta)\,\mathrm{d} \beta\bigg)$ where $g_C(\beta)=\frac{\emm^{\beta}}{\emm^{\beta}-1}\sum_{e\in E_C} \pi^\ord_{C;q,\beta}(1_e)$, cf.~\eqref{eq:gC}.
From Lemma~\ref{lem:betainf} we can approximate $Z^{\ord}_{C;q,\beta_\infty}$ with relative error $\emm^{-n^{\Omega(1/R)}}$   so 
to obtain the desired approximation to $Z^{\ord}_{C;q,\beta}$, it suffices to approximate the integral $I:=\int^{\beta_\infty}_{\beta^*}g_C(\beta)\,\mathrm{d}\beta$ within additive error $\tfrac{1}{n^{R+1}}$.  
For an edge $e\in E_C$ and a vertex $v$ incident to it, let $B_e=B_r(v)$ with $r=\lceil \tfrac{A}{d}\log n\rceil$ as in \eqref{eq:WSMuse} where we saw that WSM gives $\big|\pi^\ord_{C;q,\beta}(1_e)-\pi_{B_e^+;q, \beta}(1_e)\big|\leq q^{-r/30}\leq \tfrac{1}{n^{R+3}}$ for $q$ large enough w.r.t. $d$. The $r$-neighbourhoods for $G\sim\mathcal{G}(n,d/n)$ are 1-treelike w.h.p. (see Lemma~\ref{lem:G-1-treelike}), so by Lemma~\ref{lem:rationalfun} we can compute in time $n^{O(1/R)}$ a rational function\footnote{Here we assume for convenience that $q$ is rational; if $q$ is irrational, the argument can be modified to use instead a rational approximation $\hat q$ to $q$ satisfying $|\hat q-q|\leq 1/2^{n^{\Theta(1/R)}}$.} $g_e(x)$ so that   $\pi_{B_e^+;q,\beta}(1_e)=g_e(\emm^{\beta}-1)$. Therefore for $x=\emm^{\beta}-1$ we have that
\begin{equation}\label{eq:g54545gg}
    \Big|g_C(\beta)-\frac{x+1}{x}\sum_{e\in E_C}g_e(x)\Big|\leq \frac{|E_C|}{n^{R+3}}, \mbox{ giving that } \Big|I-\sum_{e\in E_C}\int_{x^*}^{x_\infty} \frac{g_e(x)}{x} \mathrm{d} x\Big|\leq \frac{|E_C|(\beta_\infty-\beta^*)}{n^{R+3}}.
\end{equation}
Then, using the algorithm in Lemma~\ref{lem:integrate}, for each $e\in E_C$ we  estimate the integral $\int_{x^*}^{x_\infty} \frac{g_e(x)}{x} \mathrm{d} x$ within additive error $1/n^{R+3}$ in time $n^{O(1/R)}$ since each of $g_e(x)/x$ and $x_\infty$ have size $n^{O(1/R)}$. Combining the estimates gives an $\tfrac{|E_C|}{n^{R+3}}$-estimate of the sum of integrals in \eqref{eq:g54545gg}, yielding a $\tfrac{1}{n^{R+1}}$-estimate of the integral $I$. All in all, this gives the desired $\tfrac{1}{n^R}$-estimate of $Z^{\ord}_{C;q,\beta^*}$ in time $n^{1+O(1/R)}$. Since $R$ can be arbitrarily large, the theorem follows.
\end{proof}

\subsection{Graph properties}\label{sec:graph-properties}

Before proceeding to the WSM proofs, we  use the following expansion properties of \(\G(n,d/n)\) (and its largest component). While the random graph lacks 
regularity and expansion from small sets, for connected subgraphs of size \(\Omega(\log n)\) we can lower bound the average degree and thus obtain a weak expansion property. Since we have two graphs, \(G\) and its largest component \(C\), for the clarity of notation, we use \(V_G\) for the vertex set of \(G\), and \(E_G\) for its edge set.
The following propositions are proved in Sections~\ref{sec:avg-deg} and ~\ref{proof:propexpansion}.

\begin{proposition}\label{prop:average-degree-of-connected-sets}
    For any \(\epsilon\in(0,1)\), there exists \(A = A(\epsilon) > 0\) such that for all \(d\) large enough, w.h.p. over \(G\sim\G(n,d/n)\), any \emph{connected} vertex set \(S_G\subseteq V_G\) with size at least \({A\log (n)}/{d}\) has average degree at least \((1-\epsilon)d\).
\end{proposition}

In the following proposition, and throughout the paper, we use $\deg_G(S)$ to denote the \textit{total degree} of a set $S\subseteq V_G$, that is, $\deg_G(S)$  is
the sum of the degrees (in $G$) of vertices in $S$. For a subset $T \subseteq V_G$ we write \(E_G(S,T)\) to denote the set of edges with one endpoint in \(S\) and the other in \(T\). We also use \(e_G(S) := |E_G(S, S)|\), i.e. the internal edges in \(S\). 

\begin{restatable}{proposition}{weakexpansion}\label{prop:expansion-of-core}
    There exists a constant \(A > 0\), such that for all \(d\) large enough, w.h.p. over \(G\sim\G(n,d/n)\) the following holds.
Let \(S\subseteq V_C\) be a \emph{connected} vertex set in \(C\) with \(|S|\geq {A\log (n)}{/d}\) and \(\deg_G(S)\leq \tfrac{1}{10}\deg_G(V_C)\). Then 
\(E_G(S,\overline S) \geq \tfrac{3}{5} \deg_G(S)\).
The same holds 
if $S \subseteq V_C$ is a connected set in~$C$ with   \({A\log (n)}/{d}\leq |S|\leq n /6\).
\end{restatable}

We remark that the choice of
$3/5$ in Proposition~\ref{prop:expansion-of-core}  is arbitrary, and any constant between~$1/2$ and~$1$ could be used. Also, 
note that the lower bound $A \log(n)/d$  in both propositions 
is 
asymptotically  optimal in~$n$ and~$d$, as it can be shown that \(G\) contains an induced path of length \(\approx \frac{\log n}{d}\).

Finally, we will use the property that after removing a constant fraction of edges, \(C\) still contains a large component. This property is well-known for expanders \cite{trevisan2016expanders}; to prove it for \(C\) we use the fact that its \textit{kernel} (the graph obtained from \(C\) by contracting every induced path into a single edge and removing attached trees) is an expander. The constant \({1}/{144}\) in the lemma is somewhat arbitrary and follows from the expansion constant for the kernel (see Appendix~\ref{sec:kernel-expansion}).

\begin{restatable}{lemma}{giantafterremoval}\label{lem:giant-component-in-the-core-after-edge-removal}\label{lem:short}
    For all \(d\) large enough, w.h.p. over \(G\sim\G(n,d/n)\) the following holds.
    Suppose that \(\theta \in (0,\tfrac 12)\) is a real and \(E'\subseteq E_C\) is an edge subset such that \(|E'|\leq \frac{|E_C|}{144}\left(2\theta-\frac{1}{100}\right)\).
    Then \((V_C,E_C\setminus E')\) has 
a component with total degree \(\geq \deg_G(V_C)(1-\theta)\), and all other components of \((V_C,E_C\setminus E')\) have total degree at most \(\theta \deg_G(V_C)\). 
\end{restatable}

In light of Lemma~\ref{lem:giant-component-in-the-core-after-edge-removal}, we will order components of any subgraph $C'$ of~$C$ by their total degree. So we use the phrase ``largest component'' to refer to the component of $C'$ with largest total degree. A ``giant'' component  is a component whose total degree is more than half of the total degree of $C$, and ``small'' otherwise. 
Thus, by applying Lemma~\ref{lem:giant-component-in-the-core-after-edge-removal} for $\theta=1/10$, we have that for any $E'\in \Omega^{\ord}_C$  the subgraph $(V_C,E')$ contains a giant component and all the other components are small.

\section{Weak spatial mixing within the ordered phase}\label{sec:ordered}

In this section, we prove  WSM within the ordered phase. We bound the difference of the marginals on \(e\) by the probability that, roughly, there is a long path avoiding the giant component of \((V_C,F)\). We then use the new definition of \textit{ordered polymers} that naturally captures this event.

Formally for a vertex \(v\) and an integer \(r\geq 1\), define \(\A_{v,r}\) to be the event, that for a random configuration~\(F\), all simple length-\(r\) paths from \(v\) in \(G_C\) intersect the largest component of the graph $(V_C,F)\backslash v$, i.e., the graph $(V_C,F)$ with $v$ removed. Define \(\A_{v,r}'\) to be the event that $|F \setminus E_G(B_r(v))| \geq (1-\eta)|E_C|$. Then, we have that
\begin{equation}\label{eq:WSMbound}
    \big|\pi_{B^+_r(v)}(1_e) - \pi^\ord_C(1_e)\big|\leq 2\pi^\ord_C(\neg\A_{v,r}) + 2\pi^\ord_C(\neg\A_{v,r}').
\end{equation}
This follows by conditioning on the boundary configuration outside of $B_r(v)$. Namely, using the monotonicity of the RC model, this boundary condition is dominated above by the wired boundary condition on $B_r(v)$, where all vertices at the boundary of $B_r(v)$ are conditioned to belong to the same component, and hence $\pi^\ord_C(1_e)\leq \pi_{B_r^+(v)}(1_e)$. Likewise, the event $\mathcal{A}_{v,r}'$ gives from Lemma~\ref{lem:giant-component-in-the-core-after-edge-removal} that there is a unique giant component $C_{F}$ in the graph $(V_C\backslash B_r(v), F\backslash E_C(B_r(v)))$ and $A_{v,r}$ that there is a  set $S$ of vertices inside $B_r(v)$ whose removal disconnects 
$v$ from the vertices outside of~$B_r(v)$ and every $w\in S$ belongs to $C_F$. Therefore, the event  $\mathcal{A}_{v,r}\cap \mathcal{A}_{v,r}'$ implies that there there is a  set $S$ of vertices inside $B_r(v)$ whose removal disconnects 
$v$ from the vertices outside of~$B_r(v)$ and every $w\in S$ belongs to $C_F$, i.e.,   $S$  forms a wired boundary condition closer to $v$. From monotonicity, it follows that $\pi_{B_r^+(v)}(1_e)\leq \pi^\ord_C(1_e\mid \mathcal{A}_{v,r}\cap \mathcal{A}_{v,r}')$. Combining these estimates on $\pi_{B_r^+(v)}(1_e)$ yields \eqref{eq:WSMbound}.  (A similar argument can be found in \cite[Lemma 5.7]{galanis2024plantingmcmcsamplingpottsarxiv}.)

It is relatively easy to bound the probability of $\neg \mathcal{A}_{v,r}'$ under $\pi^{\ord}_C$. For large enough $d$, standard estimates for the neighbourhoods of $\G(n,d/n)$ (see also the upcoming Lemma~\ref{lem:ball-size}) guarantee that w.h.p. over \(G\sim\G(n,d/n)\), for any \(r\leq \tfrac{A}{d}\log n\), it holds that \(|E_G(B_r(v))|\leq \sqrt n\log n\) which is at most \(\tfrac{\eta}{10}|E_C|\) for $n$ large enough since $|E_C|=\Omega(n)$ w.h.p.. Hence, by Theorem~\ref{thm:phase-transition}(ii), $\pi^{\ord}_C(\neg \mathcal{A}_{v,r}')\leq \emm^{-n}$.

It remains to bound the probability of \(\A_{v,r}\), for which we use \textit{ordered polymers}.

\subsection{Ordered Polymers}

We want to bound for each vertex \(v\), and for each \(r \sim {\log (n)}/{d}\), the quantity \(\pi^\ord_C(\A_{v,r})\).
To do so, we need to control the size of ``polymers'' with respect to a configuration $F$ and a vertex $v\in V_C$.

We first consider non-giant components in $(V_C,F)\backslash v$; roughly, we will refer to the union of their vertices as the polymer vertices. To form polymers, we will take connected components of these vertices. 
The details are as follows.


 
\begin{definition}[Ordered polymers]\label{def:polymers}
Fix \(v\in V_C\). Consider \(F\in\Omega^\ord_C\). For any \(w\in V_C\), define the   graph $H_v(F,w)$ as follows.
    \[
    H_v(F, w) := \begin{cases}
        \text{the component of }w\text{ in }(V_C,F), & \text{if }w=v \\
        \text{the component of }w\text{ in }(V_C,F)\setminus v, & \text{if }w\neq v.
    \end{cases}\]
Let
    $
    \V(F,v) := \{w \in V_C \mid \deg_G(H_v(F, w)) \leq \tfrac{1}{2}\deg_G(V_C)\}$.
The \textit{ordered polymers} of~$F$ are subgraphs of~$C$. 
We define two types of \emph{ordered polymers}, non-singleton polymers and singleton polymers.
    \begin{itemize}
        \item  
For each maximal set  \(S\subseteq\V(F,v)\)
that is connected in \(C\),
there is a \emph{non-singleton polymer}.
The edge set of this polymer consists of all edges in~$E_C$ that are incident to vertices in~$S$.
The vertex set of this polymer is the set of endpoints of these edges.
    
\item For each edge $e\in E_C\setminus F$ that is not an edge of a non-singleton polymer of~$F$, there is a
      \emph{singleton polymers} of \(F\) 
  consisting of the single edge~$e$.    
      
\end{itemize}
    Let \(\Gamma^\ord_v(F)\) denote the set of all ordered polymers of \(F\) (with respect to $v$).

\end{definition}


As we show in Lemma~\ref{lem:ordered-polymers-are-small}, a giant exists in \((V_C,F)\) for all \(F\in\Omega_C^\ord\).
 
\begin{definition}\label{def:45f43w}
    Fix \(v\in V_C\).  Fix \(F\in\Omega^\ord_C\) and \(\gamma = (U, B)\in\Gamma^\ord_v(F)\).
    The \emph{inner vertices} of \(\gamma\) are \(\Vin := U\cap \V(F,v)\). The \emph{out-edges} of \(\gamma\) are the edges in \(E_\out(\gamma) := (E_C\setminus F)\cap B\) and \(e_\out(\gamma) = |E_\out(\gamma)|\) is the number of out-edges of~$\gamma$.
    The weight of \(\gamma\) is \(w^\ord_C(\gamma) := q^{c'(\gamma)} (\emm^\beta - 1)^{-e_\out(\gamma)}\), where \(c'(\gamma)\) is the number of  components of \((V_C, E_C\setminus E_\out(\gamma)))\) with total degree at most $(1/2) \deg_G(V_C)$.
\end{definition}

Lemma~\ref{lem:path-large-polymer}
relates the definition of the polymers to the event \(\A_{v,r}\).
Later, in Lemma~\ref{lem:relate}, we will also relate
the weight of a configuration to the product of weights of its polymers.

\begin{lemma}\label{lem:path-large-polymer}
Fix \(F\in\Omega^\ord_C\).
Let $r$ be a positive integer.  Suppose that there is a simple path \(v_0,v_1,\dots,v_r\) in \(C\) such that \emph{none} of \(v_1,\dots, v_r\) belongs to the largest component of \((V_C,F)\setminus v_0\). Then there is a polymer \(\gamma\in\Gamma_{v_0}^\ord(F)\) with 
\(\{v_1,\dots,v_r\}\subseteq\Vin\),
hence $|\Vin| \geq r$.
\end{lemma}
\begin{proof}
Consider $i\in [r]$.
Since $v_i \neq v_0$, 
$H_{v_0}(F,v_i)$ is the component of~$v_i$ in $(V_C, F) \setminus v_0$.
This is not the largest component in $(V_C,F)\setminus v_0$, so
$\deg_G(H_{v_0}(F,v_i)) \leq \tfrac12 \deg_G(V_C)$ by Lemma~\ref{lem:short}.
Now note that the $v_1,\ldots,v_r$ form a path in~$C$ so they are contained in a single set $S$ in Definition~\ref{def:polymers}. Thus,
    there is \(\gamma\in\Gamma_{v_0}^\ord(F)\) with \(\{v_1,\dots,v_r\}\subseteq\Vin\).
\end{proof}

In order to prove Lemma~\ref{lem:relate}, and to
bound the probability that there is a large polymer, we next  use Lemma~\ref{lem:giant-component-in-the-core-after-edge-removal} to get an upper bound on a size of an ordered polymer.  

\begin{lemma}\label{lem:ordered-polymers-are-small}
    For all \(d\) large enough, w.h.p. over \(G\sim \G(n,d/n)\), the following holds. Fix any \(v\in V_C\). For any \(F\in\Omega^\ord_C\), \((V_C,F)\setminus v\) contains a component of degree \(\geq (1-\tfrac{1}{10})\deg_G(V_C)\).
    Hence in particular, every polymer of \(\gamma\in\Gamma^\ord_v(F)\) has \(\deg_G(\Vin) \leq \tfrac{1}{10}\deg_G(V_C)\).
\end{lemma}
\begin{proof}
    Let $\theta=1/10$. w.h.p. the maximum degree of \(G\) is \(O(\frac{\log n}{\log\log n})\), see e.g. \cite[Theorem 3.4]{Frieze_Karoński_2015}. Since $|F|\geq (1-\eta)|E_C|$ and $\eta=1/1000$, it follows that the graph $(V_C,F)\backslash v$  contains at most \(\frac{|E_C|}{144}(2\theta  - \tfrac{1}{100})\) edges, so by Lemma~\ref{lem:giant-component-in-the-core-after-edge-removal} it has a component \(\kappa\) of total degree \(\geq (1-\tfrac{1}{10})\deg_G(V_C)\).
    By Definition~\ref{def:polymers},  none of the vertices in $\kappa$ is in \(\V(F,v)\). Thus \(\deg_G(\V(F,v))\leq \tfrac{1}{10}\deg_G(V_C)\), and so in particular any component of \(\V(F,v)\) has total degree at most \(\tfrac{1}{10}\deg_G(\V(F,v))\).

    Finally, we note that for any singleton polymer \(\gamma\in\Gamma^\ord_v(F)\), from  Definition~\ref{def:45f43w}, \(\Vin = \emptyset\).
\end{proof}

We use Lemma~\ref{lem:ordered-polymers-are-small} to relate the weight of a configuration to the product of weights of its polymers.

\begin{lemma}\label{lem:relate}
    For all \(d\) large enough, w.h.p. over \(G\sim\G(n,d/n)\) the following holds. For any \(F\in\Omega^\ord_C\) and \(v\in V_C\)
    \[
    w_C(F)=q(e^\beta-1)^{|E_C|}\prod_{\gamma\in\Gamma^\ord_v(F)} w_C^\ord(\gamma).
    \]
\end{lemma}
\begin{proof}
    Write
    \[
    q(e^\beta-1)^{|E_C|}\prod_{\gamma\in\Gamma^\ord_v(F)} w_C^\ord(\gamma) = q^{1 + \sum_{\gamma\in\Gamma^\ord_v(F)} c'(\gamma)} (e^\beta-1)^{|E_C| - \sum_{\gamma\in\Gamma^\ord_v(F)}  e_\out(\gamma)}.
    \]
    Every out-edge of \(F\) is in exactly one ordered polymer of \(F\), thus \( |E_C\setminus F| = \sum_{\gamma\in\Gamma^\ord_v(F)} e_\out(\gamma)\). It remains to show that  \(c(F) = 1 + \sum_{\gamma\in\Gamma^\ord_v(F)} c'(\gamma)\).
    
    By Lemma~\ref{lem:ordered-polymers-are-small} w.h.p. \(G\) is such that for all \(F\in\Omega_\ord\), \((V_C,F)\) contains a giant component with total degree \(>\frac{1}{2}\deg_C(G_C)\). The \(1\)-term in the above equality corresponds to the giant. Now consider any non-giant component \(\kappa\) of \((V_C,F)\). Its vertices must be in \(\V(v,F)\), and thus in \(\Vin\) for some \(\gamma\in\Gamma^\ord_v(F)\), hence also \(E_C(V(\kappa), \overline{V(\kappa)})\subseteq E_\out(\gamma)\). Since \(\kappa\) is connected in \((V_C,F)\), it must be connected in \((V_C,E_C\setminus E_\out(\gamma))\), as \(F\subseteq E_C\setminus E_\out(\gamma)\), thus it is accounted for in \(c'(\gamma)\). Therefore \(1 + c(F) \leq 1 + \sum_{\gamma_\in\Gamma_v^\ord(F)} c'(\gamma)\).

    Now consider \(\kappa\) to be a non-giant component of \((V_C,E_C\setminus E_\out(\gamma))\) for some \(\gamma\in\Gamma_v^\ord(F)\). Since \((V_C,E_C)\) is connected, \(E_\out(\gamma)\) contains an out-edge \(e\) incident to some vertex \(w\) of \(\kappa\). Since \(F\subseteq E_C\setminus E_\out(\gamma)\), the component of \(w\) in \((V_C,F)\), \(\kappa'\), must be a subgraph of \(\kappa\). Clearly, \(\kappa'\) must also be non-giant, thus all of its vertices are in \(\V(v,F)\). Since non-singleton polymers are maximal connected subgraphs consisting of vertices in \(\V(v, F)\), there is a unique polymer \(\gamma'\) containing vertices 
in $V(\kappa')\subseteq V(\kappa)$.
Since \(e\) is an out-edge incident to \(w\), \(e\in E_\out(\gamma')\). Thus \(e\) is not in a singleton polymer in \(\Gamma_{v}^\ord(F)\), so \(\gamma\) must be a non-singleton polymer with an endpoint of \(e\) contained in \(\Vin\). But then \(\Vin\cup\{w\}\) is a connected subset in \((V_C,E_C)\), and thus in particular so 
is
\(\Vin\cup \V_{\gamma'}\). Since polymers of \(F\) do not intersect, it must be the case \(\gamma = \gamma'\) and \(\kappa = \kappa'\). Thus  each non-giant component \(\kappa\) of \((V_C,E_C\setminus E_\out(\gamma))\) for any \(\gamma\in\Gamma_v^\ord(F)\) is also a non-giant component of \((V_C,E_C)\) and for any \(\gamma'\in\Gamma_v^\ord(F)\), \(\gamma'\neq \gamma\), \(\kappa\) is not a non-giant component of \((V_C,E_C\setminus E_\out(\gamma'))\). This finished the second part of the equality and thus the proof.
\end{proof}

It remains to   bound the probability that there is a large polymer. To do so, we show that the weights of (sufficiently large) polymers  decay, and then relate the weights to the probabilities that  polymers occur.

Observe that for any \(F\in\Omega^\ord_C\), any \(v\in V_C\) and any non-singleton polymer \(\gamma\in\Gamma^\ord_v(F)\), all edges in \(E_G(\Vin,\overline\Vin)\)  that are not incident to~\(v\) are out-edges.
To see this, note that, if there is an in-edge from \(w\neq v\in\overline \Vin\) to \(u\in\Vin\), then the vertices~$u$ and~$w$ would be in the same component of \(H_v(F, w)\), but then by the definition of non-singleton polymers, 
$w$ would be in $\Vin$. Since, for any singleton polymer \(\gamma\), \(\Vin = \emptyset\), the same statement applies. Formally we have the following.

\begin{observation}\label{obs:name}
    Let \(F\in\Omega^\ord_C\), \(v\in V_C\) and \(\gamma\in\Gamma^\ord_v(F)\). Then \(E_\out(\gamma)\supseteq E_G(\Vin,\overline\Vin) \setminus E_G(v,\Vin)\), and hence \(e_\out(\gamma) \geq |E_G(\Vin, \overline{\Vin})| - |E_G(v,\Vin)|\).
\end{observation}

We next use that w.h.p. \(G\) is 1-treelike  (see Lemma~\ref{lem:G-1-treelike}), i.e. that for all \(d\) large enough, w.h.p. any radius-\(\tfrac{\log n}{d}\) ball in \(G\) (and thus in \(C\)) can be made into a tree by removing at most a single edge, to show that \(E_G(v,\Vin)\) accounts for a small fraction of edges going from \(\Vin\).
\begin{lemma}\label{lem:not-too-many-edges-into-a-single-vertex-relatively}\label{lem:X}
    For all real $d$ large enough, w.h.p. over \(G\sim\G(n,d/n)\) the following holds. For any \(v\in V_G\) and for any connected subset \(S\subseteq V_G\) that doesn't contain~$v$, $
    |E_G(v,S)| \leq 2 + \frac{d|S|}{\log n}.$
\end{lemma}
\begin{proof}
Let $r$ denote $\tfrac{\log n}{d}$.
By Lemma~\ref{lem:G-1-treelike}, for all $d$ sufficiently large, w.h.p. any radius-$r$ ball in \(G\) can be made into a tree by removing at most a single edge.     
    Let \(v\in V_G\) be any vertex.
    If \(|E_G(v,S)|\leq 2\), we are done. So assume otherwise (which obviously implies that there is a neighbour of~$v$ in~$S$).  

     Consider the radius-$r$  ball around \(v\) in the induced graph on \(S\cup\{v\}\), and denote it \(B\). Since \(B\) is a subgraph of a radius-$r$ ball in \(G\),  there is an edge set \(\Ec\subseteq E_G(B)\) with \(|\Ec|\leq 1\), such that removing \(\Ec\) from \(B\) makes it a tree.

     Let \(U\) be the set of neighbours of \(v\) in \(S\). Since \(S\) is connected and does not contain \(v\), there is a spanning tree \(T\) of \(G[S]\). Note that since \(B \setminus \Ec\) is a tree where all vertices of \(U\) are adjacent to \(v\), any path between two vertices of \(U\) with length \(\leq 2r \) in \(G[S]\) has to contain the edge in \(\Ec\).     Therefore, if a component of \(T \setminus \Ec\) contains \(\ell\geq 2\) vertices from \(U\), then its size is at least \(\ell r  \).
     Removing the edge (if any) of \(\Ec\) from \(T\) disconnects it to at most \(2\) components. Since we assumed \(|U| = |E_G(v,S)|>2\), at least \(|E_G(v,S)| - 1\) of the vertices in \(U\) are in components of \(T\setminus\Ec\) together. Thus, summing up we obtain $|S| \geq (|U| - 1)  r.$ 
     Rearranging we get that 
     \(|E_G(v,S)|\leq 1 + \max(1, \frac{|S|}{ r} )\leq 2 + \frac{d|S|}{\log n}\).
\end{proof}

\begin{corollary}\label{cor:sat}
    For all real \(d\) large enough, w.h.p. over \(G\sim\G(n,d/n)\) the following holds. Let \(v\in V_C\), \(F\in\Omega^\ord\). Then for any \(\gamma\in\Gamma^\ord_v(F)\), it holds that $e_\out(\gamma) \geq |E_G(\Vin,\overline\Vin)| - 2 - \frac{d|\V_\gamma|}{\log n}$.
\end{corollary}
\begin{proof}
By Observation~\ref{obs:name}, any in-edges in \(E_G(\Vin,\overline\Vin)\) are from \(\Vin\) to \(v\). If \(v\in\Vin\), all edges in \(E_G(\Vin,\overline\Vin)\) are out, so the inequality follows. If \(v\not\in\Vin\), 
recall that,  
for a singleton polymer~$\gamma$, $\Vin=\emptyset$ and for a non-singleton polymer~$\gamma$,
$\Vin$ is connected in~$C$.
Apply  Lemma~\ref{lem:X} to  show that
$E_G(v,\Vin) \leq 2 + d |\Vin|/\log(n)$. 
\end{proof}

Now we can bound weight of any ordered polymer \(\gamma\), provided it is large enough, to apply expansion and average degree bounds.

\begin{lemma}\label{lem:weight-decaying-for-large-polymers}
    There exists a real \(A > 0\) such that for all \(d\) large enough, and for all real \(q > 1\), w.h.p. over \(G\sim\G(n,d/n)\) the following holds for all \(\beta \geq \beta_0\). 
    For any \(F\in\Omega^\ord_C\), \(v\in V_C\) and \(\gamma\in\Gamma_v^\ord(F)\) with \(|\Vin|\geq\frac{A}{d}\log n\), 
 it holds that $w^\ord_C(F) \leq q^{-\deg_G(\V_\gamma)/(20d)}$.
\end{lemma}
\begin{proof}
    By Propositions~\ref{prop:average-degree-of-connected-sets} and~\ref{prop:expansion-of-core}, there is a real \(A > 0\) such that for all \(d\) large enough, w.h.p. over \(G\), for any connected vertex subset \(S\subseteq V_C\)  with \(|S|\geq \frac{A\log n}{d}\),  its average degree is at least \(\frac{35}{36}d\). Moreover,  if \(\deg_G(S)\leq\tfrac{1}{10}\deg_G(V_C)\), then \(|E_G(S,\overline{S})|\geq\tfrac 35\deg_G(S)\). Thus for the rest of this proof, assume that \(G\) satisfies these properties. 
    
    We first show that for any ordered polymer \(\gamma\) with \(|\V_\gamma|\geq\frac{A}{d}\log n\),  \begin{equation}\label{eq:ordered-polymer-out-edges-bound}
        e_\out(\gamma)\geq \tfrac{4}{7}\deg_G(\V_\gamma).
    \end{equation}    
    By Lemma~\ref{lem:ordered-polymers-are-small}, \(\deg_G(\V_\gamma)\leq \tfrac{1}{10}\deg_G(V_C)\). Thus  by Corollary~\ref{cor:sat}, \(e_\out(\gamma)\geq\tfrac 35\deg_G(\V_\gamma) - 2 - \frac{d|\V_\gamma|}{\log n}\). Finally, \(\V_\gamma\) is a connected subset of~$V_C$
   with $|V_\gamma| \geq \tfrac{A}{d} \log n$,   hence \(\frac{2+d|\V_\gamma|/\log n}{\deg_G(\V_\gamma)} \leq \frac{2+d|\V_\gamma|/\log n}{35d|\Vin|/36}\). Since  \(|\Vin|=\Omega(\log n)\), 
this fraction is at most~$1/35$ for sufficiently large~$n$. Thus,   
 \(\frac{e_\out(\gamma)}{\deg_G(\V_\gamma)}\geq\tfrac35-\tfrac{1}{35} = \tfrac{4}{7}\), concluding the proof of~\eqref{eq:ordered-polymer-out-edges-bound}.

Recall 
from Definition~\ref{def:45f43w}
that \(c'(\gamma)\) is the number of  components of \((V_C, E_C\setminus E_\out(\gamma))\) with total degree at most $(1/2) \deg_G(V_C)$, that is, the number of small components of \((V_C, E_C\setminus E_\out(\gamma)))\).

We first argue that $c'(\gamma) \leq |\Vin|$. Let $\kappa$ be an arbitrary small component of 
\((V_C,E_C\setminus E_\out(\gamma))\), we will show that $V_G(\kappa) \subseteq \Vin$.
Since 
$F \subseteq E_C \setminus E_\out(\gamma)$, 
all vertices in $\kappa$ are contained in small components of $(V_C,F)$.
Thus, from the polymer definition, all vertices of $\kappa$ are contained in   $\V(F,v)$. Since $C$ is connected, the 
component $\kappa$ of $(V_C,E_C \setminus E_\out(\gamma))$ has a vertex incident to some  edge $e\in E_\out(\gamma)$. 
Since $\gamma$ is a non-singleton polymer,   
the edges of $E_\out(\gamma)$ are all incident to vertices in $\Vin$, and thus in particular there is \(w\in\Vin\) incident to $e$. Without loss of generality, we may assume that $w\in V_G(\kappa)$.\footnote{If not, then  let $u$ be the other endpoint of $e$. 
Since it is adjacent to~$e$, 
$u\in V_G(\kappa)$ so it is in $\V(F,v)$. 
But since \(u\in\V(F,v)\) and \(u\) is 
an endpoint of $e$, it is a vertex of~$\gamma$,  so  
by the definition of~$\Vin$, \(u\in\Vin\). In this case we can just use $u$ as the vertex $w$. So we can assume $w \in V_G(\kappa)$.} 
We have also shown already that all vertices of $\kappa$ are contained in $\V(F,v)$. By construction, $\Vin$ is a maximal subset of $\V(F,v)$ that is connected in $C$.
Since $\kappa$ contains a vertex $w$ of $\Vin$, all of its vertices are contained in $\Vin$, so $V_G(\kappa) \subseteq \Vin$. Hence, $c'(\gamma) \leq |\Vin|$.

Next, using~\eqref{eq:ordered-polymer-out-edges-bound}, the average degree bound and the fact that \(c'(\gamma)\) is at most \(|\V_\gamma|\), we get that
    \begin{equation}\label{eq:ordered-polymers-component-bound}
        \frac{c'(\gamma)}{e_\out(\gamma)} \leq \frac{|\V_\gamma|}{\frac{4}{7}\deg_G(\V_\gamma)} \leq \frac{1}{\frac{4}{7}\frac{35}{36}d} =
        \frac{9}{5d}
    \end{equation}
To finish the proof, plug  
\eqref{eq:ordered-polymers-component-bound} and
then \eqref{eq:ordered-polymer-out-edges-bound}
into the definition of a polymer weight (Definition~\ref{def:45f43w}) and use the fact that \(\emm^\beta -1 \geq q^{(2-1/10)/d}\) for $\beta\geq \beta_0$ (See~\eqref{eq:beta0beta1}) to obtain:
    \[
w^\ord_C(\gamma) =  q^{c'(\gamma)} (\emm^\beta - 1)^{-e_\out(\gamma)}
\leq
[q^{\frac{9}{5d} - \frac{2-1/10}{d}}]^{e_\out(\gamma)} = q^{-e_\out(\gamma)/10}\leq q^{-\deg_G(\V_\gamma)/(20d)}.\qedhere
    \]
\end{proof}

\begin{remark}\label{rem:d}
Note that Lemma~\ref{lem:weight-decaying-for-large-polymers} crucially requires \(d\) to be large enough for sufficient expansion and average degree properties to hold for all connected sets of size at least  \( A\log n / d\). While expansion properties hold for all \(d > 1\) (see e.g.~\cite{fountoulakis2007evolution}), we
must treat sets as large as $A \log n/d$
and for these we require large~$d$ to obtain the required 
average degree properties and to show that the weights of the polymers are sufficiently small.
\end{remark}

We next relate the weight of an ordered polymer to the probability of it occurring.
\begin{lemma}\label{lem:bbbb}
    For all \(d\) large enough, w.h.p. over \(G\sim\G(n,d/n)\) the following holds. For \(v\in V_C\) and any \(\gamma\in\bigcup_{F\in\Omega_C^\ord}\Gamma^\ord_v(F)\), it holds that $\pi^\ord_C(\gamma\in\Gamma_v^\ord(\cdot))\leq w^\ord_C(\gamma).$
\end{lemma}
\begin{proof}
    Fix \(v\in V_C\) and let \(\gamma\in\bigcup_{F\in\Omega_C^\ord} \Gamma_v^\ord(F)\).
    If \(e_\out(\gamma) = 0\), then since \(C\) is connected, \(w^\ord_C(\gamma) = 1\), and the inequality is trivial. Thus, for the rest of the proof, assume that \(E_\out(\gamma)\neq \emptyset\). Let $\Omega_C^{\ord}(\gamma)=\{F\in\Omega^\ord_C; \gamma\in\Gamma_v^\ord(F)\}$ and note that
    \[  \pi^\ord(\gamma\in\Gamma_v^\ord(\cdot)) = \frac{\sum_{F\in \Omega_C^{\ord}(\gamma)} w_C(F)}{\sum_{F\in\Omega^\ord_C} w_C(F)}.
    \]
    For every ordered \(F\) with \(\gamma\in\Gamma_v^\ord(F)\), define \(F' := F\cup E_\out(\gamma) \neq F\). We claim that \(w_C(F') = w_C(F) / w_C^\ord(\gamma)\).
Clearly, \(|F'|=|F|+|E_\out(\gamma)|\) since all edges in \(E_\out(\gamma)\) are out-edges. 

If \(\gamma\) is a singleton polymer, then neither of its endpoints are in small components of \((V_C,F)\), thus setting its single edge to be in does not change the number of small components, so \(c(F) = c(F')\), and it follows that \(w_C(F') = w_C(F) (\emm^\beta-1) = w_C(F)/w_C^\ord(\gamma)\).

Hence, assume that \(\gamma\) is a non-singleton polymer 
with $E_\out(\gamma) \neq \emptyset$ 
and consider the components of \((V_C,F)\) and \((V_C,F')\). 
From Lemma~\ref{lem:ordered-polymers-are-small}, there is exactly one giant component in each of \((V_C,F)\) and \((V_C,F')\).
We claim that the components of \((V_C,F')\) are the giant component, 
together with
the small components of \((V_C,F)\) that do \textit{not} contain vertices of \(\Vin\).  

To prove the claim,
first suppose \(\kappa\) is a small component of \((V_C,F)\) not containing vertices of \(\Vin\). We use \(V_G(\kappa)\) for the set of vertices of \(\kappa\). Then  no edge from \(E_G(V_G(\kappa), \overline{V_G(\kappa)})\) is in \(E_\out(\gamma)\), thus  $\kappa$ is a component of \((V_C,F')\).

For the other direction, let \(\kappa\) be a small component of \((V_C,F')\). 
The vertices of $\kappa$ are in small components of $(V_C,F)$, so  by the definition of $\V(F,v)$,
$V_G(\kappa) \subseteq \V(F,v)$. 
We wish to show that $\kappa$ is a small component of $(V_C,F)$ that does not contain vertices of~$\Vin$.

Suppose for contradiction that
\(\Vin\cap V_G(\kappa)\neq \emptyset\). Then, by the definition of non-singleton polymers, \(V_G(\kappa)\subseteq \Vin\), resulting in \(E_G(V_G(\kappa),\overline{V_G(\kappa)})\subseteq E_\out(\gamma)\), contradicting the fact that \(\kappa\) is a  small component of \((V_C,F')\).  
For the rest of the proof, we can therefore assume that \(\Vin\cap V_G(\kappa) = \emptyset\).

To finish the proof of the claim, we wish to show that $\kappa$ is a small component of $(V_C,F)$.
Since \((V_C,F)\) is a subgraph of \((V_C,F')\), there is a non-empty subgraph 
\(\kappa'\) of \(\kappa\) which is a component of \((V_C,F)\).  Suppose for contradiction that \(V_G(\kappa')\subsetneq V_G(\kappa)\). Since \(\kappa\) is connected in~$F'$, there exists an edge \(e = \{w, u\}\) where \(w\in V_G(\kappa)\setminus V_G(\kappa')\) and \(u\in V_G(\kappa')\). But then \(e\) must be an in-edge for \(F'\) but not for \(F\), hence \(e\in E_\out(\gamma)\). By the definition of non-singleton polymers, either \(w\) or \(u\) is in \(\Vin\). Both $w$ and $u$ are in $V_G(\kappa)$ so we have contradicted
\(\Vin\cap V_G(\kappa) = \emptyset\), which we have already shown.
Thus \(\kappa = \kappa'\), which concludes the proof of the claim.

By the claim,
\(c(F) - c(F')\) is the number of small components of \((V_C,F)\) that contain vertices of \(\Vin\). By the definition of polymers this is exactly \(c'(\gamma)\), so by the definition of polymer weights 
\(w_C(F') 
= q^{c(F')}(\emm^\beta-1)^{|F'|}
= \frac{w_C(F)(\emm^\beta-1)^{e_\out(F)}}{q^{c'(\gamma)} } = \frac{w_C(F)}{w_C^\ord(\gamma)}\). 

For every $F\in \Omega^\ord_C(\gamma)$, both $F$ and $F'$ are in $\Omega^\ord_C$.  
Furthermore, $F' \notin \Omega^\ord_C(\gamma)$
and for any distinct $F_1,F_2 \in 
\Omega^\ord_C(\gamma)$,
$F_1'$ and $F_2'$ are distinct.
Thus
 \[  \pi^\ord(\gamma\in\Gamma_v^\ord(\cdot)) = \frac{\sum_{F\in \Omega_C^{\ord}(\gamma)} w_C(F)}{\sum_{F\in\Omega^\ord_C} w_C(F)}
 \leq
\frac{\sum_{F\in \Omega_C^{\ord}(\gamma)} w_C(F)}{\sum_{F\in\Omega^\ord_C(\gamma)} (w_C(F) + w_C(F'))}.
    \]
Hence we can write
    \[
    \pi^\ord(\gamma\in\Gamma_v^\ord(\cdot))  \leq \frac{\sum_{F\in \Omega_C^{\ord}(\gamma)} w_C(F)}{\sum_{F\in \Omega_C^{\ord}(\gamma)} (w_C(F) + w_C(F'))} \leq \frac{\sum_{F\in \Omega_C^{\ord}(\gamma)} w_C(F)}{\sum_{F\in \Omega_C^{\ord}(\gamma)} w_C(F)(1 + 1/w^\ord_C(\gamma))},
    \]
    which is \(\frac{1}{1+1/w^\ord_C(\gamma)}\leq w^\ord_C(\gamma)\).
\end{proof}

With this in hand, we can show an analogue of the Koteck\'y-Preiss condition for our ordered polymers (but only considering those that are large enough).
We will use this to   bound the probability of a large polymer occurring.

\begin{proposition}\label{prop:no-large-ordered-polymer}
    There is a real  \(A > 0\) such that for all \(d\) large enough, for all sufficiently large reals $q$, the following holds w.h.p. over \(G\sim\G(n,d/n)\) for all real \(\beta \geq \beta_0\). For any \(v\in V_C\),
    \[
    \pi^\ord_C(\exists \gamma\in\Gamma^\ord_v(\cdot) \text{ with } |\Vin|\geq \tfrac{A}{d}\log n) \leq q^{-\frac{A}{25d}\log n}.
    \]
\end{proposition}
\begin{proof}
By Proposition~\ref{prop:average-degree-of-connected-sets} and Lemma~\ref{lem:weight-decaying-for-large-polymers}, there is a constant \(A\) such that w.h.p. any connected vertex set with size at least \(\tfrac{A\log n}{d}\) has average degree at least \(\tfrac{9}{10} d\), and any ordered polymer \(\gamma\) with \(|\V_\gamma|\geq\tfrac{A\log n}{d}\) has 
$w_C^{\ord}(\gamma) \leq  q^{-\deg_F(\V_\gamma)/(20d)}$. Also w.h.p. the bound from  Lemma~\ref{lem:bbbb} holds.

Let $\Gamma_v^\ord=\bigcup_{F\in\Omega_C^\ord}\Gamma_v^\ord(F)$. By a union bound and Lemma~\ref{lem:bbbb} it is enough to show   
\[        \sum_{\gamma\in \Gamma_v^\ord;\,  |\Vin|\geq A\log n/d} w^\ord_C(\gamma) \leq q^{-A\log n/(25d)}.
\]
Start by using the bound on the polymer weights from Lemma~\ref{lem:weight-decaying-for-large-polymers}   to show that 
\[
 \sum_{\gamma\in \Gamma_v^\ord;\,  |\Vin|\geq A\log n/d} w^\ord_C(\gamma) \leq
\sum_{\gamma\in \Gamma_v^\ord;\,  |\Vin|\geq A\log n/d} q^{-\deg_G(\V_\gamma)/(20d)}.\]
Then note that any $\gamma$ in this sum
has $|\Vin|\geq A\log n/d$ hence
by Proposition~\ref{prop:average-degree-of-connected-sets} it has average degree at least
$\tfrac{9}{10} d$ hence degree at least $\tfrac{9A}{10} \log n$.
Thus, the sum is at most 
    \[ 
    \sum_{D\geq \frac{9A}{10}\log n}\sum_{\gamma\in \Gamma_v^\ord;\,  \deg_G(\Vin)= D} q^{-\deg_G(\V_\gamma)/(20d)} 
    \]
    Now we bound the number of polymers of the given total degree~$D$. There are \(\leq n(2e)^{2D-1}\) connected subgraphs of total degree \(D\) in any graph (see e.g. \cite[Lemma 6]{GALANIS2022104894}). For any connected subset \(S\subseteq V_C\) of total degree \(D\), there are at most \(D\) edges incident to its vertices, thus there are at most \(2^D\) polymers \(\gamma\) with \(\Vin = S\) (each of its edges can be either in or out). Thus the sum is at most
    \[
    \sum_{D\geq 9A\log n/10} n(2\emm)^{2D-1} 2^D q^{-D / (20d)}.
    \]
    For all \(q\) sufficiently large, \(8\emm^2 q^{-1/(20d)}\leq q^{-1/(22d)}\), so the sum is \(O(nq^{-9A\log n/(220d)})\). For \(q\) sufficiently large, this is \(\leq q^{-A\log n/(25d)}\).
\end{proof}

Proposition~\ref{prop:no-large-ordered-polymer}, Lemma~\ref{lem:path-large-polymer} and equation~\eqref{eq:WSMbound} together with the bound from Theorem~\ref{thm:phase-transition} now directly give Lemma~\ref{lem:ordered-WSM}, i.e. WSM within the ordered phase at the distance \(r = \frac{A\log n}{d}\) (for some constant \(A\) -- take twice the constant from the the Proposition~\ref{prop:no-large-ordered-polymer} to take into account that we want to avoid polymers of size \(r - 1\)).

\section{Fast convergence of the RC dynamics and proof of Theorem~\ref{thm:mainthm}}\label{sec:RCdynamics}
In order to prove  Theorem~\ref{thm:mainthm}, we use the  RC dynamics $(X_t)_{t\geq 0}$ to obtain samples from the ordered and the disordered phases separately. For the sake of completeness, recall that a transition from \(X_t\) to \(X_{t+1}\) is done as follows:
\begin{enumerate}[(1)]
    \item Pick an edge \(e\in E\) uniformly at random
    \item If \(e\) is a cut edge of the graph \((V,X_t\cup\{e\})\), then with probability \(\hat p := \frac{\emm^\beta-1}{\emm^\beta+q-1}\) set \(X_{t+1} := X_t\cup \{e\}\). With all remaining probability set \(X_{t+1} := X_t \setminus \{e\}\)
    \item Otherwise, with probability \(p := 1-\emm^{-\beta}\), set \(X_{t+1} := X_t\cup\{e\}\). With all remaining probability set \(X_{t+1} := X_t \setminus \{e\}\).
\end{enumerate}

Proving mixing of the RC dynamics $(X_t)$ is in general a challenging task, especially since its moves depend on the current component structure. Fortunately, by now there is a well-established strategy \cite{SinclairsGheissari2022, Galanis_Goldberg_Smolarova_2024, galanis2024plantingmcmcsamplingpottsarxiv} that we can use to show the convergence bounds of Theorems~\ref{thm:disordered-mixing} and~\ref{thm:ordered-mixing} based on the polymer sizes. It consists of utilising the following two main ingredients: (i)  \textit{weak spatial mixing within a phase} (WSM) that as we saw asserts that the marginal probability of an edge $e$ being an in-edge depends only on a small neighbourhood around $e$, and (ii) the mixing time of the dynamics on that neighbourhood is polynomial in its size.  Ingredient (ii) is fairly straightforward in sparse random graphs since small-distance neighbourhoods are typically tree-like (i.e., they contain at most a constant number of cycles) which makes poly-time bounds relatively simple to obtain.
\begin{theorem}\label{thm:disordered-mixing}
    Let $d$ be sufficiently large. Then, for all $q$ sufficiently large,  w.h.p. over \(G\sim\G(n,d/n)\), for all \(\beta \leq \beta_1\) and \(\epsilon > \emm^{-\Omega(n)}\), the random cluster dynamics \(X_t\) initialized from all-out satisfies \(||X_T - \pi^\dis_C|| \leq \epsilon\) for \(T = O(n\log n\log\tfrac 1\epsilon)\).
\end{theorem}

\begin{theorem}\label{thm:ordered-mixing}
    Let $d$ be sufficiently large. Then, for all $q$ sufficiently large,  w.h.p. over \(G\sim\G(n,d/n)\), for all $\beta\geq \beta_0$ and \(\epsilon > \emm^{-\Omega(n)}\), the random cluster dynamics \(X_t\) initialized from all-in satisfies \(||X_T-\pi^\ord_C||\leq\epsilon\) for \(T = \mathrm{poly}(n,\log\frac{1}{\epsilon})\)
\end{theorem}

\begin{proof}[Proof of Theorem~\ref{thm:mainthm}]
    For \(\beta \leq \beta_0\) and \(\beta\geq\beta_1\), it follows from Theorems~\ref{thm:phase-transition}, \ref{thm:disordered-mixing} and~\ref{thm:ordered-mixing} that we can get an \(e^{-\Omega(n)}\)-sample by initializing the Glauber dynamics from all-out or all-in respectively.

    Note that both theorems give mixing even for \(\beta\) that depends on \(n\). We use this to obtain the an approximation for \(Z_C\) and hence mixing for \(\beta\in(\beta_0,\beta_1)\). By considering a sequence of \(\beta\)'s we can approximate \(Z_C^\ord\) and \(Z_C^\dis\) within \((1\pm \emm^{-\Omega(n)})\)-factor using standard self-reducibility techniques (for more details see e.g. \cite{fastpolymers, bezakova}), and thus, by Theorem~\ref{thm:phase-transition} get a \((1\pm \emm^{-\Omega(n)})\) approximation for \(Z_C\). For \(\beta\in (\beta_0,\beta_1)\), we then obtain an \(\emm^{-\Omega(n)}\) sample by starting in the appropriate random mixture of all-in and all-out.
\end{proof}
Given WSM within the phases and using the monotonicity of the random-cluster model, we complete the proof of Theorems~\ref{thm:disordered-mixing} and~\ref{thm:ordered-mixing} in Appendix~\ref{sec:putting-ingredients-together}.
\bibliographystyle{alpha}
\bibliography{fn.bib}

@inproceedings{liao2019counting,
author =	{Liao, Chao and Lin, Jiabao and Lu, Pinyan and Mao, Zhenyu},
  title =	{Counting Independent Sets and Colorings on Random Regular Bipartite Graphs},
  booktitle =	{Approximation, Randomization, and Combinatorial Optimization. Algorithms and Techniques (APPROX/RANDOM 2019)},
  pages =	{34:1--34:12},
  year =	{2019},
  volume =	{145}
}

@INPROCEEDINGS{10756117,
  author={Jenssen, Matthew and Perkins, Will and Potukuchi, Aditya and Simkin, Michael},
  booktitle={2024 IEEE 65th Annual Symposium on Foundations of Computer Science (FOCS)}, 
  title={Sampling, Counting, and Large Deviations for Triangle-Free Graphs Near the Critical Density}, 
  year={2024},
  pages={151--165},
  doi={10.1109/FOCS61266.2024.00020}
}

@article{containers1,
author = {Jenssen, Matthew and Perkins, Will and Potukuchi, Aditya},
title = {Approximately counting independent sets in bipartite graphs via graph containers},
journal = {Random Structures \& Algorithms},
volume = {63},
number = {1},
pages = {215--241},
year = {2023}
}

@article{BGRC,
author = {Blanca, Antonio and Galanis, Andreas and Goldberg, Leslie Ann and \v{S}tefankovi\v{c}, Daniel and Vigoda, Eric and Yang, Kuan},
title = {Sampling in Uniqueness from the {P}otts and Random-Cluster Models on Random Regular Graphs},
journal = {SIAM Journal on Discrete Mathematics},
volume = {34},
number = {1},
pages = {742--793},
year = {2020}
}

@article{SWtrees,
author = {Blanca, Antonio and Chen, Zongchen and \v{S}tefankovi\v{C}, Daniel and Vigoda, Eric},
title = {The {S}wendsen–{W}ang dynamics on trees},
journal = {Random Structures \& Algorithms},
volume = {62},
number = {4},
pages = {791--831},
year = {2023}
}

@article{PR,
  title={Deterministic polynomial-time approximation algorithms for functions and graph polynomials},
  author={Patel, Viresh and Regts, Guus},
  journal={SIAM Journal on Computing},
  volume={46},
  number={6},
  pages={1893--1919},
  year={2017}
}

@book{barvinokbook,
  title={Combinatorics and Complexity of Partition Functions},
  author={Barvinok, A.},
  series={Algorithms and Combinatorics},
  year={2017},
  publisher={Springer International Publishing}
}

@article {KP,
    AUTHOR = {Koteck\'{y}, R. and Preiss, D.},
     TITLE = {Cluster expansion for abstract polymer models},
   JOURNAL = {Comm. Math. Phys.},
  FJOURNAL = {Communications in Mathematical Physics},
    VOLUME = {103},
      YEAR = {1986},
    NUMBER = {3},
     PAGES = {491--498}
}

@INPROCEEDINGS {Kuikui,
author = { Liu, Kuikui and Mohanty, Sidhanth and Rajaraman, Amit and Wu, David X. },
booktitle = { 2024 IEEE 65th Annual Symposium on Foundations of Computer Science (FOCS) },
title = {Fast Mixing in Sparse Random {I}sing Models},
year = {2024},
pages = {120-128}
}

@InProceedings{efthymiou,
  author =	{Efthymiou, Charilaos and Feng, Weiming},
  title =	{On the Mixing Time of {G}lauber Dynamics for the Hard-Core and Related Models on $G(n,d/n)$},
  booktitle =	{50th International Colloquium on Automata, Languages, and Programming (ICALP 2023)},
  pages =	{54:1--54:17},
  year =	{2023},
  volume =	{261}
}

@article{bezakova,
author = {Bez\'{a}kov\'{a}, Ivona and Galanis, Andreas and Goldberg, Leslie Ann and \v{S}tefankovi\v{c}, Daniel},
title = {Fast Sampling via Spectral Independence Beyond Bounded-degree Graphs},
year = {2024},
volume = {20},
number = {1},
journal = {ACM Trans. Algorithms},
articleno = {7}
}

@article{integrate,
author = {Moses, Joel},
title = {Symbolic integration: the stormy decade},
year = {1971},
volume = {14},
number = {8},
issn = {0001-0782},
journal = {Commun. ACM},
pages = {548–560},
numpages = {13},
}

@article{NEFF199681,
title = {An Efficient Algorithm for the Complex Roots Problem},
journal = {Journal of Complexity},
volume = {12},
number = {2},
pages = {81--115},
year = {1996},
author = {C.Andrew Neff and John H. Reif},
}

@inProceedings{efthymiou_colorings,
author = {Charilaos Efthymiou and Thomas P. Hayes and Daniel {\v{S}}tefankovi{\v{c}} and Eric Vigoda},
title = {Sampling Random Colorings of Sparse Random Graphs},
booktitle = {Proceedings of the 2018 Annual ACM-SIAM Symposium on Discrete Algorithms (SODA)},
pages = {1759-1771},
year=2018
}

@InProceedings{yin_zhang,
  author =	{Yin, Yitong and Zhang, Chihao},
  title =	{Sampling in {P}otts Model on Sparse Random Graphs},
  booktitle =	{Approximation, Randomization, and Combinatorial Optimization. Algorithms and Techniques (APPROX/RANDOM 2016)},
  pages =	{47:1--47:22},
  year =	{2016},
  volume =	{60}
}

@article{Galanis_Goldberg_Smolarova_2024,
    title={Sampling from the random cluster model on random regular graphs at all temperatures via Glauber dynamics}, 
    journal={Combinatorics, Probability and Computing},
    author={Galanis, Andreas and Goldberg, Leslie Ann and Smolarova, Paulina},
    year={2024},
    pages={1–33}
}

@misc{galanis2024plantingmcmcsamplingpottsarxiv,
      title={Planting and {MCMC} Sampling from the {P}otts model}, 
      author={Andreas Galanis and Leslie Ann Goldberg and Paulina Smolarova},
      year={2024},
      eprint={2410.14409},
      archivePrefix={arXiv},
      primaryClass={math.PR},
      url={https://arxiv.org/abs/2410.14409}, 
}

@article{blanca2021random,
  title={Random-cluster dynamics on random regular graphs in tree uniqueness},
  author={Blanca, Antonio and Gheissari, Reza},
  journal={Communications in Mathematical Physics},
  volume={386},
  number={2},
  pages={1243--1287},
  year={2021}
}

@article{Blanca2,
	title = {Sampling in Uniqueness From the {P}otts and Random-Cluster Models on Random Regular Graphs},
	author = {Antonio Blanca and Andreas Galanis and Goldberg, {Leslie Ann} and Daniel \v{S}tefankovi\v{c} and Eric Vigoda and Kuan Yang},
	year = {2020},
	volume = {34},
	pages = {742--793},
	journal = {SIAM Journal on Discrete Mathematics},
	number = {1}	
}

@article{fastpolymers,
author = {Chen, Zongchen and Galanis, Andreas and Goldberg, Leslie A. and Perkins, Will and Stewart, James and Vigoda, Eric},
title = {Fast algorithms at low temperatures via {M}arkov chains},
journal = {Random Structures \& Algorithms},
volume = {58},
number = {2},
pages = {294--321},
year = {2021}
}

@inproceedings{RClattice,
author = {Reza Gheissari and Alistair Sinclair},
title = {Spatial mixing and the random-cluster dynamics on lattices},
booktitle = {Proceedings of the 2023 Annual ACM-SIAM Symposium on Discrete Algorithms, \emph{(SODA '23)
}},
pages = {4606--4621},
year      = {2023}
}

@inproceedings{helmuth2019algorithmic,
  title={Algorithmic {P}irogov-{S}inai theory},
  author={Helmuth, Tyler and Perkins, Will and Regts, Guus},
  booktitle={Proceedings of the 51st Annual ACM SIGACT Symposium on Theory of Computing, , \emph{(STOC '19)
}},
  pages={1009--1020},
  year={2019}
}

@inproceedings{SinclairsGheissari2022,
  author    = {Reza Gheissari and
               Alistair Sinclair},
  title     = {Low-temperature {I}sing dynamics with random initializations},
  booktitle = {54th Annual {ACM} {SIGACT} Symposium on Theory of Computing, \emph{(STOC '22)}},
  pages     = {1445--1458},
  year      = {2022}
}

@article{MosselSly,
author = {Elchanan Mossel and Allan Sly},
title = {Exact thresholds for {I}sing {G}ibbs samplers on general graphs},
volume = {41},
journal = {The Annals of Probability},
number = {1},
pages = {294 -- 328},
year = {2013}
}

@article{RCM-Helmuth2020,
  title={Finite-size scaling, phase coexistence, and algorithms for the random cluster model on random graphs},
  author={Helmuth, Tyler and Jenssen, Matthew and Perkins, Will},
  journal={Annales de l'Institut Henri Poincare (B) Probabilites et statistiques},
  volume={59},
  number={2},
  pages={817--848},
  year={2023}
}

@misc{trevisan2016expanders,
  title={Lecture Notes on Graph Partitioning, Expanders and Spectral Methods},
  author={Trevisan, Luca},
  year={2016}, 
  note={URL: \url{https://lucatrevisan.github.io/books/expanders-2016.pdf}}
}

@inproceedings{gore1997swendsen,
  title={The {S}wendsen-{W}ang process does not always mix rapidly},
  author={Gore, Vivek K and Jerrum, Mark R},
  booktitle={Proceedings of the twenty-ninth annual ACM symposium on Theory of computing},
  pages={674--681},
  year={1997}
}

@article{mean1,
author = {Reza Gheissari and Eyal Lubetzky and Yuval Peres},
title = {Exponentially slow mixing in the mean-field {S}wendsen–{W}ang dynamics},
volume = {56},
journal = {Annales de l'Institut Henri Poincaré, Probabilités et Statistiques},
number = {1},
pages = {68 -- 86},
year = {2020}
}

@article{GJ,
  title={Approximating the partition function of the ferromagnetic {P}otts model},
  author={Goldberg, Leslie Ann and Jerrum, Mark},
  journal={Journal of the ACM},
  volume={59},
  number={5},
  pages={1--31},
  year={2012}
}

@article{coja2023,
  title={Metastability of the {P}otts ferromagnet on random regular graphs},
  author={Coja-Oghlan, Amin and Galanis, Andreas and Goldberg, Leslie Ann and Ravelomanana, Jean Bernoulli and {\v{S}}tefankovi{\v{c}}, Daniel and Vigoda, Eric},
  journal={Communications in Mathematical Physics},
  pages={1--41},
  year={2023},
  publisher={Springer}
}

@article{CAI2016690,
title = {\#BIS-hardness for 2-spin systems on bipartite bounded degree graphs in the tree non-uniqueness region},
journal = {Journal of Computer and System Sciences},
volume = {82},
number = {5},
pages = {690--711},
year = {2016},
author = {Jin-Yi Cai and Andreas Galanis and Leslie Ann Goldberg and Heng Guo and Mark Jerrum and Daniel \v{S}tefankovi\v{c} and Eric Vigoda},
}

@article{GSVY,
  title={Ferromagnetic {P}otts model: Refined {\#BIS}-hardness and related results},
  author={Galanis, Andreas and \v{S}tefankovic, Daniel and Vigoda, Eric and Yang, Linji},
  journal={SIAM Journal on Computing},
  volume={45},
  number={6},
  pages={2004--2065},
  year={2016}
}

@article{BlaGhe,
  author       = {Antonio Blanca and
                  Reza Gheissari},
  title        = {On the tractability of sampling from the {P}otts model at low temperatures
                  via {S}wendsen-{W}ang dynamics},
  journal      = {CoRR},
  volume       = {abs/2304.03182},
  year         = {2023}
}

@misc{fountoulakis2007evolution,
      title={The Evolution of the Mixing Rate}, 
      author={Nikolaos Fountoulakis and Bruce Reed},
      year={2007},
      eprint={math/0701474},
      archivePrefix={arXiv},
      primaryClass={math.CO}
}

@article{anatomy-of-giant,
title = {Anatomy of the giant component: The strictly supercritical regime},
journal = {European Journal of Combinatorics},
volume = {35},
pages = {155-168},
year = {2014},
author = {Jian Ding and Eyal Lubetzky and Yuval Peres}
}

@article{GALANIS2022104894,
title = {Fast mixing via polymers for random graphs with unbounded degree},
journal = {Information and Computation},
volume = {285},
year = {2022},
url = {https://www.sciencedirect.com/science/article/pii/S0890540122000384},
author = {Andreas Galanis and Leslie Ann Goldberg and James Stewart}
}

@article{rcm-on-unbounded-degree-graphs,
author = {Antonio Blanca and Reza Gheissari},
title = {{Sampling from Potts on random graphs of unbounded degree via random-cluster dynamics}},
volume = {33},
journal = {The Annals of Applied Probability},
number = {6B},
pages = {4997 -- 5049},
year = {2023}
}

@book{Frieze_Karoński_2015,
    place={Cambridge}, 
    title={Introduction to Random Graphs},
    publisher={Cambridge University Press}, 
    author={Frieze, Alan and Karoński, Michał}, 
    year={2015}
}

@book{random-graphs-janson-book,
author = {Janson, Svante and {\L}uczak, Tomasz and Ruci\'nski, Andrzej},
booktitle = {Random graphs},
publisher = {John Wiley \& Sons, Inc.},
title = {Random graphs },
year = {2000 - 2000},
}

@article{10.1145/3470865,
author = {Galanis, Andreas and Goldberg, Leslie Ann and Stewart, James},
title = {Fast Algorithms for General Spin Systems on Bipartite Expanders},
year = {2021},
issue_date = {December 2021},
publisher = {Association for Computing Machinery},
address = {New York, NY, USA},
volume = {13},
number = {4},
journal = {ACM Trans. Comput. Theory},
month = sep,
articleno = {25},
numpages = {18}
}

@InProceedings{Spectral,
  author =	{Carlson, Charlie and Davies, Ewan and Kolla, Alexandra and Potukuchi, Aditya},
  title =	{A Spectral Approach to Approximately Counting Independent Sets in Dense Bipartite Graphs},
  booktitle =	{51st International Colloquium on Automata, Languages, and Programming (ICALP 2024)},
  pages =	{35:1--35:18},
  year =	{2024},
  volume =	{297}
}

@article{doi:10.1137/19M1286669,
author = {Jenssen, Matthew and Keevash, Peter and Perkins, Will},
title = {Algorithms for \#{BIS}-Hard Problems on Expander Graphs},
journal = {SIAM Journal on Computing},
volume = {49},
number = {4},
pages = {681-710},
year = {2020}
}

@InProceedings{Unique,
  author =	{Coulson, Matthew and Davies, Ewan and Kolla, Alexandra and Patel, Viresh and Regts, Guus},
  title =	{{Statistical Physics Approaches to Unique Games}},
  booktitle =	{35th Computational Complexity Conference (CCC 2020)},
  pages =	{13:1--13:27},
  ISBN =	{978-3-95977-156-6},
  ISSN =	{1868-8969},
  year =	{2020},
  volume =	{169}
}

@ARTICLE{replica,
       author = {{Barbier}, Jean and {Chan}, Chun Lam and {Macris}, Nicolas},
        title = "{Concentration of Multi-overlaps for Random Dilute Ferromagnetic Spin Models}",
      journal = {Journal of Statistical Physics},
         year = 2019,
        month = dec,
       volume = {180},
       number = {1-6},
        pages = {534--557}
}

@INPROCEEDINGS{Pottsexp,
  author={Carlson, Charlie and Davies, Ewan and Fraiman, Nicolas and Kolla, Alexandra and Potukuchi, Aditya and Yap, Corrine},
  booktitle={2022 IEEE 63rd Annual Symposium on Foundations of Computer Science (FOCS)}, 
  title={Algorithms for the ferromagnetic {P}otts model on expanders}, 
  year={2022},
  pages={344-355}
}

@inproceedings{DBLP:conf/approx/BlancaZ23,
  author       = {Antonio Blanca and
                  Xusheng Zhang},
  title        = {Rapid Mixing of Global Markov Chains via Spectral Independence: The
                  Unbounded Degree Case},
  booktitle    = {Approximation, Randomization, and Combinatorial Optimization. Algorithms
                  and Techniques, {APPROX/RANDOM} 2023},
  volume       = {275},
  pages        = {53:1--53:19},
  year         = {2023}
}

\appendix

\section{Proofs of properties of \(\G(n,d/n)\)}\label{sec:graph-proofs}

In this section we prove the properties from Section~\ref{sec:graph-properties}, as well as introduce a few more results used in proving mixing within the disordered phase.

\subsection{Preliminaries}

To bound the size of radius-\(r\) balls, we use the following lemma by Mossel and Sly \cite{MosselSly}. Note that while their lemma is more general and stated for \(r = (\log\log n)^2\), this bound on~$r$ was not used for upper bounding the volume.

\begin{lemma}[{\cite[Lemma 7]{MosselSly}}]\label{lem:ball-size}
    Let \(d > 1\) be a real. Then w.h.p. over \(G\sim\G(n,d/n)\), for any integer \(r\) and any \(v\in V(G)\), \(|V(B_r(v))|\leq d^r \log n\).
\end{lemma}

Next we show local treelikeness properties of \(\G(n,d/n)\).
we say that a subgraph is \(k\)-\textit{treelike} if at most \(k\) edges can be deleted from it to make it a tree (or a forest).

First we will consider balls of radius \(\approx \tfrac 1d \log n\). Recall that the radius-\(r\) ball around a vertex \(v\) is an induced subgraph consisting of all vertices with distance \(\leq r\) to \(v\).

\begin{lemma}\label{lem:G-1-treelike}
    Let \(A>0\) be a real. Then for all \(d\) large enough, w.h.p. over \(G\sim\G(n, d/n)\), each radius-\(\tfrac Ad\log n\) ball in \(G\) is \(1\)-treelike.
\end{lemma}
\begin{proof}
    This follows from the proof of \cite[(4.2)]{rcm-on-unbounded-degree-graphs} (note this proof is stated for more general graph family, they note in Example 1.5 that it includes \(\G(n,d/n)\). Noting that the maximum degree of \(\G(n,d/n)\) is \(O(\log n)\), the number of edges is \(\tfrac{nd}2+o(n)\) and the size of any such ball is \(\leq n^{A\log d/d}\log n\leq n^{\epsilon}\) (for any constant \(\epsilon\), provided that \(d\) is large enough), we get from their equation (4.2) that the probability that a radius-\(\frac{A}{d}\log n\) ball is not \(1\)-treelike is at most

    \[
    \P\left(\Bin\left(n^{\epsilon}, \frac{n^\epsilon \log n }{dn/2 - o(n)-n^\epsilon}\right) > 1\right) = O(n^{-2+4\epsilon}).
    \]

    Taking \(\epsilon < 1/4\), the result follows from union bound.
\end{proof}

We also show \(O(1)\)-treelikness for sets of constant and \(O(\log n)\) size.

\begin{lemma}\label{lem:logn-sets-are-treelike}
    For any \(d > 1\), there exists \(k > 0\) such that w.h.p. over \(G\sim\G(n,d/n)\), any connected subset of size \(\leq \log n\) is \(k\)-treelike.
\end{lemma}
\begin{proof}
    Here, \(k > 0\) will be a sufficiently large constant we specify later.

    Fix \(1\leq s\leq \log n\). We union bound over connected trees of size \(s\) and events ``\(T\) is a subtree of \(G\) with \(e(V(T))>s-1+k\)''. Fix any tree \(T\) of size \(s\). Then note that conditional on \(T\) being a subgraph of \(G\), we have that \(e(V(T))\sim s-1+\Bin(\binom{s}{2}-(s-1), d/n)\leq s+\Bin(s^2,d/n)\).

    We have \(\P(\Bin(s^2, d/n) \geq k)\leq s^{2k} (\tfrac{d}{n})^k\) by union bound. Lemma~\ref{lem:expected-number-of-subtrees} gives that the expected number of connected trees of size \(s\) is at most \(n(\emm d)^s\), thus expected number of subtrees \(T\) of \(G\) of size \(s\) and \(>s-1+k\) edges between vertices of \(T\) is at most
    \[
    \left(\frac{ds^2}{n}\right)^k n(\emm d)^s = n\exp\{s\log(\emm d) + k\log(\tfrac{ds^2}{n})\} \leq d^k n\exp\{\log(\emm d)\log n - k\log(\tfrac{n}{\log^2n})\}
    \]
    For \(k > \log(\emm d) + 2\) and all \(n\) large enough, this is \(\leq \frac{d^k}{\sqrt n}\).  Then union bound over \(\log n\) possible sizes and then using Markov inequality gives probability \(\leq \log n/\sqrt n= o(1)\).
\end{proof}

To prove the next lemma, as well as Propositions~\ref{prop:average-degree-of-connected-sets} and~\ref{prop:expansion-of-core}, we will use the following bound on the expected number of subtrees of a given size. This result is proven in \cite{fountoulakis2007evolution}.

\begin{lemma}[{\cite[(3)]{fountoulakis2007evolution}}]\label{lem:expected-number-of-subtrees}
    For all reals \(d > 1\) and integers \(s > 0\), the expected number of subtrees with \(s\) vertices in \(\G(n,d/n)\) is at most \(n(\emm d)^s\).
\end{lemma}

\begin{lemma}\label{lem:1-treelikeness-for-constant-size}
    For any \(d > 1\), and an integer \(R > 0\), w.h.p. over \(G\sim\G(n,d/n)\), any connected vertex subset of \(G\) with at most \(R\) vertices is \(1\)-treelike (i.e we can remove at most one edge to make it a tree).
\end{lemma}
\begin{proof}
    Let \(X\) be the number of connected components of size \(\leq R\) with more than one extra edge. For any tree of size \(\leq R\), let \(X_T\) be the event that the tree \(T\) is a subgraph of \(G\) and the subgraph on its vertices \(G[V(T)]\) has \(> |V(T)|\) edges. Then

    \[
    \E[X] \leq \sum_{\substack{T\text{ a tree}\\|V(T)|\leq R}} \P[X_T] = \sum_T \P[X_T \mid T\text{ a subgraph of }G] \P(T\text{ a subgraph of }G).
    \]

    Note that the rest of the \(\binom{|T|}{2}-(|T|-1)\leq R^2\) edges with both endpoints in \(V(T)\) are independent of the event ``\(T\) a subgraph of \(G\)'', thus by union bound \(\P(X_T \mid T\text{ a subgraph of }G)\leq \P(\Bin(R^2, d/n) > 1) \leq \binom{R}{2} \frac{d^2}{n^2} \).
    
    By Lemma~\ref{lem:expected-number-of-subtrees}, the expected number of trees of size up to \(R\) is at most \(nR(\emm d)^R\). Thus \(\E[X]\) is at most
    \[
    \sum_{\substack{T\text{ a tree}\\|V(T)|\leq R}} \P(T\text{ a subgraph of }G) \binom{R}{2}\left(\frac{d}{n}\right)^2 \leq n R(\emm d)^R \cdot \binom{R}{2}\left(\frac{d}{n}\right)^2 = \frac{R(\emm d)^R d^2 \binom{R}{2}}{n} \rightarrow 0.
    \]
    Hence the lemma follows from Markov's inequality.
\end{proof}

\subsection{Average degree of connected sets}\label{sec:avg-deg}

In this subsection, we prove Proposition~\ref{prop:average-degree-of-connected-sets}. We split the proof into three parts, depending on the size of the set. For sets of \(\Omega(n)\) size, we can prove a stronger result -- they do not have to be connected for the average degree bound to hold.

\begin{lemma}\label{lem:avg-deg-of-small-connected-sets}
    Let \(\xi < \epsilon \in(0,1)\). Then there exists a real \(A(\epsilon,\xi) > 0\) such that for all \(d\) large enough, w.h.p. over \(G\sim\G(n,d/n)\) it holds that all connected sets of size \(\tfrac{A\log n}{d}\) and at most \(\xi n\) have average degree at least \((1-\epsilon)d\).
\end{lemma}
\begin{proof}
    Assume we have some \(A > 0\) for now. 
    Since for any connected set of \(G\), there is a subtree of \(G\) with this vertex set, it is enough to show that

    \[
    \sum_{\substack{T\text{ a possible subtree of }G\\|V_G(T)|\geq A\log n/d}} \E[\1\{T\text{ is a subtree of }G\text{ and }\deg_G(T)\leq(1-\epsilon)d|V_G(T)|\}] = o(1),
    \]

    and then the lemma follows from Markov inequality.

    Fix any possible tree \(T\) with \(|V_G(T)| = s\). Then the expectation in the sum for this tree can be rewritten as

    \[\E[\1\{T\text{ a subtree of }G\}]\P(\deg_G(T) < (1-\epsilon)ds \mid T\text{ a subtree of } G).\]

    Note that conditional on ``\(T\) a subtree of \(G\)'', \(\deg_G(T) = 2e_G(T) + |E_G(V_G(T),\overline{V_G(T)})| \geq |E_G(V_G(T),\overline{V_G(T)})| \sim \Bin(s(n-s),\tfrac{d}{n})\). Using \(n-s\geq n(1-\xi)\), this is dominated by \(\Bin(sn(1-\xi),\tfrac dn)\). Thus \(\P(\deg_G(T) < (1-\epsilon)ds \mid T\text{ a subtree of } G)\) is at most 

    \[
    \P(\Bin(s(1-\xi)n, d/n) \leq (1-\epsilon)ds),
    \]

    by Chernoff 
    this is at most \(\exp\{-\frac{ds(\epsilon-\xi)^2}{2(1-\xi)}\}\).

    Then when we sum first through all sizes \(s\geq\tfrac{A\log n}{d}\) and using Lemma~\ref{lem:expected-number-of-subtrees} to bound the expected number of a subtrees of given size, we obtain that the extended number of subtrees of \(G\) of size \(\tfrac{A\log n}{d}\leq s\leq \xi n\) with average degree \(<(1-\epsilon)ds\) is at most

    \begin{align*}
        &\sum_{s=\frac{A\log n}{d}}^{\xi n} \E[\#\text{subtrees of } G\text{ with }s\text{ vertices}] \exp\left\{-\frac{ds(\epsilon-\xi)^2}{2(1-\xi)}\right\} \leq\sum_{s=\frac{A\log n}{d}}^{\xi n} n(\emm d)^s \exp\left\{-\frac{ds(\epsilon-\xi)^2}{2(1-\xi)}\right\}
    \end{align*}

    For all \(d\) sufficiently large, \(\frac{d(\epsilon-\xi)^2}{2(1-\xi)} - 1 - \log d \geq \frac{d(\epsilon-\xi)^2}{3(1-\xi)}\), and hence the sum is
    
    \[
    O\left(\exp\left\{\log n\left[1-\frac{A(\epsilon-\xi)^2}{3(1-\xi)}\right]\right\}\right).
    \]

    For any \(A > \frac{3(1-\xi)}{(\epsilon-\xi)^2}\), this is \(o(1)\), hence the lemma follows.
\end{proof}

\begin{lemma}\label{lem:avg-deg-of-medium-sets}
    Let \(\xi, \epsilon \in (0,1)\) be arbitrary constants. Then, for all \(d\geq\frac{32\log 2}{(\epsilon\xi)^2}\), w.h.p. over \(G\sim\G(n,d/n)\), all vertex sets \(S\subseteq V_G\) of size at least \(\xi n\) and at most \((1-\xi)n\) have average degree at least \((1-\epsilon)d\), and in particular have \(|E_G(S,\overline S)|\geq \frac{d}{n}(1-\epsilon)|S|(n-|S|)\).
\end{lemma}
\begin{proof}
    For any given set \(S\) with size \(\xi n\leq|S|\leq(1-\xi)|S|\), \(\deg_G(S)\) is a sum of independent \(2\Bin(\binom{|S|}{2}, \tfrac dn)\) and \(\Bin(|S|(n-|S|), \tfrac dn)\) random variables, corresponding to inner edges, and edges going outside of \(S\) respectively.

    First we bound \(\P(2\Bin(\binom{|S|}{2}, \tfrac dn)<\tfrac{(1-\epsilon)d|S|^2}{n})\). By Chernoff, for \(n\) large enough so that \((1-\epsilon)/(|S|-1) < \epsilon / 2\) and \(1/n < \xi / 2\),
    this is at most \(\exp\{-\frac{d\epsilon^2 \xi^2}{32}\}\).

    Then for the proportion of the degree corresponding to the edges of going from \(S\) to \(\overline S\), this is \(\sim\Bin(|S|(n-|S|),\tfrac dn)\), and by Chernoff \(\P(\Bin(|S|(n-|S|), \tfrac dn)<(1-\epsilon)\frac{d|S|(n-|S|)}{n})\leq \exp\{-d\epsilon^2 \xi^2 n/2\}\). 

    Hence by union bound, there is such a set with average degree \(<(1-\epsilon)dn\) with probability at most \(\exp\{-d\epsilon^2 \xi^2 n/2\} + \exp\{-\frac{d\epsilon^2 \xi(\xi n - 1)}{32}\}\). There are at most \(2^n\) such sets \(S\), thus by another union bound and by the condition on \(d\), we obtain that the probability of a large set of small average degree existing is at most \(\exp\{-\Omega(n)\} = o(1)\).
\end{proof}

\begin{lemma}\label{lem:avg-deg-of-large-sets}
    For \(\xi,\epsilon\in(0,1)\) with \(\xi < \epsilon / 2\), and all \(d > \frac{8\log 2}{(\epsilon - 2\xi)^2}\), w.h.p. over \(G\sim\G(n,d/n)\), \emph{all} vertex sets \(S\subseteq V(G)\) of size \(\geq(1-\xi)n\) have \(\deg_G(S)\geq(1-\epsilon)d|S|\).
\end{lemma}
\begin{proof}
    For any vertex set \(S\), \(\deg_G(S)\geq 2e_G(S)\sim 2\Bin(\binom{|S|}{2}, d/n)\). If \(|S|\geq (1-\xi)n\), then \(\E[\Bin(\binom{|S|}{2},d/n)] \geq \frac{d|S|(|S| - 1)}{2n} \geq d|S|(1-\xi - 1/n)/2\). For \(n \geq1/\xi\), this is \(\geq d|S|(1-2\xi)/2 > (1-\epsilon)d|S|/2\). Use Chernoff bound
    to get that (the last inequality uses that \(\frac{1-\xi}{1-2\xi}\geq 1\)),

    \[
    \P(e_G(S)\leq (1-\epsilon)d|S|/2)\leq \exp\{-\frac{d|S|(\epsilon-2\xi)^2}{8(1-2\xi)}\} \leq \exp\{-dn(\epsilon-2\xi)^2/8\}.
    \]
    
    Thus, union bounding over at most \(2^n\) possible sets, the probability that exists a vertex set \(S\) with \(|S|\geq(1-\xi)n\) and \(\deg_G(S)\leq (1-\epsilon)d|S|\) is at most \(\exp\{n[\log 2 - d(\epsilon-2\xi)^2/8]\}\), and by the lower bound on \(d\), this is \(\exp\{-\Omega(n)\} = o(1)\).
\end{proof}

Now the Proposition~\ref{prop:average-degree-of-connected-sets} follows directly by applying Lemmas~\ref{lem:avg-deg-of-small-connected-sets}, \ref{lem:avg-deg-of-medium-sets} and~\ref{lem:avg-deg-of-large-sets} we get the proposition with, say \(\xi := \epsilon / 3 < \epsilon/2\).

\subsection{Weak Expansion}
\label{proof:propexpansion}

In this section, we prove Proposition~\ref{prop:expansion-of-core}. As a first step, we show that the number of interval edges of connected subsets is not too large.

\begin{lemma}\label{lem:not-too-many-internal-edges}
    For any \(\xi\in(0, 1]\), there exists a real \(A(\xi) > 0\) such that for all \(d\) large enough, w.h.p. over \(G\sim\G(n,d/n)\), any connected vertex set \(S\subseteq V_G\) with  \(\frac{A\log n}{d}\leq |S|\leq \xi n\) satisfies \(e_G(S) < (2\xi d/3 + 1)|S|\).
\end{lemma}
\begin{proof}
    Assume we have some \(A > 0\) for now. We use Markov inequality on the expected number of subtrees of \(G\), \(T\) with \(\frac{A\log n}{d}\leq |V_G(T)| \leq \xi n\) (as every connected set has a spanning tree). Thus it is enough to show 
    \[
    \sum_{\substack{T\text{ a tree}\\ \frac{A\log n}{d}\leq |V_G(T)|\leq \xi n}} \P(e_G(V_G(T)) \geq (\xi d + 1)|V_G(T)| \mid T\text{ a subtree of }G) \P(T\text{ a subtree of }G) = o(1)
    \]
    Fix a tree \(T\) with \(s\) vertices. Then conditioned on \(T\) being a subgraph of \(G\), \(e_G(V(T))\sim s-1+\Bin(\binom s2 - (s-1),\tfrac dn)\). Since \(\binom s2 - s-1 \leq s^2 / 2\), we can bound

    \[
    \P(e_G(V_G(T)) \geq (2\xi d/3 + 1)s \mid T\text{ a subtree of G}) \leq \P(\Bin(\tfrac{s^2}{2}, \tfrac dn)\geq \tfrac{2\xi ds}{3}).
    \]

    Noting that \(s\leq \xi n\), Chernoff bound gives that
    this is at most \(\exp\{-\xi ds / 216\}\).

    By Lemma~\ref{lem:expected-number-of-subtrees}, the expected number of subtrees of size \(s\) is at most \(n(\emm d)^s\). Hence summing over possible subtree sizes, we get that the expected number of ``bad'' subtrees is at most

    \[
    \sum_{\frac{A\log n}{d}\leq s\leq \xi n} n[d\emm^{1+\xi d/216}]^s.
    \]

    For all \(d\) large enough, \(d\emm^{1+\xi d/216} < \exp\{-\xi d / 250\}\), and thus then the sum is \(O(n^{1-A \xi / 250})\), which is \(o(1)\) for any \(A > \tfrac{250}{\xi}\).
\end{proof}

With this lemma in hand, we can prove Proposition~\ref{prop:expansion-of-core}, which we restate for convenience.

\weakexpansion*
\begin{proof}
    By Proposition~\ref{prop:average-degree-of-connected-sets} and Lemma~\ref{lem:not-too-many-internal-edges}, there exists a real \(A>0\) such that for all \(d\) large enough, w.h.p., any connected set \(|S|\) with size at least \(\frac{A\log n}{d}\) has \(\deg_G(S) \geq \tfrac{9d|S|}{10}\), and if furthermore \(|S|\leq\tfrac n6\), then \(e_G(S)\leq (1+\tfrac d9)|S|\).

    Since \(\deg_G(S) = 2e_G(S) + |E_G(S,\overline S)|\), we can rearrange and then use the above inequalities to obtain that for any connected vertex set \(S\) with \(\frac{A\log n}{d}\leq |S|\leq \tfrac n6\),

    \[
    \frac{|E_G(S,\overline S)|}{\deg_G(S)} = 1 - \frac{2e_G(S)}{\deg_G(S)} \geq 1 - \frac{20}{9d} - \frac{1}{5}.
    \]

    For all \(d\) large enough, this is at most \(\tfrac{3}{5}\). To finish the proof, note that whenever \(|S|\geq \tfrac n6\), its total degree is at least \(\frac{3dn}{10}\). But w.h.p. \(\tfrac 34 nd > |E_G|\geq |E_C|\), giving \(\frac{\deg_G(V_C)}{10} < \frac{3nd}{10}\leq \deg_G(S)\). Thus whenever \(S\) has \(\deg_G(S)\leq \deg_G(V_C)/10\), then also \(|S|\leq n/6\), thus the expansion applies.
\end{proof}

\subsection{Kernel of \(\G(n,d/n)\)}\label{sec:kernel-expansion}

While \(\G(n,d/n)\) (and \(C\)) only exhibit weak expansion, its \textit{kernel} -- that is, a graph obtained by taking a maximal subgraph of \(C\) with minimum degree \(\geq 2\) and then contracting maximal induced paths into a single edge -- is an expander. This will allow to prove Lemma~\ref{lem:giant-component-in-the-core-after-edge-removal}, as well as Theorem~\ref{thm:phase-transition}. Note that the difference between \(C\) and its maximal subgraph of minimum degree \(\geq 2\) is a \textit{forest}, and we refer to it as the \textit{attached trees} in \(C\).

We will denote the kernel as \(G_K = (V_K,E_K)\) and write \(n_K := |V_K|\). To reason about the kernel, and show its expansion, we are going to use the following contiguous model for the kernel from~\cite{anatomy-of-giant}.

\begin{proposition}[{\cite[Theorem 1]{anatomy-of-giant}}]\label{prop:giant-model}
    Let \(d > 1\) be a constant and \(\mu\in(0,1)\) be its \emph{conjugate} (i.e. the unique solution in \((0,1)\) to \(d\emm^{-d}=\mu\emm^{-\mu}\)). Let \(n > 1\) be an integer. Consider the following graph model \(\Kb(n,d)\) generated as follows:

    Let \(\Delta\sim\Nc(d-\mu,1/n)\) and then let \(D_1,\dots,D_n\) be i.i.d. \(\Poisson(\Delta)\) for conditional on \(\sum_i^n D_i \1_{D_i \geq 3}\) even. 

    The degree sequence of \(\Kb(n,d)\) is \(\{D_i\}_{D_i\geq 3}\), and \(\Kb(n,d)\) is a (multi)graph with this degree sequence chosen u.a.r..

        

    Then \(\Kb(n,d)\) is \emph{contiguous} to the kernel \(G_K\) of \(\G(n,d/n)\), i.e. whenever a statement holds w.h.p. about \(\Kb(n,d)\), it also holds w.h.p. about \(G_K\).
\end{proposition}

Using this model, we first use the following series of bounds the size of the kernel. As a first step, we bound the conjugate of \(d\).

\begin{lemma}\label{lem:mu-bound}
    For all \(d\geq 3\), \(\mu < \emm^{-d/2}\).
\end{lemma}
\begin{proof}
    The function \(x\mapsto x\emm^{-x}\) is (strictly) decreasing on \(x\in(1,\infty)\) and (strictly) increasing on \(x\in(0,1)\), thus there is a unique conjugate for each \(d > 1\) in \((0,1)\) and in particular, note that if for \(a \in (0,1)\) it holds that \(a\emm^{-a} > d\emm^{-d}\), then \(\mu(d) < a\). We can easily verify that \( \emm^{-d/2} \emm^{-\emm^{-d/2}} > d\emm^{-d}\) holds for all \(d\geq 3\).
\end{proof}

With an upper bound on \(\mu\), the following two lemmas give us a lower bound on the size, and the number of edges in the kernel (and thus in \(C\)).
\begin{lemma}\label{lem:lower-bound-on-kernel-size}
    For \(d\geq 4\), w.h.p. over \(K\sim\Kb(n,d)\), \(n_K\geq n(1-\emm^{-d/3})\). Thus also w.h.p. over \(G\sim\G(n,d/n)\), \(n_K \geq n(1-\emm^{-d/3})\).
\end{lemma}
\begin{proof}
    In the model from Proposition~\ref{prop:giant-model}, we have \(\Delta\sim\Nc(d-\mu,1/n)\), hence w.h.p. \(\Delta \geq d - \emm^{-d/2}\) (we used the inequality \(\mu < \emm^{-d/2}\). from the previous lemma). Then note \(n - n_K = \sum_i^n \1\{D_i < 3\}\). Note that while \(D_i\)'s are conditional to have an even sum, this event has a probability bounded below by a constant, thus it is enough to show that w.h.p. sum of \(n\) i.i.d. Bernoulli(\(\P(\Poisson(\Delta) < 3)\)) r.v.s. is at most \(n\emm^{-d/3}\).
    w.h.p. (when on \(\Delta \geq d - \emm^{-d/2}\)) they have mean \(\leq n\emm^{-d-\emm^{-d/2}}(1 + d - \emm^{-d/2} + (d-\emm^{-d/2})^2/2)\). We can verify that for all \(d \geq 4\), there exists constant (independent of \(n\)) \(\delta(d) > 0\) such that this mean is \(\leq n \emm^{-d/3} - n\delta\). By standard Chernoff bounds this implies w.h.p. \(n - n_K \leq n\emm^{-d/3}\), so the result follows. From contiguity of \(\Kb(n,d)\) (Proposition~\ref{prop:giant-model}), and since \(n_K\leq n_C\), the rest of the lemma follows.
\end{proof}
We use the following standard Chernoff bound for real \(0<\delta<\lambda\) and \(X\sim\Poisson(\lambda)\).
\begin{equation}
    \P(X \leq d - \delta) \leq \emm^{-\frac{\delta^2}{\delta+d}}\label{eq:Poisson-lower-tail-bound}
\end{equation}
\begin{lemma}\label{lem:bound-on-kernel-edges}
    For any real \(d \geq 4\), 
    w.h.p. over over \(K\sim\Kb(n,d)\), \(\frac{dn}{2} - 2n\emm^{-d/3}\leq |E_K|\). Thus also w.h.p. over \(G\sim\G(n,d/n)\), \(\frac{dn}{2} - 2n\emm^{-d/3}\leq |E_K|\leq |E_C|\) 
\end{lemma}
\begin{proof}
    We have \(\Delta\sim\Nc(d-\mu,1/n)\), thus w.h.p. \(d - \mu - \tfrac 12\emm^{-d/3} < \Delta\). For \(d\geq 3\), \(\mu\leq \emm^{-d/2}<\emm^{-d/3}\), hence \(\Delta > d-\tfrac32 \emm^{-d/3}\). 

    
    From the previous lemma, w.h.p. \(n - n_K\leq n\emm^{-d/3}\). Also, we can write \(2|E_K| \geq \sum_i^n D_i - 2(n - n_K)\). It is therefore enough to show that w.h.p. \(\sum_i^n D_i \geq n(d - 2\emm^{-d/3})\). Since the probability of \(\sum_i^n D_i\) being even is bounded below by a constant, it is enough to show that w.h.p. \(\Poisson(\Delta n)\geq n(d-2\emm^{-d/3})\). Since  \(\Delta \geq d - \tfrac 32\emm^{d/3}\) w.h.p., we have from \eqref{eq:Poisson-lower-tail-bound} that 
    $\P(\Poisson(\Delta n) \leq n(d-2\emm^{-d/3}))=\emm^{-\Omega(n))}$. 
    The result for \(\G(n,d/n)\) follows from contiguity of \(\Kb(n,d)\).
\end{proof}

With a lower bound on \(|E_K|\), we can now bound \(|E_C|/|E_K|\).

\begin{lemma}\label{lem:not-too-many-expanded-edges-overal}
    For all \(d\) large enough, w.h.p. over \(G\sim\G(n,d/n)\), \(|E_C| \leq (1+\emm^{-d/3})|E_K|\). 
\end{lemma}
\begin{proof}
    We have that w.h.p. \(|E_C|\leq \tfrac{dn}{2}+n\emm^{-d/3}\) (since \(|E_C|\leq |E_G| = \tfrac{nd}{2}+o(n)\)) and \(|E_K|\geq \tfrac{dn}{2}-2n\emm^{-d/3}\) (by previous lemma). Thus
    
    \[\frac{|E_C|}{|E_K|}\leq \frac{d/2 + \emm^{-d/3}}{d/2 - 2\emm^{-d/3}} = 1 + \frac{3\emm^{-d/3}}{d/2 - 2\emm^{-d/3}}.\]
    
    For all \(d\) large enough, \(\frac{3}{d/2 - 2\emm^{-d/3}}\leq 1\), hence giving that \(|E_C|\) is at most \((1+\emm^{-d/3})|E_K|\).
    



\end{proof}

Using the kernel we also show a stronger variant of the average-degree result for linear-sized sets, by bounding the number of low-degree vertices in the kernel.
We use the following notation. For a vertex set \(S_G\subset V_G\), let the corresponding \textit{kernel-representation} be the set \(S_K\subseteq V_K\) containing all vertices which have their counterpart in \(S_G\) (i.e. that have degree \(>2\) in \(G\) and are part of the maximal subgraph of minimal degree \(\geq 2\) of \(C\)). Note that \(\deg_K(\cdot),\, E_K(\cdot,\cdot)\) are defined analogously to \(\deg_G(\cdot),\, E_G(\cdot,\cdot)\).

\begin{proposition}\label{prop:exponentially-few-small-degree-vertices}
    For all real \(d\) large enough, w.h.p. over \(G\sim\G(n,d/n)\), the number of vertices in \(G\), (and thus in its largest component \(C\) and its kernel \(G_K\)) of degree \(\leq(1-d^{-1/3})d\) is at most \(2n\emm^{-\frac 12 d^{1/3}}\). Hence then for any vertex set \(S_G\subseteq V_G\)

    \[
    \deg_G(S_G) \geq \deg_K(S_K) \geq (1-d^{-1/3})d|S_G| - 2dn\emm^{-\frac 12d^{1/3}}
    \]
\end{proposition}
\begin{proof}
    We first claim that w.h.p., the number of vertices of degree \(\leq (1-d^{-1/3})d\) in the \textit{kernel} is at most \(o(n) + n\emm^{-\frac{1}{2}d^{1/3}}\). Note it is enough, by Proposition~\ref{prop:giant-model}, to show this for \(K\sim\Kb(n,d)\).
    
    There, \(\Delta\sim\Nc(d-\mu,1/n)\), and hence w.h.p. \(\Delta \geq d-\emm^{-d/3}\) (recall that by Lemma~\ref{lem:mu-bound}, \(\mu\leq \emm^{-d/2}\) for all \(d\geq 3\)). The number of kernel vertices with degree \(\leq (1-d^{-1/3})d\) is at most \(\sum_i^n \1\{D_i \leq (1-d^{-1/3})d\}\). Since we are proving a w.h.p. result, we can drop the conditioning on the sum of \(D_i\)'s being even, so it is enough to bound the sum of i.i.d. \(\1\{\Poisson(d-2\mu) < (1-d^{-1/3})d\}\) r.v.s. Since \(d - \emm^{-d/3} > (1-d^{-1/3})d\), we use~\eqref{eq:Poisson-lower-tail-bound} to get
    \[
    \P\Big(\Poisson(d-\emm^{-d/3}) \leq (1-d^{-1/3})d\Big) \leq \exp\left(-\tfrac{(d^{2/3}-\emm^{-d/3})^2}{d+d^{2/3}-2\emm^{-d/3}}\right)\leq \emm^{-\frac 12 d^{1/3}},
    \] for all \(d\) large enough. Then using a Chernoff bound for the sum of Bernoulli r.v.s., w.h.p. \(\sum_i^n \1\{D_i \leq (1-d^{-1/3})d\} \leq o(n) + n\emm^{-\frac 12 d^{1/3}}\), which finishes the claim.

    Now, all vertices in \(G\) that correspond to kernel vertices with degree \(\geq (1-d^{-1/3})d\) also have degree \(\geq (1-d^{-1/3})d\). It remains to account for the remaining vertices. The number of these vertices is \(n - n_K \leq n\emm^{-d/3}\), where the last inequality follows w.h.p. from Lemma~\ref{lem:lower-bound-on-kernel-size}. Thus the number of vertices with degree \(\leq (1-d^{-1/3})d\) in \(G\) (and thus in \(C\)) is w.h.p. at most \(n (\emm^{-d/3} + o(1) + \emm^{-\frac 12 d^{1/3}})\), which is at most \(2\emm^{-\frac 12 d^{1/3}}\) for all \(d\) large enough.

    For the second part of the proposition it is enough to note that \(|S_K|\geq |S_G|-(n - n_K) \geq |S_G| - n\emm^{-d/3}\), and all but \(n\emm^{-\frac 12 d^{-1/3}} + o(n)\) of these vertices have degree \(\geq (1-d^{-1/3})d\), and the remaining vertices have degree \(<d\). 

    Recall that kernel is formed from \(C\) by considering the maximal degree-\(2\) subgraph of \(C\) and then contracting the edges, and then my contracting the subtrees. The degree ``lost'' by removing attached trees is exactly twice the number of tree vertices (as for any tree with \(N\) vertices, its total degree is exactly \(2(N-1)\)). Also clearly, the degree lost by contracting edges is twice the number vertices lost this way.
    Thus in total, the degree lost when going from \(C\) to the kernel is at most \(2\emm^{-d/3}\), thus giving that
    \[\deg_K(S_K) \geq (1-d^{-1/3})d|S_G|-dn[\emm^{-d/3}+o(1) + \emm^{-\frac 12d^{-1/3}}]\geq (1-d^{-1/3})d|S_G|-2dn\emm^{-\frac{1}{2}d^{1/3}},\]
    for all \(d\) large enough. To conclude, \(\deg_G(S_G)\geq \deg_K(S_K)\) clearly holds.
\end{proof}

\subsubsection{Expansion of the kernel}

Next, we prove strong expansion properties of the kernel. Due to the absence of degree-\(2\) and degree-\(1\) vertices, the expansion property holds for subgraphs without a lower bound on their size.

To handle expansion of \(O(\log n)\) sets, we will make use of the following result from~\cite{GALANIS2022104894}, which gives expansion properties of small sets w.h.p. over random graphs with fixed degree sequences with a) all degrees between \(3\) and \(n^{1/50}\), and b) sum of degree square being \(O(n)\). Note that in \(K\sim\Kb(n,d)\), w.h.p. all degrees are between \(3\) and \(O(\frac{\log n}{\log \log n})\) (which is at most \(n^{1/50}\) for all \(n\) large enough). Then sum of squares of degrees of \(K\sim\Kb(n,d)\) is dominated by a sum of squares of \(n\) independent \(\Poisson(d)\) r.v.s.
By Chebyshev, with probability \(1-O(1/n)\), that sum is at most \(2(d^2+d)n\), which is (w.h.p.) linear in \(n_K\) (Lemma~\ref{lem:lower-bound-on-kernel-size}), thus we can apply their result.

\begin{lemma}[{\cite[Lemma 15]{GALANIS2022104894}}]\label{lem:small-set-expansion}
    w.h.p. over \(K\sim\Kb(n,d)\), for any connected \emph{kernel} vertex set \(S\subseteq V_K\) with \(|S|\leq\log^2 n_K\), \(|E_K(S,\overline S)|\geq \deg_K(S) / 144\).
\end{lemma}

We note that while~\cite{GALANIS2022104894} gives total-degree-expansion for \textit{all} sets of size \(\leq n_K/2\), the expansion on sets of larger size is worse than we require. To improve on their result, we make use of the weak expansion of \(G\). Observe that we can directly translate the size of the respective cuts.

\begin{observation}\label{obs:cuts-are-preserved}
    For any vertex set of the kernel, \(S_K\subseteq V_K\), let \(S_G\subseteq V_G\) be the vertex set containing: all vertices corresponding to those in \(S_K\), all vertices lying on the expanded edges of \(G_K[S_K]\), and all vertices belonging to the attached trees adjacent to vertices corresponding to \(S_K\) in \(C\). Then \(|E_K(S_K,\overline{S_K})| = |E_G(S_G,\overline{S_G})|\).
\end{observation}
\begin{proof}
    For any \(e\in E_G(S_G,\overline{S_G})\), there is a corresponding \(e'=\{v,w\}\in E_K\) such that \(e\) is an edge on the path corresponding to \(e'\) (as clearly, \(e\) cannot be an edge going to an attached tree, since all edges in a tree attached to vertices of \(S_K\) in \(C\) are in \(S_G\)), and an endpoint of \(e\) must be a kernel vertex -- in particular it is also an endpoint of \(e'\), say \(v\in S_K\). But then if \(e'\not\in E_K(S_K,\overline{S_K})\), it must be the case \(v,w\in S_K\). Then all vertices on the path corresponding to \(e'\) are in \(S_G\), which implies \(e\not\in E_G(S_G,\overline{S_G})\). Contradiction. Thus \(e'\in E_K(S_K,\overline{S_K})\).
    
    Also note that for there cannot be two edges \(e_1\neq e_2\in E_G(S_G,\overline{S_G})\) with the same corresponding kernel-edge: since \(S_G\) is a the respresentation of \(S_K\) in \(C\), both \(e_1,e_2\) must have an endpoint in \(S_K\). But these endpoints cannot be identical, as then \(e_1 = e_2\) (as they are on a single induces path). Thus they must be on the other ends of maximal induced path. But then the both endpoints of the path are in \(S_K\), thus all vertices on this path are in \(S_G\), contradicting that \(e_1,e_2\) are in a cut.

    Therefore \(|E_K(S_K,\overline{S_K})|\geq |E_G(S_G,\overline{S_G})|\).

    For the converse, if \(e\in E_K(S_K,\overline{S_K})\), say \(e = \{v,w\}\), \(v\in S_K\), then the first edge \(e'\) on the induced path corresponding to \(e\) in \(G\) going from \(v\) is in \(E_G(S_G,\overline{S_G})\). This mapping \(e\mapsto e'\) is an injection, as it has a unique inverse (contracting a maximal induced path gives a unique vertex set).
\end{proof}

With these two results in hand, we can prove expansion of the kernel.

\begin{proposition}\label{prop:expansion-of-kernel}
    For all \(d\) large enough, w.h.p. over \(G\sim\G(n,d/n)\), the kernel \(G_K\) is a \emph{\(\tfrac{1}{144}\)-total-degree-expander}, i.e. for \emph{every} set \(S\subseteq V_K\) with \(\deg_K(S) \leq \tfrac 12 \deg_K(G_K)\), \(|E_K(S,\overline S)|\geq \tfrac{1}{144} \deg_K(S)\). 
\end{proposition}
\begin{proof}
    It is sufficient to show that for any connected set of total degree \(\leq \deg_K(G_K)\), the number of edges in the cut is at least \(\tfrac{1}{144}\)-fraction of the total degree. Then to handle arbitrary vertex set, combine the results for its connected components.
    
    By Lemma~\ref{lem:lower-bound-on-kernel-size}, w.h.p. \(n_K \geq n(1-\emm^{-d/3})\), thus w.h.p. \(\log^2 n_K \geq \log n\). Then w.h.p., by contiguity of \(\Kb(n,d)\) and Lemma~\ref{lem:small-set-expansion}, for any \(|S_K|\leq \log n\), \(|E_K(S_K,\overline{S_K})|\geq\deg_K(S_K)/144\).
    Also assume that \(d\) is large enough so that all previously proven w.h.p. results hold, in particular:
    Propositions~\ref{prop:average-degree-of-connected-sets} (with \(\epsilon = \tfrac{1}{10}\)), and~\ref{prop:expansion-of-core}, and~\ref{lem:not-too-many-internal-edges} (with \(\xi = 2/3\)). 
    Further on we assume \(G\) is such that the properties from these results hold. 

    Let \(S_G\subseteq V_G\) be the set as in the Observation~\ref{obs:cuts-are-preserved}. 
    We divide the proof into four cases, based on the size of \(S_K\) and \(S_G\).
    \begin{enumerate}
        \item If \(|S_K|\leq \log n\), then the Lemma~\ref{lem:small-set-expansion} gives \(|E_K(S_K,\overline S_K)|\geq \deg_K(S_K) / 144\).
        
        \item If \(\log n\leq |S_K| \leq |S_G| \leq \tfrac n6\) (note that the middle inequality always holds), then by Proposition~\ref{prop:expansion-of-core} and the above observation we get
        \[
            |E_K(S_K,\overline{S_K})| = |E_G(S_G,\overline{S_G})| \geq \frac{3}{5} \deg_G(S_G) \geq \frac{3}{5}\deg_K(S_K),
        \]
        where the last inequality follows from \(\deg_K(S_K) \leq \deg_G(S_G)\).

        \item If \(\tfrac n6 < |S_G| \leq \tfrac{2n}{3}\), then Lemma~\ref{lem:not-too-many-internal-edges} we gives \(e_G(S_G) \leq (\tfrac 49d + 1)|S_G|\), and Proposition~\ref{prop:average-degree-of-connected-sets} gives \(\deg_G(S_G)\geq \tfrac{9d}{10}|S_G|\). Thus
        \[
        \frac{|E_K(S_K,\overline{S_K})|}{\deg_K(S_K)}\geq \frac{|E_G(S_G,\overline{S_G})|}{\deg_G(S_G)} = 1 - \frac{2e_G(S_G)}{\deg_G(S_G)} \geq 1 - \frac{8}{9} - \frac{10}{9d},
        \]
        where the first inequality follows from \(\deg_K(S_K)\leq\deg_G(S_G)\) and the second from the above inequalities. To finish, note that for all \(d\) large enough, this is at least \(\tfrac{1}{10}\).

        \item Finally, we rule out the case \(|S_G|\geq \frac{2n}{3}\).
        
        By Proposition~\ref{prop:exponentially-few-small-degree-vertices}, \(\deg_K(S_K)\geq (1-d^{-1/3})d|S_G| - 2dn\emm^{-\frac 12 d^{1/3}} \geq nd[1-d^{-1/3}-2\emm^{-\frac 12 d^{1/3}}] \geq 2dn/3\) for all \(d\) large enough. But also (w.h.p.) \(|E_K|\leq |E_G|\leq nd / 2 + o(n)\), so for any such set \(S_K\), \(\deg_K(S_K)\geq\frac{1}{2}\deg_K(G_K)\), hence we can rule this case out.\qedhere
    \end{enumerate}    
\end{proof}

\subsubsection{Proof of Lemma~\ref{lem:giant-component-in-the-core-after-edge-removal}}

Now we are equipped to prove Lemma~\ref{lem:giant-component-in-the-core-after-edge-removal}. To do so, we first prove similar statements for the case of degree-expander graphs. The following Lemma is a modification of \cite[Lemma 2.3]{trevisan2016expanders}.

\begin{lemma}[Analogue of {\cite[Lemma 2.3]{trevisan2016expanders}}]\label{lem:giant-component-after-edge-removal-in-expanders}
    Let \(H=(V_H,E_H)\) be a \(\alpha\)-total degree expander, i.e. for any vertex set \(S\subseteq V\) with \(\deg_H(S)\leq \tfrac12\deg_H(H)\), \(E_H(S,\overline S)\geq\alpha\deg_H(S)\).

    Then, for any \(\theta\in(0,\alpha)\) and any subset of edges \(E'\subseteq E_H\) with \(|E'|\leq \theta |E_H|\), it holds that \((V_H, E_H\setminus E')\) contains a component with total degree at least \(\left(1 - \tfrac{\theta}{2\alpha}\right)\deg_H(V_H)\).
\end{lemma}
\begin{proof}
    Let \(C_1,\dots,C_k\) be components of \((V_H,E_H\setminus E')\) and wlog \(C_1\) is the component of the highest total degree.

    \noindent{\bf Case 1.} \(\deg_H(C_1) < \frac{\deg_H(V_H)}{2}\). Then for all \(1\leq i\leq k\), \(\deg_H(C_i) < \frac{\deg_H(V_H)}{2}\), and hence by Proposition~\ref{prop:expansion-of-kernel} to get \(|E'| \geq \frac{1}{2} \sum_{i=1}^k \alpha \deg_H(C_i) = \alpha |E_H|\). However, \(|E'|\leq \theta |E_H|\leq \alpha|E_H|\). Contradiction, this case cannot happen.

    \noindent{\bf Case 2.} \(\deg_H(C_1) \geq \frac{\deg_H(V_H)}{2}\). Then it follows \(\deg_H(C_2 \cup \dots \cup C_k) \leq \deg_H(V_H) / 2\), and thus we can use the third expansion property to get
    \[|E'|\geq \alpha \deg_H(C_2\cup\dots\cup C_H) = \alpha(\deg_H(V_H) - \deg_H(C_1))\]
    This gives $
    \deg_H(C_1) \geq \deg_H(V_H) - |E'|/\alpha \geq (1 - \frac{\theta}{2\alpha})\deg_H(V_H)$.
\end{proof}

Since w.h.p. the kernel is \(\tfrac{1}{144}\)-total-degree expander, we can apply this lemma for the kernel. We use that the number of edges in \(C\) is a small fraction away from the number of kernel edges, and hence get Lemma~\ref{lem:giant-component-in-the-core-after-edge-removal}, which we restate for convenience.

\giantafterremoval*
\begin{proof}[Proof of Lemma~\ref{lem:giant-component-in-the-core-after-edge-removal}]
    Let \(E'_K\) be the set of edges in the kernel whose corresponding induced path in \(G\) contains \textit{at least one} edge from \(E'\). Notice this does not increase the number of edges, i.e. \(|E'_K|\leq |E'|\). Also, from Lemma~\ref{lem:not-too-many-expanded-edges-overal}, \(|E_C|\leq (1+\emm^{-d/3})|E_K|\), thus
    \[
    |E'_K|\leq \frac{|E_K|}{144}(2\theta-2\emm^{-d/3})(1+\emm^{-d/3}) < \frac{|E_K|}{144}.
    \]
    Note that for all \(d\) large enough, \(\tfrac{1}{100}\leq 2\emm^{-d/3}\), thus \(|E'|\leq \frac{|E_C|}{144}(2\theta-2\emm^{-d/3})\).
    Hence we can apply Lemma~\ref{lem:giant-component-after-edge-removal-in-expanders} to get that \((V_K,E_K\setminus E'_K)\) contains a giant component of total degree at least
    \[
    (1 - (\theta - \emm^{-d/3}))(1+\emm^{-d/3}))\deg_K(V_K) \geq \frac{1 - (\theta - \emm^{-d/3})(1+\emm^{-d/3})}{1+\emm^{-d/3}} \deg_G(V_C)\geq (1-\theta)|\deg_G(V_C).
    \]
    To finish, we note that vertices of this component are also in the same component of \((V_C,E_C\setminus E')\), their degrees do not decrease, and hence there is a component of \((V_C,E_C\setminus E')\) with size \(\geq (1-\theta)\deg_G(V_C)\). The lemma follows.
\end{proof}

\section{Weak spatial mixing within the disordered phase}\label{sec:disordered}

In this section, we prove WSM within the disordered phase at a distance \(r\sim \frac{\log n}{d}\). We make use of \cite[Lemma 5.3]{blanca2021random}, which says that in order to bound \(|\pi^\dis_C(1_e) - \pi_{B_r^-(v)}(1_e)|_\TV\), it is enough to bound the event that there is component of in-edges of size \(\geq r-1\). To capture the probability of this event we use \textit{disordered polymers} (analogously to \cite{RCM-Helmuth2020}).
\begin{definition}
    For \(F\in\Omega^\dis_C\), the \emph{disordered polymers} of \(F\) are \emph{connected components} of \((V_C,F)\). Let \(\Gamma^\dis(F)\) be the set of all disordered polymers of \(F\).
    For a disordered polymer \(\gamma = (U, B)\in\Gamma^\dis(F)\), its \textit{weight} is defined as $  w^\dis_C(\gamma) := q^{1-|U|} (\emm^\beta-1)^{|B|}.$
\end{definition}

Recall that \(\beta_1 = \beta_1(q)\) is the temperature satisfying \(\emm^{\beta_1}-1 = q^{(2+1/10)/d}\). We next show that when \(\beta \leq \beta_1\), then weight of all sufficiently large polymers is exponentially decaying in their size. 

\begin{lemma}\label{lem:disordered-polymer-weight-decay}
    There exists a real \(A > 0\), such that for all \(d\) sufficiently large, w.h.p. over \(G\sim\G(n,d/n)\) the following holds for all real parameters \(0 <q, 0 < \beta \leq \beta_1\). Let \(F\in\Omega_C^\dis\), and \(\gamma\in\Gamma^\dis(F)\). If \(|V_G(\gamma)|\geq\frac{A\log n}{d}\), then also \[w^\dis_C(\gamma) \leq q^{1-|V_G(\gamma)|/10}.\]
\end{lemma}
\begin{proof}
    We divide the proof into two cases. First consider the case \(|V_G(\gamma)| \geq \frac{n}{6}\). Then also note that since \(\gamma\in\Gamma^\dis(F)\), \(|e_G(\gamma)|\leq |F| \leq \eta |E_C| \leq 3\eta dn / 4\) (the last inequality follows w.h.p. over \(G\) from \(|E_C|\leq|E_G| = nd/2 + o(n)\)).

    Thus, using that \(\emm^\beta-1 \leq q^{(2+1/10)/d}\),
    \[
    w^\dis_C(\gamma)\leq q^{1-n/6 + \frac{2+1/10}{d} (3\eta dn / 4)} \leq q^{1-n[1/6 - 2\eta ]}.
    \]
    Since \(\eta < 1/30\), \(1/6 - 2\eta \geq 1/10\). Thus \(w^\dis_C(\gamma)\leq q^{1-n/10}\leq q^{1-|V_G(\gamma)|/10}\), where the last inequality follows from \(|V_G(\gamma)|\leq n\).

    Now let \(A\) be the constant from Lemma~\ref{lem:not-too-many-internal-edges} (with \(\xi = n/6\)) and thus consider the case \(\frac{A\log n}{d}\leq|V_G(\gamma)| < n/6\). By the Lemma, w.h.p. it holds that any such \(\gamma\) has \(e_G(V_G(\gamma))\leq(d/9+1)|V_G(\gamma)|\).    Since we have \(e_G(\gamma) \leq e_G(V_G(\gamma)) \leq (d/9 + 1)|V_G(\gamma)|\), we combine the inequalities to obtain:
    \[
    w^\dis_C(\gamma)\leq q^{1-|V_G(\gamma)|(1 - \frac{21}{10d}(d/9+1))} \leq q^{1-|V_G(\gamma)|(23/30-21/10d)}.
    \]
    For all \(d\) large enough this is at most \(q^{1-|V_G(\gamma)|/10}\).
\end{proof}


Next we show that polymer weight is an upper bound on the the probability of it occurring in the disordered phase.

\begin{lemma}
    For all connected graphs \(C\), the following holds. For any \(F\in\Omega^\dis_C\),
    \[
w_C(F)=q^{n_C}\prod_{\gamma\in\Gamma^\dis(F)} w^\dis_C(\gamma),
    \]and for any \(\gamma\in\bigcup_{F\in\Omega_C^\dis}\Gamma^\dis(F)\), it holds that $\pi^\dis_C(\gamma\in\Gamma^\dis(\cdot))\leq w^\dis_C(\gamma).$
\end{lemma}
\begin{proof}
    For the first part, expand out the rhs using polymer weight definition, 
    \[
    q^{n_C}\prod_{\gamma\in\Gamma^\dis(F)} w^\dis_C(\gamma) = q^{n_C+|\Gamma^\dis(F)| - \sum_{\gamma\in\Gamma^\dis(F)} |V_G(\gamma)|} (\emm^\beta-1)^{\sum_{\gamma\in\Gamma^\dis(F)}  e_G(\gamma)}
    \]
    Disordered polymers of \(F\) are exactly the connected components on \((V_C,F)\), thus \(|\Gamma^\dis(F)| = c(F)\), \(\sum_{\gamma\in\Gamma^\dis(F)} |V_G(\gamma)| = n_C\). Every in-edge of \(F\) is in a polymer and polymers of \(F\) do not overlap, thus also \(\sum_{\gamma\in\Gamma^\dis(F)} e_G(\gamma) = |F|\). Hence the the first part of the lemma follows.

    Next, notice that any disordered polymer containing no edges has weight \(1\). So for any \(\gamma = (\{v\}, \emptyset)\), the inequality in the lemma is trivial. So suppose \(\gamma\) is a polymer with at least one edge. Then for any \(F\in\Omega^\dis_C\) with \(\gamma\in\Gamma^\dis(F)\), let \(F' := F\setminus E_G(\gamma)\) and note \(F'\neq F\). The components of \((V_C, F')\) are exactly the components of \((V_C, F)\), except for the vertices of \(\gamma\), that are in \(\gamma\) in \((V_C,F)\), but in isolated vertices in \((V_C,F')\). Thus \(w_C(F') = w_C(F) q^{|V_G(\gamma)| - 1} (\emm^\beta-1)^{-e_G(\gamma)} = w_C(F) / w^\dis_C(\gamma)\).

    Let $\Omega^\dis_C(\gamma)=\big\{F\in\Omega^\dis_C\mid  \gamma\in\Gamma^\dis(F)\big\}$. Then we can  write 
    \[
    \pi^\dis_C(\gamma\in\Gamma^\dis(\cdot)) = \frac{\sum_{F\in\Omega^\dis_C(\gamma)} w_C(F)}{\sum_{F\in\Omega^\dis_C} w_C(F)} \leq \frac{\sum_{F\in\Omega^\dis_C(\gamma)} w_C(F)}{\sum_{F\in\Omega^\dis_C(\gamma)} (w_C(F) + w_C(F'))} \leq \frac{1}{1+1/w^\dis_C(\gamma)}\leq w^\dis_C(\gamma).\qedhere
    \]
\end{proof}

Now, we can bound that with good probability, there are no large polymers.

\begin{proposition}\label{prop:disordered-kotecky-preiss}
    There exists a real \(A > 0\), such that for all sufficiently large \(d\), for all sufficiently large \(q\), w.h.p. over \(G\sim\G(n,d/n)\) the following holds for all \(0 < \beta \leq \beta_1\).
    \[\pi^\dis_C(\exists\gamma\in\Gamma^\dis(\cdot)\text{ with }|V_G(\gamma)|\geq\tfrac{A\log n}{d}) \leq q^{-A\log n/(25d)}
    \]
\end{proposition}
\begin{proof}
    By Lemmas~\ref{lem:disordered-polymer-weight-decay} and~\ref{lem:not-too-many-internal-edges}, there is \(A > 0\) such that for all \(d\) large enough, w.h.p. \(G\) is such that for all \(\beta\in(0,\beta_1]\), any polymer \(\gamma\) with at least \(\frac{A\log n}{d}\) vertices has weight at most \(q^{1-|V_G(\gamma)|/10}\), and also for any \textit{connected} vertex set \(S\subseteq V_C\) of size at least \(\frac{A\log n}{d}\) has \(e_G(S)\leq (2d/3+1)|S|\). Also by union bound and Markov inequality applied to Lemma~\ref{lem:expected-number-of-subtrees}, w.h.p. \(G\) is such that for each \(1\leq s\leq n\) there are at most \(n^{5/2}(\emm d)^s\) connected subsets of size \(s\). Further on, assume \(G\) is such that the above holds.    By the previous lemma and by union bound, it is sufficient to prove that    
    \[
\sum_{\substack{\gamma\in\bigcup_{F\in\Omega_C^\dis}\Gamma^\dis\\|V_G(\gamma)|\geq A\log n/d}} w^\dis_C(\gamma) \leq  q^{-A\log n/(25d)}.
    \]
    For any connected subset \(S\) of the size \(\geq\frac{A\log n}{d}\), \(e_G(S)\leq (2d/3+1)|S|\), hence there are at most \(2^{(2d/3+1)|S|}\) possible connected subgraphs of \(G\) with the vertex set \(S\). Hence the number of polymers with \(s\geq\frac{A}{d}\log n\) vertices is at most \(n^{5/2}(\emm d)^s 2^{(2d/3+1)s}\). 
    By Lemma~\ref{lem:disordered-polymer-weight-decay}, for any polymer with \(s\geq\frac{A}{d}\log n\) vertices, its weight is at most \(q^{1-s/10}\). Thus summing through all possible sizes gives us that the the sum is at most

    \[
    \sum_{s\geq (A/d)\log n} n^{5/2}(\emm d)^s 2^{(2d/3+1)s} q^{1-s/10} = qn^{5/2} \sum_{s\geq (A/d)\log n} [2^{2d/3+1}edq^{-1/10}]^s.
    \]
    Finally, provided that \(q\) is sufficiently large, \(2^{2d/3+1}edq^{-1/10} \leq q^{-1/20}\), and for \(q\) large enough, the sum can be also bound by \(O(n^{5/2} q^{-\frac{A \log n}{d} / 20}) \leq  q^{-A\log n/(25d)}\).
\end{proof}

Now we can prove WSM within the disordered phase.

\begin{proposition}\label{prop:wsm-disorder}
    There exists a real \(A > 0\), such that for all \(d\) large enough, and then for all \(q\) large enough reals, w.h.p. over \(G\sim\G(n,d/n)\), \(C=C(G)\) has WSM within the disordered phase at distance \(\lceil \frac{A\log n}{d}\rceil\) with rate \(q^{-1/30}\).
\end{proposition}
\begin{proof}
    By Proposition~\ref{prop:disordered-kotecky-preiss} and Theorem~\ref{thm:phase-transition}(iii), there exists a real \(A > 0\) such that for all \(d\) and \(q\) large enough, w.h.p. \(G\) is such that \(\pi^\dis_C(\exists\gamma\in\Gamma^\dis(\cdot)\text{ with }|V_G(\gamma)|\geq\tfrac{A\log n}{d}) \leq  q^{-A\log n/(25d)} \) and also \(\pi^\dis_C(|\cdot|\geq\tfrac{9\eta}{10}|E_C|)\leq \emm^{-n}\). If \(d\) is sufficiently large, then for all \(n\) large enough, \(\frac{A\log n}{2d}+1\leq \frac{A\log n}{d}\).
    Let \(r := \lceil \tfrac{A\log n}{2d}\rceil\) and note that $r<\frac{A\log n}{d}$ for large enough $d$. Fix \(e\in E_C\) and a vertex \(v\) incident to \(e\).

    Then for a subset \(F'\subseteq E_C\setminus E_G(B_r(v))\), let \(\pi_{F',B}\) be the Gibbs distribution with \textit{boundary condition} \(F'\) on \(B_r(v)\), i.e. the Gibbs distribution conditional on that in-edges in the configuration outside the ball \(B_r(v)\) are \textit{exactly} the edges in \(A\).

    Let \(\Upsilon\) be the event that there is a path from a neighbour of \(v\) to a vertex distance-\(r\) from \(v\).  \cite[Lemma 5.3]{blanca2021random} gives that for any \(F'\subseteq E_C\setminus E_C(B_r(v))\)  it holds that\footnote{Note that while technically their result is stronger and \(\Upsilon\) could be improved to an event roughly implying existence of two such paths, this is not needed for our proof.}
    \[
    |\pi_{F',B_r(v)}(1_e) - \pi_{B_r^-(v)}|_\TV\leq \pi_{F',B_r(v)}(\Upsilon).
    \]
        By monotonicity of random cluster \(\pi_{F',B_r(v)}(1_e) \leq \pi_{B_r^{-}(v)}(1_e)\), see e.g. \cite[Corollary~53]{Galanis_Goldberg_Smolarova_2024}. Then note that for any \(A\subseteq E_C\setminus E_C(B_r(v))\) with \(|F'|\leq \eta |E_C| - |E(B_r(v))|\), \(\pi_{F',B_r(v)}\) is the same distribution as the disordered distribution conditional on that the configuration outside the ball is \(F'\) (since any configuration agreeing with such \(F'\) on the outside the ball is disordered).
    By Lemma~\ref{lem:ball-size}, w.h.p. \(G\) is such that for any such \(B_r(v)\), \(|E_G(B_r(v))|\leq d^r \log n \leq \tfrac{\eta}{n}|E_C|\) (where the last inequality holds w.h.p., as Lemma~\ref{lem:bound-on-kernel-edges} gives that \(|E_C|=\Omega(n)\)). Hence by Theorem~\ref{thm:phase-transition}(1), for all \(d\) and \(q\) large enough, w.h.p. \(G\) is such that we have that \(\pi^\ord(|\cdot\cap (E_C\setminus E_C(B_r(v))|\geq \tfrac{9\eta|E_C| }{10}) \leq \pi(|\cdot|\geq \tfrac{9\eta|E_C|}{10})\leq \emm^{-n}\).
    
    Then, using triangle inequality, we get that $|\pi^\ord(1_e) - \pi_{B^-_r(v)}(1_e)|$ is at most
    \begin{align*}
        &\sum_{F'\subseteq E_C\setminus E_G(B_r(v))} |\pi^{\dis}_C(1_e \mid F' = \cdot\cap (E_C\setminus E_G(B_r(v)))) - \pi_{B^-}(1_e)|\times  \pi^\ord(\cdot\cap (E_C\setminus E_G(B_r(v))) = F') \\
        \leq & \pi^\ord(|\cdot\cap (E_C\setminus E_G(B_r(v)))|\leq \tfrac{9\eta}{10}|E_C|) + \sum_{F'\subseteq E_C\setminus E_G(B_r(v))} \pi_{F',B_r(v)}(\Upsilon)  \pi^\ord(F'=\cdot\cap (E_C\setminus E_G(B_r(v)))) \\
        \leq & \emm^{-n} + \pi^\dis_C(\Upsilon).
    \end{align*}

    To finish, note that \(\pi^\dis_C(\Upsilon)\) is at most the probability of there being a disordered polymer of size \(\geq r-1\). By the choice of \(r\) and the Proposition~\ref{prop:disordered-kotecky-preiss}, the probability is at most \(e^{-n} + q^{-A\log n/(25d)}\leq q^{-A\log n/(30d)} = q^{-r/30}\) (for all \(q\) large enough. Hence we have WSM within the disordered phase.
\end{proof}

\section{Proof of Theorem~\ref{thm:phase-transition}}\label{sec:phase-transition}

In this section, we prove Theorem~\ref{thm:phase-transition} using the following five lemmas.
First we show that the ordered phase vanishes for all \(\beta\leq\beta_0\).

\begin{lemma}\label{lem:ordered-phase-vanishes}
    Let \(d\) be a large enough real. For all sufficiently large real \(q\), w.h.p. over \(G\sim\G(n,d/n)\), it holds that \(Z^\ord_C / Z_C = \emm^{-n}\) for all \(\beta \leq \beta_0\).
\end{lemma}
\begin{proof}
    By Lemmas~\ref{lem:avg-deg-of-medium-sets}, and~\ref{lem:avg-deg-of-large-sets}, w.h.p. all vertex sets in \(G\) with size at least \(\geq \eta n\) have average degree is at least \(\tfrac{3d}{4}\). Also by the second part of the first Lemma, w.h.p., for any vertex set \(S\subseteq V_G\) with size between \(\eta n\) and \((1-\eta)n\), \(|E_G(S,\overline S)|\geq \frac{5d}{6n}|S|(n-|S|)\). Also w.h.p. \(n_C \geq (1-\emm^{-d/3})n\).
    Assume further on \(G\) satisfies these properties.
    
    Also by the choice of \(\eta\) and by Lemma~\ref{lem:giant-component-in-the-core-after-edge-removal}, w.h.p. \(G\) is such that for any \(F\in\Omega^\ord_C\), there is a giant component in \((V_C,F)\) with total degree \(\geq \tfrac{9}{10} \deg_G(V_C)\). Thus the set \(S\) of vertices in the non-giant components has total degree \(\leq\tfrac{1}{10}\deg_G(V_C)\). Note that all edges in \(E_G(S,\overline S)\) are in \(E_C\setminus F\).
    We consider three cases, based on the size of \(S\).
    
    If \(|S|\geq\frac{n}{3}\), then \(\deg_G(S)\geq \frac{dn}{4}\). But also \(\deg_G(S)\leq \tfrac{\deg_G(V_C)}{10} \leq \tfrac{dn}{10} + o(n)\), which is a contradiction.

    Next consider the case \(\eta n \leq |S|\leq \tfrac n3\). Then \(|E_C\setminus F|\geq |E_G(S,\overline{S})| \geq \frac{5\eta n d}{9}\). But w.h.p. \(|E_C|\leq \frac{dn}{2} + o(n) < \frac{5}{9}dn\), so this contradicts \(|E_C\setminus F|\leq \eta |E_C|\).

    Thus the only remaining possibility is \(|S|\leq\eta n\). Then \(c(F)\leq |S| \leq \eta n + 1\) and \(|F|\leq |E_C|\leq nd/2 + \eta n\) (the last inequality holds w.h.p.). Thus we obtain that \(w_C(F)\leq q^{\eta n + 1}(\emm^\beta-1)^{nd / 2 + \eta n} \leq q^{\eta n + \frac{2-1/10}{d}\cdot\frac{nd}{2} + \eta n} = q^{n[1 - 1 / 20 + 2\eta]}\). Note that the second inequality holds as \(1+\frac{19}{10d}\eta n \leq \eta n\) for all \(d\geq 3\) and \(n\) large enough.

    Since \(\emptyset\) is a configuration with a weight \(q^{n_C} \geq q^{n(1-\emm^{-d/3})}\), it follows that \(Z_C\geq q^{n-n\emm^{-d/3}}\). There are at most \(2^{|E_C|} = 2^{nd/2 + \eta n}\) ordered configurations, therefore,
    \[
    \frac{Z^\ord_C}{Z_C}\leq 2^{nd / 2 + \eta n} q^{n[-1 + 1 - 1/ 20 + 2\eta + \emm^{-d/3}]} = \exp\{n[(d/2 + \eta)\ln 2 + \log q(- 1/ 20 + 2\eta + \emm^{-d/3})]\}.
    \]
    Since \(1/20 - 2\eta - \emm^{-d/3} > 1/30\) for all \(d\) large enough, this becomes \(\exp\{n d\ln 2 / 2 - n\log q/30\}\), which is \(\leq \emm^{-n}\) for all \(q\) large enough.
\end{proof}

Analogously, the contribution of the disordered phase becomes negligible for \(\beta \geq \beta_1\).

\begin{lemma}\label{lem:disordered-phase-vanishes}
    Let \(d\) be a large enough real. For all sufficiently large real \(q\), w.h.p. over \(G\sim\G(n,d/n)\) it holds that \(Z^\dis_C / Z_C = \emm^{-n}\) for all \(\beta \geq \beta_1\).
\end{lemma}
\begin{proof}
    Let \(F\in\Omega^\dis_C\). Then \(|F|\leq\eta |E_C|\), and thus \(w_C(F)\leq q^{n_C} (\emm^\beta-1)^{\eta |E_C|}\). \(|\Omega^\ord_C| \leq 2^{|E_C|}\), and since \(F=E_C\) is a configuration, \(Z_C\geq q(\emm^\beta-1)^{|E_C|}\). w.h.p. \(|E_C|\leq \frac{dn}{2} + \eta n\), hence \(|\Omega^\ord_C|\leq 2^{n(d/2 + \eta)}\). Combining the inequalities and using that \(\beta\geq\beta_1\),
    \[
    \frac{Z^\dis_C}{Z_C}\leq 2^{(d/2+\eta)n} q^{n}(\emm^\beta-1)^{-(1-\eta)|E_C|} \leq q^{n - \frac{(1-\eta)(2+1/10)}{d}|E_C|}.
    \]
    Lemma~\ref{lem:bound-on-kernel-edges} gives us that w.h.p. \(|E_C|\geq\frac{dn}{2}-2n\emm^{-d/3}\), thus we get that
    \[
    \frac{Z^\dis_C}{Z_C} \leq q^{n[1-(1-\eta)(2+1/10)(d/2 - 2\emm^{-d/3})/d]} \leq q^{-n[1/20 - 2\eta+5\emm^{-d/3}/d]}
    \]
    By the choice of \(\eta\), for all \(d\) large enough this is at most \(q^{-n/40}\). Hence for all \(q\) large enough, it is at most \(\emm^{-n}\).
\end{proof}

It remains to show that the configurations with \(\leq(1-\tfrac{9\eta}{10})|E_C|\) and \(\geq \tfrac{9\eta}{10}|E_C|\) edges are very unlikely. First we show this for the Gibbs distribution without conditioning.

\begin{lemma}\label{lem:error-vanishes}
    Let $d$ be a sufficiently large real. For all sufficiently large real $q$, the following holds w.h.p. over \(G\sim\G(n,d/n)\), for any  real \(\beta > 0\).
    
    Let $\mathcal{B}$ be the set of configurations $F\in \Omega_C$ with $\tfrac{9}{10}\eta\leq \frac{|F|}{|E_C|}\leq 1-\tfrac{9}{10}\eta$.  Then, $\pi_C(\mathcal{B})\leq \emm^{-n}$.
\end{lemma}
\begin{proof}
    The main part of the proof is to show  that for any configuration \(F\in\mathcal{B}\) it holds that
    \begin{equation}\label{eq:lemma-8-equivalent}
        c(F) - \frac{2}{d}|E_C\setminus F| \leq - \rho n \mbox{ with } \rho=\eta/20.
    \end{equation}
    Let us conclude the proof of the lemma assuming \eqref{eq:lemma-8-equivalent} for the moment. We first show that \eqref{eq:lemma-8-equivalent} translates into $w_C(F)/Z_C\leq q^{-\rho n}$ for all $\beta>0$. Indeed, if \(\emm^\beta - 1 \leq q^{2/d}\), we will use the lower bound \(|Z_C|\geq q^{n_C}\geq q^{n(1-\emm^{-d/3})}\), where the second inequality holds w.h.p.. Thus, for any \(F\in\mathcal{B}\),
    \[
    \frac{w_C(F)}{Z_C} = q^{c(F) - n + n\emm^{-d/3}} (\emm^{\beta}-1)^{|E_C|-|E_C\setminus F|} \leq q^{c(F) - \frac{2}{d}|E_C\setminus F| - n + \frac{2}{d}|E_C| + n\emm^{-d/3}}.
    \]
    Using that w.h.p. \(|E_C|\leq |E_G| \leq nd / 2 + \frac{\eta n}{40}\) we have that \(\frac{2}{d}|E_C| - n \leq n - n + \frac{\eta n}{40}\). Hence, for all \(d\) sufficiently large (recall that $\eta$ is an absolute constant independent of $d$), we obtain that $w_C(F)/Z_C \leq q^{-\rho n/4}$. Analogously, if \(\emm^\beta - 1 \geq q^{2/d}\), we use the lower bound \(Z_C\geq (\emm^\beta-1)^{|E_C|}\) so that for any   \(F\in \mathcal{B}\) 
    \[
    \frac{w_C(F)}{Z_C} = q^{c(F)} (\emm^\beta-1)^{-|E_C \setminus F|} \leq q^{c(F) - \frac{2}{d} |E_C \setminus F|}\leq q^{-\rho n}.
    \]
    To finish, note that the size of $\mathcal{B}$   is trivially at most \(|\Omega_C|\leq 2^{|E_C|}\leq 2^{nd/2 + \eta n/40}\). Thus we get \[\frac{1}{Z_C}\sum_{F\in\mathcal{B}} w_C(F) \leq 2^{nd/2 + \eta n / 40}q^{-\rho n/4}\]
    which is \(\leq \emm^{-n}\) for all \(q\) large enough.

We next return to the proof of~\eqref{eq:lemma-8-equivalent}. Fix an arbitrary $F\in \mathcal{B}$. We partition the vertices of \(C\) into sets $S_1,S_2,S_3,S_4$ depending on the size of their components in \((V_C,F)\). Namely, we set $S_1$ to be the set of vertices in singleton components, $S_2$  the set of vertices in components of size between \(2\) and \(\log n\), $S_3$ the set of vertices in components of size \(> \log n\) and \(\leq \tfrac{n}{6}\), and $S_4$  the set of vertices in components of size \(>\tfrac n6\).   For \(i = 1,\dots,4\), let \(n_i := |S_i|\). We have then the trivial upper bound 
    \begin{equation}\label{eq:rvrtvta}
        c(F)\leq n_1 + \frac{n_2}{2} + \frac{n_3}{\log n} + 6.
    \end{equation}
Next we give a lower bound on $|E_C\backslash F|$ by analysing how many out-edges are incident to vertices in $S_1,S_2,S_3,S_4$. 
    \begin{enumerate}
        \item Consider first components of size 1. Every edge adjacent to any $v\in S_1$ must be in \(E_C\setminus F\). Thus, aggregating over $v\in S_1$, we have at least  \(\tfrac{1}{2}\deg_G(S_1)\) edges in \(E_C\setminus F\) adjacent to vertices in $S_1$; note that the $\tfrac{1}{2}$-factor accounts for potential double-counting (edges that join two vertices in $S_1$).

        \item Let \(k\) be a constant such that w.h.p. all sets of size \(\leq \log n\) are \(k\)-treelike (Lemma~\ref{lem:logn-sets-are-treelike}). Also by Lemma~\ref{lem:1-treelikeness-for-constant-size}  w.h.p. any connected set of size \(\leq 2k\) is \(1\)-treelike.

        Consider now any component $\kappa$ whose size is between $2$ and $\log n$. If $\kappa$ has size \(<2k\), then it is 1-treelike, so \(e_G(\kappa)\leq |V_G(\kappa)|\). This gives
        \[|E_G(V_G(\kappa),\overline{V_G(\kappa)})|\geq \deg_G(\kappa) - 2|V_G(\kappa)|.\]
        If \(\kappa\) has size between \(2k\) and \(\log n\),  the it is $k$-treelike, so \(e_G(\kappa) \leq k + |V_G(\kappa)| \leq \tfrac32|V_G(\kappa)|\), yielding that       
        \[|E_G(V_G(\kappa),\overline{V_G(\kappa)})| \geq \deg_G(\kappa) - 3|V_G(\kappa)|.\] 
        Aggregating over all components $\kappa$ in $S_2$, we therefore obtain at least $\tfrac{1}{2}(\deg_G(S_2)-3n_2)$ edges in $E_C\backslash F$ that are adjacent to vertices in $S_2$. (Once again, the factor $\tfrac{1}{2}$ accounts for double-counting).
        
        \item Consider next components with size between \(\log n\) and \(\tfrac n6\). By Proposition~\ref{prop:expansion-of-core} and taking $d$ large enough (namely, $d>A$ where $A$ is the absolute constant there in), we have that  w.h.p. every connected  subset \(S\)  in $G$ with size $\log n\leq |S|\leq \tfrac{n}{6}$ satisfies \(|E_G(S,\overline{S})|\geq\tfrac 35|\deg_G(S)|\). Therefore, for any component $\kappa$ of \((V_C,F)\) in this size range, we have at least $35|\deg_G(\kappa)|$ edges going out of it, and therefore belong to $E_C\backslash F$. Aggregating over $\kappa$ (and accounting once again for potential double-counting), we obtain that the total number of out-edges adjacent to vertices in \(S_3\) is at least \(\tfrac{1}{2}\cdot \tfrac{3}{5}\deg_G(S_3)\).

        \item Finally consider components with size $>\tfrac{n}{6}$. If \(\eta n/4  \leq |S_4|\leq n(1-\eta / 4) \), we get from the second part of Lemma~\ref{lem:avg-deg-of-medium-sets} that for all \(d\) large enough \(|E_G(S_4,\overline S_4)|\geq \frac{4d|S_4|(n-|S_4|)}{5n}\geq\frac{d\eta n}{6}\) (since \(\eta/4\leq \tfrac{1}{30}\)).  Therefore, the total number of out-edges adjacent to vertices in \(S_4\) is at least $\frac{d\eta n}{6}\1\{\eta n/4\leq |S_4|\leq (n-\eta n/4)\}$. (Note that if $|S_4|$ is outside the range $\eta n/4$ and  $n(1-\eta / 4)$ we use the trivial lower bound 0.)
    \end{enumerate}
    Combining the above, we conclude that
    \begin{align}
    |E_C\backslash F|&\geq \tfrac{1}{2}\big[\deg_G(S_1\cup S_2) - 3n_2 + \tfrac{3}{5}\deg_G(S_3) + \tfrac{d\eta n}{6}\1\{\eta n/4\leq |S_4|\leq (1-\eta n/4)\}\big]\nonumber\\
    &\geq \tfrac{1}{2}\big[(1-d^{-1/3})d (n_1+n_2+\tfrac{3}{5}n_3-2n \emm^{-\frac12d^{1/3}})+\tfrac{d\eta n}{6}\1\{\eta n/4\leq |S_4|\leq (1-\eta n/4)\}\big],\label{eq:rvrtvtb}
    \end{align}
    where the second inequality follows from the fact that the number of  vertices of $G$ which have degree $<(1-d^{-1/3})d$ is less than $2n \emm^{-\frac12d^{1/3}}$ (Proposition~\ref{prop:exponentially-few-small-degree-vertices}).
    From \eqref{eq:rvrtvta} and \eqref{eq:rvrtvtb}, we therefore have that
    \begin{multline*}
        c(F) - \tfrac{2}{d}|E_C\setminus F| \leq 
        n_1 + (\tfrac12+\tfrac 3d) n_2 + \frac{n_3}{\log n} + 6 - (1-d^{-1/3})\left[ n_1 + n_2 + \tfrac{3}{5}n_3 \right] \\ + 2n\emm^{-\frac{1}{2}d^{1/3}} - \tfrac{\eta n}{6}\1\{\eta n/4\leq |S_4|\leq (n-\eta n/4)\}.
    \end{multline*}
    For \(d\) and $n$ sufficiently large, we  have that $\tfrac 12-\tfrac3d - d^{-1/3}\geq \tfrac 14$ and \(\tfrac 35 - \frac{1}{\log n} - \frac{3d^{-1/3}}{5} \geq \tfrac 14\), so     \begin{equation}\label{eq:lemma-8-equiv-upper-bound-1}
        c(F)-\tfrac 2d|E_C\setminus F|\leq d^{-1/3}n_1 - \frac{n_2+n_3}{4} - \frac{\eta n}{6}\1\{\eta n/4\leq |S_4|\leq (n-\eta n/4)\} + 2n\emm^{-\frac{1}{2}d^{1/3}}+6.      
    \end{equation}

    Now we split the analysis into four cases. We will use that w.h.p. \(|E_C|\geq |E_K|\geq \tfrac{dn}{2}-2n\emm^{-d/3}> \frac{25}{27}\tfrac{dn}{2}\) for all \(d\) large enough (Lemma~\ref{lem:bound-on-kernel-edges}) and also that w.h.p. \(|E_C|\leq |E_G|\leq \tfrac{nd}{2} + \tfrac{\eta n}{40}\).
    \begin{itemize}
        \item Case 0: \(\frac{n_1}{n}\geq 1-\frac{2\eta}{3}\).         Here we show that \(|F| < \tfrac{9}{10}\eta|E_C|\), contradicting our assumption on $|F|$ coming from $F\in \mathcal{B}$ (and hence we can rule out this case). Indeed, no edge incident to a vertex in \(S_1\) is in \(F\), so $|E_C\setminus F|\geq \frac{1}{2} \deg_G(S_1) \geq \tfrac12 d(1-d^{-1/3})(n_1-2n\emm^{-\frac12d^{1/3}})$ where the last inequality follows  from the fact that  \(G\) has at most $2n\emm^{-\frac12d^{1/3}}$ vertices with degree \(<(1-d^{-1/3})d\) (Proposition~\ref{prop:exponentially-few-small-degree-vertices}). Hence \[|F| = |E_C|-|E_C\setminus F| \leq \frac{dn}{2} + \frac{\eta n}{40} - \frac{dn}{2}[(1-\tfrac{2\eta}{3})(1-d^{-1/3}) - 2\emm^{-\frac12d^{1/3}}].\]
        For all \(d\) large enough, \(\tfrac{\eta}{40d} + \tfrac{2\eta}{3} + d^{-1/3} +2\emm^{-\frac12d^{1/3}}\leq \frac{5\eta}{6}\), thus also \(|F|\leq\frac{dn}{2}\cdot\frac{5\eta}{6}\). But \(|E_C| > \tfrac{25}{27}\tfrac{dn}{2}\), which gives $|F| < \frac{9}{10}\eta|E_C|$, contradicting $F\in \mathcal{B}$.

        \item Case 1: \(n_4 > n(1-\eta / 2)\). Since we have at most 6 components with size $>n/6$, in this case it follows that \(c(F) < n - n(1-\eta / 2) + 6 = \eta n/ 2 + 6\). Since $|F|\leq (1-\frac{9\eta}{10})|E_C|$ we have \(|E_C\setminus F|\geq\tfrac{9}{10}\eta |E_C| > \tfrac{5\eta}{12} dn\), so for $n$ large enough
        \[
        c(F) - \frac{2}{d}|E_C\setminus F| < 6 + \frac{\eta n}{2} - \frac{5}{12} \eta n = 6 - \frac{\eta n}{12}\leq -\frac{\eta n}{20}.
        \]
        \item Case 2: \(n_2 + n_3 \geq \eta n / 3\).        
        Plugging this into~\eqref{eq:lemma-8-equiv-upper-bound-1} (using the trivial bound \(n_1 < n\) and ignoring the indicator), we obtain that
        \[
        c(F) - \frac{2}{d}|E_C\setminus F|\leq d^{-1/3} n - \frac{\eta n}{12} + 2\emm^{-\tfrac{1}{2}d^{1/3}}+6
        \]
        Since \(\eta\) does not depend on \(d\), for all \(d\) (and $n$) large enough, \(d^{-1/3} +2\emm^{-\tfrac{1}{2}d^{1/3}}+ 6/n \leq \eta / 30\), and so for all \(n\) large enough, yielding once again~\eqref{eq:lemma-8-equivalent}.

        \item Case 3: \(n_2 + n_3 < \eta n / 3\), \(n_1 < n(1-2\eta/3)\), and \(n_4 < n(1-\eta/2)\). Since \(n_1 + n_2 + n_3 + n_4 = |V_C|\), we have that \(n_4\geq n_C - n_1 - n_2 - n_3 \geq n - n\emm^{-d/3} - \frac{\eta n}{3} - n + \tfrac{2\eta n}{3}\) which is at least \(\eta n / 4\) for all \(d\) large enough. In this case,~\eqref{eq:lemma-8-equiv-upper-bound-1} gives for $d$ and $n$ large that
        \[
        c(F) - \frac{2}{d}|E_C\setminus F|\leq d^{-1/3}n - \frac{\eta n}{6}+2\emm^{-\tfrac{1}{2}d^{1/3}}+6\leq -\eta n /20.
        \]
    \end{itemize}
    This finishes the proof of~\eqref{eq:lemma-8-equivalent} and completes the proof of Lemma~\ref{lem:error-vanishes}.
\end{proof}

We now prove analogous lemmas for \(\pi^\dis_C\) and \(\pi^\ord_C\). 

\begin{lemma}\label{lem:error-vanishes-under-disorder}
    Let $d$ be a sufficiently large real. For all sufficiently large real $q$, w.h.p. over \(G\sim\G(n,d/n)\), the following holds for all real \(\beta\leq\beta_1\).

    Let $\mathcal{B}_1$ be the set of configurations $F\in \Omega_C$ with $ |F|\geq \tfrac{9}{10}\eta |E_C|$.  Then, $\pi^{\dis}_C(\mathcal{B}_1)\leq \emm^{-n}$.
\end{lemma}
\begin{proof}
    Consider arbitrary $F\in \mathcal{B}_1$.  As we saw in Case 0 of the proof of Lemma~\ref{lem:error-vanishes},  the number of singleton components in \((V_C,F)\) is at most \((1-\frac{2\eta}{3})n\). Let \(S\) be the set of vertices that are in components of size \(\geq 2\).

    Note that the edges of \(F\) have all endpoints in \(S\), thus \(|F|\leq e_G(S)\). Combining Lemmas~\ref{lem:logn-sets-are-treelike}, \ref{lem:1-treelikeness-for-constant-size} (with \(R\) being the  \(k\) from the previous lemma) and~\ref{lem:not-too-many-internal-edges}, we get for all \(d\) large enough, w.h.p., for any connected subset \(S'\subseteq S\) of \((V_C,F)\), \(e_G(S') \leq |S'|(1+\max(1, d\eta/3))\). Hence, \(e_G(S)\leq |S|(2+d\eta/3)\). Since \(e_G(S)\geq |F|\geq\tfrac{9\eta}{10}\), we must have \(|S|\geq \frac{9\eta|E_C|}{10d(\eta/3 + 2/d)}\). Since each component containing \(|S|\) has at least \(2\) vertices, \(c(F)\leq n-|S| + |S|/2 = n - |S|/2\). Therefore we can bound
    \[
    \frac{w_G(F)}{w_G(\emptyset)} = q^{n-\frac{9\eta|E_C|}{20d(\eta/3+2/d)}-n_C} (\emm^\beta -1)^{\eta |E_C|}\leq q^{n\emm^{-d/3}-\frac{\eta|E_C|}{10d}(\frac{9}{2(\eta/3+2/d)} - 21)},
    \]
    where the last inequality follows from \(\beta\leq \beta_1\) and that w.h.p. \(n_C\geq n(1-\emm^{-d/3})\). For all \(d\) large enough \(\frac{9}{2(\eta/3 +2/d)}-21\geq 1\). By Lemma~\ref{lem:bound-on-kernel-edges}, w.h.p. \(|E_C|\geq \frac{dn}{2}-2n\emm^{-d/3}\geq \frac{dn}{3}\), giving that \( \frac{w_G(F)}{Z_C^\dis}  \leq q^{-n[\eta/20-\emm^{-d/3}]}\). For all \(d\) large enough, this is at most \(q^{-\eta n /30}\). Aggregating over at most \(2^{|E_C|}\leq 2^{dn}\) configurations \(F\), we get that \(\pi^\dis_C(\mathcal{B}_1)\leq \emm^{-n}\) for all \(q\) large enough.
\end{proof}

\begin{lemma}
    Let $d$ be a sufficiently large real. For all sufficiently large real $q$, w.h.p. over \(G\sim\G(n,d/n)\), the following holds for any real \(\beta\geq\beta_0\).

    Let $\mathcal{B}_0$ be the set of configurations $F\in \Omega_C$ with $|F|\leq (1-\tfrac{9}{10}\eta)|E_C|$.  Then, $\pi^{\ord}_C(\mathcal{B}_0)\leq \emm^{-n}$.
\end{lemma}
\begin{proof}
    Consider arbitrary \(F\in \mathcal{B}_0\). First of all, as we showed in the proof of Lemma~\ref{lem:ordered-phase-vanishes}, for all \(d\) large enough, w.h.p. \(G\) is such that for any ordered \(F\), \(c(F)\leq 1+\eta n\). Let \(S\) be the set of vertices in non-giant components of \((V_C,F)\).
    
    We now bound \(w_C(F)/Z_C^\ord\) by considering two cases.  If \(c(F) - 1\leq |S|\leq \eta n/2\), then using \(Z_C^\ord\geq (\emm^\beta-1)^{|E_C|}\) and \(\beta\geq\beta_0\) we get
    \[
    \frac{w_C(F)}{Z_C^\ord}\leq q^{\eta n / 2} (\emm^\beta-1)^{-|E_C\setminus F|}\leq q^{\eta n/2 - \frac{19}{10d}\frac{9\eta}{10}|E_C|}.
    \]
    w.h.p. \(|E_C|\geq \frac{\eta n}{2}-2n\emm^{-d/3}\geq \frac{9}{19}dn\) for all \(d\) large enough, thus giving \(\frac{w_C(F)}{Z_C^\ord}\leq q^{\eta n(1 / 2 - 81/100)}\leq q^{-\eta n/10}\). Taking the aggregate sum over at most \(2^{|E_C|}\leq 2^{dn}\) such configurations, we get that for all \(q\) sufficiently large, \(\pi_C^\ord(\mathcal{B}_0)\leq \emm^{-n}\).
\end{proof}

Now Theorem~\ref{thm:phase-transition} follows directly from these five lemmas. 

\section{Remaining proofs for Theorem~\ref{thm:mainthm2}} \label{sec:remproof2}
In this section, we prove Lemmas~\ref{lem:rationalfun},~\ref{lem:integrate} and~\ref{lem:betainf} from Section~\ref{sec:proofThm2}, thus completing the proof of Theorem~\ref{thm:mainthm2}. We restate them for convenience.
\rationalfun*
\begin{proof}
We argue first about $\pi_{B_r^{+}(v);q,\beta}(1_e)$. We consider first the case of trees, and then explain how to get 1-treelike graphs. So, assume that $T$ is a tree rooted at $v$, and let $L$ be the vertices at the boundary of $B_r(v)$.  For a vertex $u$ in $T$, let $T_u$ be the subtree of $T$ rooted at $u$ (up to the boundary of $B_r(v)$), and let $V_u$ denote its vertices. Moreover, let $Z^+(u)$ be the partition function of the subtree $T_u$ rooted at $u$, where the vertices in $L\cap V_u$ are wired; let also $Z^+_1(u)$ be the contribution to $Z^+(u)$ coming from configurations in $T_u$ where there is a path of in-edges from $u$ to $L\cap V_u$ and set $Z^+_0(u)=Z^+(u)-Z^+_1(u)$. 

Note that $\pi_{B_r^{+}(v);q,\beta}(1_e)$ equals $\hat Z^+(v)/Z^+(v)$ where $\hat Z^+(v)$ is the analogue of $Z^+_v$ for the tree $T'$ obtained by contracting the edge $e$ (onto $v$). So, it suffices to explain how to compute $Z^+(u)$ for an arbitrary $u\in T$. Let $u_1,\hdots, u_k$ be the children of $u$. Then, with $t:=\frac{\emm^\beta-1}{q}+1$ the variables $Z^+_0,Z^+_1$ defined above satisfy the following, see e.g. \cite[Fact 3.2]{rcm-on-unbounded-degree-graphs} or \cite[Lemma 33]{BGRC}:
\[
Z^+_0(u)=q \prod^k_{i=1}\frac{Z^+_1(u_i)+tZ^+_0(u_i)}{q}-q\prod^k_{i=1}\frac{Z^+_1(u_i)+tZ^+_0(u_i)}{q},\quad Z^+_0(u)=q^2 \prod^k_{i=1}\frac{Z^+_1(u_i)+tZ^+_0(u_i)}{q}.
\]
Using these, we can therefore compute $Z^+(u)=Z^+_0(u)+Z^+_1(u)$ (as a function of $q$ and $x=\emm^{\beta}-1$) recursively in time $n^{O(1)}$. 

In the case of 1-treelike graphs, we view $B^+_r(v)$ as a cycle $C$ with hanging trees from its vertices (whose leaves are wired); note that the cycle has length at most $2r$. To find the partition function of $B_r(v)$ with wired boundary, we first use the tree-algorithm from above for each of the hanging trees to compute $Z^+_0(u), Z^+_1(u)$ for each vertex $u$ of the cycle. Then, we consider all possible edge-configurations on the cycle (of which there are at most $2^{2r}$), and for each one of them we consider all the vertex-to-boundary possibilities for the vertices $u$ on the cycle, i.e., whether $u$ is connected by a path to the boundary of its hanging tree (of which there are at most $2^{2r}$). For each of these edge configurations and vertex-to-boundary settings, we take the product of corresponding variables from $\{Z^+_0(u),Z^+_1(u)\}_{u\in C}$ (divided by $q^{|C|}$ to avoid double-counting the wired component)  multiplied by the number of in-edges on the cycle and the number of new components formed; aggregating all of these gives us the partition function of $B_r(v)$ in time $2^{4r}n^{O(1)}$.

For $\pi_{B_r^{-}(v);q,\beta}(1_e)$ the argument is much simpler since the partition function for an $n$-vertex tree under free boundary is just $q(\emm^{\beta}-1+q)^{n-1}$ and the above approach can be easily adapted.
\end{proof}
\integrate*
\begin{proof}
Since $\epsilon$ has size $\leq n$, it suffices to produce an additive $1/2^{\Theta(n)}$ approximation to $I$.
 
 Henceforth, we will assume without loss of generality that the degree of $P$ is less than the degree of $Q$. We may also  focus on the case that all roots of $Q(x)$ have multiplicity 1, there are standard symbolic integration techniques to reduce the case with higher-multiplicity roots to the case where all roots have multiplicity-1 in $\text{poly}(n)$ time (e.g. Hermite's reduction, see \cite{integrate} for details). Hence we can write the   partial fraction decomposition of $g(x)=\frac{P(x)}{Q(x)}$ in the form 
 \[g(x)=\sum^{n_0}_{i=1} \frac{a_i}{b_i  x+c_i}+\sum^{n_1}_{i=1} \frac{a_i'x+b_i'}{c_i' x ^2+d_i' x +e_i'}\] for some reals $a_i,b_i,c_i,a_i',b_i',c_i'.d_i',e_i'$ which can be determined using the roots of $Q$ and integers $n_0,n_1\leq n$. Since $Q(x)$ is an integer polynomial which can be represented with $n$ bits,   its roots can be additively approximated with precision $\epsilon'=1/2^{\Theta(n)}$ in time $\text{poly}(n,\log \frac{1}{\epsilon'})=\text{poly}(n)$, see for example the algorithm in \cite{NEFF199681}. Therefore  we can approximate $a_i,b_i,c_i,a_i',b_i',c_i',d_i',e_i'$ with precision  $1/2^{\Theta(n)}$. Note that the exponential dependency of the precision on $n$ is needed to account for negative numbers that might arise in the computations and the magnitude of the coefficients of $P(x),Q(x)$.  The exponential precision will also ensure that the approximation will be good enough later on for all $x\in [a,b]$ (note that $a,b$ use $n$ bits). Putting everything together,  we  obtain an approximation to the partial fraction decomposition of $g(x)$, i.e., a function $\hat g(x)$ of the form $\sum^k_{i=1} \frac{\hat a_i}{\hat b_i x+\hat c_i}+\sum^k_{j=1} \frac{\hat a_i'x+\hat b_i'}{\hat c_i' x ^2+\hat d_i' x +\hat e_i'}$ which satisfies $|g(x)-\hat g(x)|\leq 1/2^{\Theta(n)}$ for all $x\in [a,b]$.  Now the integral $\int^{b}_{a}\hat g(x)\mathrm{d} x$ can be symbolically integrated, and then numerically evaluated with additive precision $1/2^{\Theta(n)}$, which gives us therefore $\hat{I}$ with $|\hat{I}-I|\leq 1/2^{(\Theta(n)}$. This completes the proof.
\end{proof}
\betainf*
\begin{proof}
Consider arbitrarily small $\epsilon>0$. We claim that for $F\sim \pi^{\ord}_{C}$, it holds that $\E[F]\geq |E_C|- \emm^{-n^\epsilon}$. Assuming this for the moment, by Markov's inequality we have that $\P\big(|E_C\backslash F|\geq 1\big)\leq  \emm^{-n^\epsilon}$. Therefore in the sum $Z^{\ord}_{C}=\sum_{F\in \Omega^{\ord}_C} (\emm^{\beta}-1)^{|F|}q^{c(F)}$  terms other than $F=E_C$ have $\emm^{-n^\epsilon}$ relative contribution, yielding the desired bound on $Z^{\ord}_{C}$.

To prove the bound on the expectation of $F$, it suffices to show that for any edge $e\in E_C$ it holds that   $\pi^\ord_{C}(1_e)\geq 1-\emm^{-n^\epsilon}$; equivalently, because $\left\|\pi^\ord_C-\pi_C\right\|_\TV=\emm^{-\Omega(n)}$ by Theorem~\ref{thm:phase-transition}, we only need to show that $\pi_C(1_e)\geq 1-\emm^{-n^\epsilon}.$ This is true for $\beta\geq n^{2\epsilon}$ since the probability that $e$ is an  in-edge is at least  $\frac{\emm^{\beta}}{\emm^{\beta}+q-1}\geq  1-\emm^{-n^{\epsilon}}$ for large enough $n$. 
\end{proof}
\section{Remaining proofs for Theorem~\ref{thm:mainthm}} \label{sec:putting-ingredients-together}

In this section, we finish the proof of Theorems~\ref{thm:disordered-mixing} and~\ref{thm:ordered-mixing}.

\subsection{Local mixing on treelike balls}\label{sec:tree-mixing}

In addition to WSM within the phase, we need another proof ingredient: rapid mixing of the Glauber dynamics on treelike balls with wired and free boundary condition.

For the free boundary condition, we can directly use the Lemma 6.7. from~\cite{rcm-on-unbounded-degree-graphs}:

\begin{lemma}[{\cite[Lemma 6.7]{rcm-on-unbounded-degree-graphs}}]
    Let \(B = (U,A)\) be a \(k\)-treelike graph. For every \(q,\beta > 0\), there exists \(\alpha_0(q,\beta,k)\) such that the Glauber dynamics \(X_t\) on \(B\) satisfies,
    \[
    \max_{X_0\in\Omega_{B}} || X_t - \pi_{B^+} ||_\TV \leq \tfrac{1}{\sqrt 2} \emm^{-\alpha_0 t} \log\left(\frac{1}{\min_{F\in\Omega_B} \pi_{B^-}(F)}\right)^{-1},
    \]
    where \(\pi_{B^-}\) is the Gibbs distribution on \(B\) with the free boundary condition (no additional wirings).
\end{lemma}

We use this Lemma for \(B = B_r(v)\). Note that for any configuration \(F\), \(\pi_{B_r^-(v)}(F)\geq \emm^{-O(|V_G(B_r(v))|)}\), and thus this implies \(O(N\log N)\) mixing for a free treelike ball with \(N\) vertices.

For the case of Glauber dynamics on a wired treelike ball, we use the result of Blanca and Gheissari \cite{BlaGhe} bounding the spectral gap for Glauber dynamics on a wired tree using its edge-cut-width. We note that while their lemma was stated for trees with bounded degree, that restriction was not needed to get the following bound.

\begin{lemma}[{\cite[Lemma 5.10]{BlaGhe}}]
    For any tree \(T\) with \(N\) vertices and edges \(E(T) = \{e_1,\dots,e_{N-1}\}\), define its \emph{edge-cut-width} to be
    \[
    \CW(T) = \min_{\sigma}\max_{1\leq i < N} |V(\{e_{\sigma(j)} : j\leq i\}\cap V(\{e_{\sigma(j)} : j > i\}))|,
    \]
    where the maximum is over all permutations on \(\{1,\dots,N-1\}\).

    Then for any reals \(q\geq 2\) and \(\beta > 0\), there is a constant \(\theta(q,\beta) > 0\) such that the inverse spectral gap of the Glauber dynamics on \(T\) with a \emph{wired boundary} (all leaves of \(T\) belonging to the same component) is \(\theta N^2 \exp\{2(\CW(T)+1)\log q\}\).
\end{lemma}

We now observe that the edge-cut-width of a tree is at most its height.

\begin{lemma}
    Let \(T\) be a height-\(r\) tree. Then 
    \(\CW(T)\leq r\).
\end{lemma}
\begin{proof}
    Consider a depth first search traversal of \(T\) starting from its root. Let \(\sigma\) be the order in which the edges are explored. Note that with DFS traversal on a tree, if at any point, a vertex has at least one but not all incident edges explored, then the traversal is then in the subtree of this vertex. Thus, after exploring \(1\leq i\leq |E(T)|\) edges, the vertices that have at least one, but not all incident edges traversed for a path to a root. Hence only those vertices are in \(V(\{e_{\sigma(j)} : j\leq i\}\cap V(\{e_{\sigma(j)} : j > i\}))\).

    Since the longest path from a root is \(r\), hence \(\CW(T)\leq r\).
\end{proof}

To obtain the mixing time for a treelike ball from the mixing on a tree, we use that adding an edge to the graph changes the probabilities by at most a constant factor, and therefore the spectral gap is only changed by a constant factor.

\begin{lemma}\label{lem:wired-mixing}
    For any real \(A>0\), for all \(d\) large enough, and then for all \(q\) large enough reals, w.h.p. over \(G\sim\G(n,d/n)\) the following holds. For all integer \(r\leq \frac{A\log n}{d}\) and all \(v\in V_C\), let \(B := B_r(v)\).
    
    Then the mixing time of the Glauber dynamics on \(B\) with wired boundary condition is $O(n^{4\frac{A}{d} \log q})$.
\end{lemma}
\begin{proof}
    By Lemma~\ref{lem:G-1-treelike}, for all \(d\) large enough, w.h.p. any such \(B\) is \(1\)-treelike. By Lemma~\ref{lem:ball-size}, w.h.p. any such \(B\) has size \(\leq d^{r}\log n \leq n^{A\log d / d}\log n\).
    
    Let \(T\) be a tree obtained from \(B\) by removing at most one edge (zero if \(B\) is already a tree). By the previous two lemmas, the inverse spectral gap for Glauber on \(T\) with a wired boundary condition is \(\lambda_T^{-1} =O(|V_C(B)|^2 \emm^{2r\log q}) = O(n^{2A(\log d + \log q)/d} \log^2 n)\).

    Then using standard techniques (for details see e.g.  \cite[Lemma~6.10]{BlaGhe}), the inverse spectral gap \(\lambda^{-1}_B\) of the Glauber dynamics on \(B\) with a wired boundary condition is at most \(C(q,\beta) \lambda^{-1}_T\) (for some constant \(C(q,\beta)\) independent of \(n\)).

    We then have an additional factor of \(O(n^{A\log d/d}\log n)\), since the  that the mixing time is at most \(\lambda_B^{-1}\log n\log({\frac{4}{\min_F \pi_{B^+}(F)}})\), and that for any configuration \(F\) on \(B\), \(\log \pi_{B^+}(F) = \Omega(|V(B)|^{-1})\). Then we just note that \(O(n^{3A\log d/d + 2\log q/d}\log^3 n)\) is at most \(O(n^{4A\log q/d})\) for all \(q\) large enough.
\end{proof}

\subsection{Proof of Theorems~\ref{thm:disordered-mixing} and~\ref{thm:ordered-mixing}}\label{sec:proof-sketch}

Here we outline the proof of Theorem~\ref{thm:ordered-mixing}, the proof of Theorem~\ref{thm:disordered-mixing} is analogous. The proofs follow closely those in \cite{galanis2024plantingmcmcsamplingpottsarxiv,SinclairsGheissari2022,Galanis_Goldberg_Smolarova_2024}, thus we only outline the key ideas.

Before we start, let \(A\) be such that for all \(d\) and \(q\) large enough reals, w.h.p. \(G\) has WSM within the ordered phase at the radius~\(r=\lceil\frac{A\log n}{d}\rceil\) with rate \(q^{-1/30}\) (Lemma~\ref{lem:ordered-WSM}).

Let \(X_t\) be Glauber dynamics initialized from all-in. We focus on proving that for \(T = \text{poly}(n)\) steps, \(|X_T-\pi^\ord_C|_\TV\leq \tfrac 14\), as then any accuracy \(\epsilon\geq \emm^{-\Omega(n)}\) can be achieved by probability amplification (for details see \cite[Theorem~3.4]{galanis2024plantingmcmcsamplingpottsarxiv}).

To bound the total variation distance, we couple \(X_t\) with another copy of the Glauber dynamics \(\hat X_T\), which starts from \(\pi_C^\ord\), and is restricted to the ordered phase -- if it was to make an update outside the phase, it ignores the update instead. We couple the chains by having the chose the same edge to update, and both make the updated based on the same \(\text{Uniform}(0,1)\) random variable. One can easily verify that \(\pi_C^\ord\) is the stationary distribution for the restricted Glauber dynamics, and thus for all \(t\geq 0\), \(\hat X_t\sim\pi_C^\ord\).

To bound \(||X_T-\pi^\ord_C||_\TV\leq \P(X_T\neq \hat X_t)\), we bound the probability of disagreement on each edge separately, and then use union bound. Thus fix an edge \(e\) and a vertex \(v\) incident to \(e\).

We introduce another coupled copy (using the same coupling as \(X_t\) and \(\hat X_t\)) of \textit{wired} Glauber dynamics \(X_t^v\) starting from all-in restricted to the radius-\(r = \lceil\frac{A\log n}{d}\rceil\) ball around \(v\), \(B_r(v)\), i.e. \(X_t^v\) acts as if every vertex distance-\(r\) from \(v\) is in a single component, and ignores every update that would change an edge outside this ball. Note that the stationary distribution of this chain is exactly \(\pi_{B_r^+(v)}\).

Using standard monotonicity arguments (see e.g. \cite[Appendix B]{Galanis_Goldberg_Smolarova_2024}), we can see that conditional on \(\hat X_t\) having not ignored any update by time \(t\) (call this event \(\Ec_{\leq t}\)), \(\hat X_t \subseteq X_t\subseteq X_t^v\). This gives us, that conditional on \(\Ec_{\leq t}\), whenever \(X_t\) and \('\hat X_t\) disagree on \(e\), then also \(\hat X_t\) and \(X_t^v\) disagree on \(e\). With some arithmetic (for details see e.g. the proof of Theorem 1 in \cite{Galanis_Goldberg_Smolarova_2024}), we obtain that 
\[
\P(\1\{e\in X_T\}\neq \1\{e\in X_T\})\leq 5\P(\Ec_{\leq T}) + ||\pi^\ord_C(1_e) - \pi_{B_r^+(v)}||_\TV + ||X_t^v-\pi_{B_r^+(v)}||_\TV.
\]
The first term can be bound using the Theorem~\ref{thm:phase-transition}(ii) -- noting that the probability that an update would be ignored is at most the probability that the current state has \(\geq\frac{9\eta}{10}|E_C|\) out-edges, which is w.h.p. \(\leq \emm^{-n}\) (for all \(d\) and \(q\) large enough). Thus the first term is \(\emm^{-\Omega(n)}\). The second term is exactly the WSM within the ordered phase, and provided that \(q\) is sufficiently large, it is at most \(\tfrac{1}{n^2}\). 

The second term is the TV-distance between the wired Glauber on the ball after \(T\) steps and its stationary distribution. By Lemma~\ref{lem:wired-mixing}, after \(T_v = O(n^{5A\log q/d})\) updates within the ball \(B_r(v)\), \(||X_t^v-\pi_{B_r^+(v)}||_\TV\leq \tfrac{1}{n^{3}}\) (note that \(\log (n^3) = O(n^{A\log q/d})\)). Using standard Chernoff bounds, if we take \(T \geq 30 T_v \frac{|E_C|}{|E_G(B_r(v))|}\), the probability that less than \(T_v\) updates were made in \(B_r(v)\) within \(T\) steps is \(\leq \emm^{-\Omega(\log^3 n)} = O(1/n^3)\).

Summing up, this gives that for all \(n\) large enough, \(\P(\1\{e\in X_T\}\neq \1\{e\in X_T\})\leq \frac{1}{4|E_C|}\), and thus the theorem follows from union bound.



\end{document}